\documentclass[11pt, letterpaper]{amsart}

\usepackage{amsfonts}
\usepackage{amsmath}
\usepackage{amsthm} 
\usepackage{color}
\usepackage{enumerate}
\usepackage{amssymb}
\usepackage{graphicx}
\graphicspath{ {Figures/} }
\usepackage{caption}
\usepackage{subcaption}
\usepackage{appendix}
\usepackage{ulem}
\usepackage{cancel}
\usepackage[]{mathtools}
\usepackage[]{bbm} 

\usepackage{caption}
\usepackage{subcaption}

\theoremstyle{definition}
\newtheorem{theorem}{Theorem}[section]
\newtheorem{assumption}[theorem]{Assumption}

\newtheorem{proposition}[theorem]{Proposition}

\newtheorem{lemma}[theorem]{Lemma}
\newtheorem{remark}[theorem]{Remark}

\numberwithin{equation}{section}

\newcommand{\E}{\mathbb{E}}
\newcommand{\filt}{\mathbb{F}}
\newcommand{\filtg}{\mathbb{G}}

\newcommand{\prob}{\mathbb{P}}
\newcommand{\tprob}{\widetilde{\mathbb{P}}}
\newcommand{\qprob}{\mathbb{Q}}
\newcommand{\reals}{\mathbb R}

\newcommand{\sdpos}{\mathbb{S}^d_{++}}
\newcommand{\A}{\mathcal{A}}

\newcommand{\C}{\mathcal{C}}
\newcommand{\Ecal}{\mathcal{E}}
\newcommand{\F}{\mathcal{F}}
\newcommand{\G}{\mathcal{G}}
\newcommand{\Hcal}{\mathcal{H}}

\newcommand{\M}{\mathcal{M}}
\newcommand{\tM}{\widetilde{\mathcal{M}}}
\newcommand{\OO}{\mathcal{O}}
\newcommand{\mcp}{\mathcal{P}}
\newcommand{\R}{\mathcal{R}}

\newcommand{\We}{\mathcal{W}}

\newcommand{\ca}{\check{a}}

\newcommand{\cvas}{\xrightarrow{a.s.}}
\newcommand{\cvprob}{\xrightarrow{p}}
\newcommand{\eps}{\varepsilon}
\newcommand{\such}{\ | \ }

\newcommand{\probtripleg}{(\Omega, \mathcal{G}, \mathbb{P})}

\newcommand{\dfn}{\, := \,}

\newcommand{\tr}{\mathrm{Tr}}

\newcommand*{\dif}{{\mathop{}\!\mathrm{d}}}

\newcommand{\e}{\mathrm{e}}
\newcommand{\Afn}{\mathbf{K}_1}
\newcommand{\Bfn}{\mathbf{K}_2}
\newcommand{\Cfn}{\mathbf{K}_3}
\newcommand{\Dfn}{\mathbf{K}_4}
\newcommand{\Efn}{\mathbf{K}_5}
\newcommand{\Efnalt}{v_c}

\newcommand{\bAfn}{\ol{\mathbf{K}}_1}
\newcommand{\bBfn}{\ol{\mathbf{K}}_2}
\newcommand{\bEfn}{\ol{\mathbf{K}}_5}
\newcommand{\tprobnu}{\widetilde{\nu}}

\newcommand{\expv}[3]{\mathbb{E}^{#1}_{#2}\left[#3\right]}
\newcommand{\wtexpv}[3]{\widetilde{\mathbb{E}}^{#1}_{#2}\left[#3\right]}
\newcommand{\condexpv}[4]{\mathbb{E}^{#1}_{#2}\left[\left. #3 \right| #4\right]}
\newcommand{\wtcondexpv}[4]{\widetilde{\mathbb{E}}^{#1}_{#2}\left[\left.#3 \right| #4\right]}
\newcommand{\expvs}[1]{\mathbb{E}\left[#1\right]}
\newcommand{\condexpvs}[2]{\mathbb{E}\left[ \left. #1 \right| #2\right]}

\newcommand{\condprobs}[2]{\mathbb{P}\left[#1\big| #2\right]}
\newcommand{\wtcondprobs}[2]{\widetilde{\mathbb{P}}\left[#1\big| #2\right]}
\newcommand{\xpn}[1]{\exp\left(#1\right)}

\newcommand{\nada}[1]{}

\newcommand{\wt}[1]{\widetilde{#1}}
\newcommand{\wh}[1]{\widehat{#1}}

\newcommand{\bra}[1]{\left[#1\right]}
\newcommand{\cbra}[1]{\left\{#1\right\}}
\newcommand{\ol}[1]{\overline{#1}}
\newcommand{\ul}[1]{\underline{#1}}
\newcommand{\relent}[2]{H\left(#1\big|#2\right)}

\newcommand{\mbf}[1]{\mathbf{#1}}

\newcommand{\abs}[1]{\left\lvert#1\right\rvert}

\newcommand{\spos}[1]{\mathbb{S}^{#1}_{++}}
\newcommand{\plog}[1]{\textrm{PL}\left(#1\right)}
\newcommand{\plogf}[2]{\textrm{PL}\left(#1,#2\right)}

\newcommand{\idmat}[1]{\mathbf{1}_{#1}}

\begin{document}

\title[CDS Optimal Investment]{Optimal Investment in Equity and Credit Default Swaps in the Presence of Default}

\author{Zhe Fei}
\author{Scott Robertson}
\address{Questrom School of Business, Boston University, Boston MA, 02215, USA}
\email{zhefei@bu.edu; scottrob@bu.edu}

\date{\today}

\begin{abstract}

We consider an equity market subject to risk from both unhedgeable shocks and default.  The novelty of our work is that to partially offset default risk, investors may dynamically trade in a credit default swap (CDS) market. Assuming investment opportunities are driven by functions of an underlying diffusive factor process, we identify the certainty equivalent for a constant absolute risk aversion investor with a semi-linear partial differential equation (PDE) which has quadratic growth in both the function and gradient coefficients. For general model specifications, we prove existence of a solution to the PDE which is also the certainty equivalent. We show the optimal policy in the CDS market covers not only equity losses upon default (as one would expect), but also losses due to restricted future trading opportunities.  We use our results to price default dependent claims though the principal of utility indifference, and we show that provided the underlying equity market is complete absent the possibility of default, the equity-CDS market is complete accounting for default.  Lastly, through a numerical application, we show the optimal CDS policies are essentially static (and hence easily implementable) and that investing in CDS dramatically increases investor indirect utility. 
    
\end{abstract}

\maketitle

\section{Introduction}

In this article, we consider an optimal investment problem  with random endowment, partially hedgeable shocks, and the possibility of default in one or more of the traded assets. The novelty of our work is that to partially offset the default risk, the investor trades dynamically in a market for credit default swaps (CDS) on the traded assets.   Our goal is to identify how the investor uses the CDS market to mitigate her default risk, and how existence of this market alters both her indirect utility from trading and the way she prices default dependent contingent claims.

Continuing the line of research studied in \cite{bielecki2006portfolio, MR2779555, MR2178034, MR2212266, sircar2007utility, MR3846288} and especially \cite{MR4086602}, we work in a reduced form intensity based model where investment opportunities (such as excess returns and  volatility; default intensities, losses and recovery rates; and random endowments/contingent claim payoffs) are driven by an underlying economic factor process $X$, modeled as a multi-dimensional diffusion.  Furthermore, the shocks driving $X$ are only partially correlated with those driving the equities, and hence even absent default, the market is incomplete.

New to our model, especially in comparison to \cite{MR4086602} and \cite{bielecki2006portfolio, MR2779555, capponi2014dynamic, MR3846288}, is that we allow the investor to dynamically trade in a rolling (c.f. \cite{MR2474544}) or ``on the run'' CDS market offering protection upon equity default. To obtain the wealth process associated to a dynamic CDS strategy, we depart from the
 current literature (see in particular \cite{capponi2014dynamic, dabadghao2014dynamic}) by using the rolling CDS strategies of \cite{MR2474544}.  These strategies arise as the investor enters and unwinds positions in CDS of a fixed remaining time to maturity with ever increasing frequency, and enables continual investment in on-the-run CDS contracts, avoiding (as the authors in \cite{MR2474544} discuss) liquidity issues associated to off-the-run CDS markets.  

Even absent default, and certainly accounting for default, the above market may be incomplete. This prohibits pricing claims solely using through absence of arbitrage arguments. Therefore, and also to account for her preferences, the investor prices  through the principal of utility indifference (see \cite{HN1989}).  This requires us to identify the investor's value function in the presence of random endowments with default dependent payoffs, such as defaultable bonds. To ensure tractability, we assume her preferences are described by an exponential, or constant absolute risk aversion (CARA), utility function.  This implies the indifference price for any contingent claim is independent of the investor's initial wealth, and that up to translation by the initial wealth, the indirect utility function depends only on time and the factor process.

In fact, due to the diffusive Markovian structure, the indirect utility function is expected to satisfy a certain semi-linear partial differential equation (PDE): see \eqref{E:G_HJB} below.  Due to the CDS market, this PDE differs from that in \cite{MR4086602}, as the instantaneous covariance matrix of the combined equity and CDS wealth processes may degenerate on the interior of the state space. On a technical level, this requires non-trivial extensions of the PDE results in \cite{MR4086602}, but we are still able to verify the certainty equivalent solves the PDE \eqref{E:G_HJB} by appealing to the classical theory of semi-linear equations in \cite{MR0181836, MR1465184}, and using both duality (see, e.g \cite{MR1865021, MR2489605}) and delicate localization (c.f. \cite{MR4086602}) arguments.

Qualitatively, our main finding is that the investor does not hold a position in the CDS solely to offset losses in the equity.  Rather, the investor holds a position to satisfy the heuristic relationship (see Sections \ref{SS:detwtg} and \ref{SS:opt_policy})
\begin{equation}\label{E:qual}
    \begin{split}
    \textrm{CDS Dollar Position}  &= \textrm{Equity Loss }  + \textrm{ Loss due to Stoppage of Trade.}
    \end{split}
\end{equation}
Above, the second term is the monetary value ``lost'' by the investor because after default, she cannot trade in the defaulted securities.  This decomposition is intuitively clear, as the investor is aware that default means more than just a loss in the equity position: it means she cannot trade after default as well.

Second, we find that if the equity market absent default is complete, then the equity-CDS market including default is complete, provided a certain (very mild) non-degeneracy condition holds: see equation \eqref{E:E_barA_def} and Assumption \ref{A:complete_mkt} below.   While on the one hand it is clear that by adding a tradeable asset one may hedge against an additional source of uncertainty, on the other hand, the non-degeneracy condition was a surprise (at least a-prori). However, as we explain in Section \ref{S:complete_mkt}, this condition is necessary to rule out the CDS being a redundant asset (compared to the equity), and can be verified using the non-degeneracy results of \cite{MR3131287, MR3590708}. 

Continuing, we show (for general model specifications) the investor hedges against default primarily through her CDS position, and not through the equity position.  This is seen numerically in Section \ref{S:numerics}, and shows the CDS market is doing what it should: providing a mechanism to hedge against default in deteriorating market conditions.  In particular, we show the investor does not hedge default risk by shorting the stock, which would be difficult to practically implement.  This stands in direct contrast to when the CDS market is not present, as therein it was shown in \cite[Section 4]{MR4086602}) the investor does short the equity.

We now describe the model. The time horizon is $[t,T]$ and we set the interest rate to $0$. Absent default, the equity process $S^e$ has dynamics driven by a diffusive factor process $X$, with shocks to $S^e$ only partially correlated with those to $X$,\footnote{This is is line with the models encountered in \cite{kim1996dnp, wachter2002pac, MR2048829, MR2178034, MR2206349, liu2007pss, buraschi2010correlation, MR2932547, MR4086602} among many others.} see \eqref{E:SDE} and \eqref{E:SE_dynamics}. The default (or ''credit-event'') time is $\tau$, at which time the equities experience a proportional loss governed by a loss function $\ell_e$.   The default time has $\filt^{X}$ intensity function $\gamma$, and in addition to viewing $X$ and $S^e$, investors also observe the default indicator process $H = \cbra{H_t = 1_{t\geq \tau}}$. The CDS price process $S^r$ dynamics are from  \cite{MR2474544} and require comment. Therein, the dynamics were obtained under an exogenously specified ``spot pricing measure'' $\tprob$,  and the rolling CDS contract has horizon $\wt{T}$, which we assume is larger than the investor's horizon $T$.\footnote{This corresponds to market being in existence throughout investor's time period.} We connect pull the dynamics back to $\prob$ by specifying the $\tprob$ default intensity $\wt{\gamma}$ and $\filt^X$ risk premia $\wt{\nu}$.\footnote{The idea of exogenously specifying CDS price dynamics under a pricing measure and then pulling back to the physical measure is also used in, for example, \cite{capponi2014dynamic}.} To ensure consistency with \cite{MR2474544}, we require $\tprob$ to be a martingale measure for $S^e$, but as our market is generically incomplete, $\tprob$ is simply one of the martingale measures.  We then combine $S^e,S^r$ into a single price process $S$.

The agent has CARA preferences (with risk aversion $\alpha$) from terminal wealth, and a random endowment of the form $\phi(X_T)1_{\tau > T} + \psi(\tau,X_{\tau})1_{\tau \leq T}$. $\phi$ is a default dependent claim, and while we allow for general payoff functions $\phi$, we are primarily interested in $\phi \equiv q$ for some $q\in\reals$, as this corresponds to $q$ face of  defaultable bond.  $\psi$ is a ``payoff'' the investor receives upon default, which allows us to account for both partial recovery in the defaultable bond and the investor's indirect utility had she traded over the period $[\tau,1]$ in the remaining non-defaulted assets.   In this setting, standard heuristic arguments indicate the value function at $t<T, X_t = x$, and given investor wealth $w$, is of the form $u(t,x,w) = -\xpn{-\alpha (w + G(t,x))}$, where the indirect utility function $G$ satisfies the Hamilton-Jacoby-Bellman (HJB) equation \eqref{E:G_HJB} with Hamiltonian $H$ from \eqref{E:hamil}.

To solve the HJB equation we assume (see Assumptions \ref{A:complete_mkt} and \ref{A:incomplete_mkt} respectively) one of two scenarios. First, that both the equity market absent default and the equity-CDS market allowing for default, are complete with martingale measure $\tprob$. Interestingly, equity-CDS market completeness requires non-degeneracy of the function $\Efnalt$ of \eqref{E:E_barA_def}, which itself ensures the CDS is not a redundant asset, compared to $S^e$.  As the equity-CDS market is complete, the PDE for $G$ linearizes and, under a very mild no arbitrage condition (see Assumptions \ref{A:phi_psi_alt}, \ref{A:complete_mkt}), we obtain in Theorem \ref{T:main_result_complete} a smooth solution to the HJB equation which is also the certainty equivalent function.

In the second scenario, the market absent default is ``strictly'' incomplete (in that the shocks affecting $S^e,X$ have correlation matrix which lies below $(1-\eps)\idmat{d}$ uniformly for some $\eps > 0$).  Here, the PDE for $G$ does not linearize, and we cannot directly invoke the results of \cite{MR4086602} to obtain existence of solutions due to potential degeneracy of the CDS volatility function $\sigma_r$ of \eqref{E:sigR_def}\footnote{For example, \eqref{E:wtu_wtv_dfn} implies that when the default intensity under the spot pricing measure is deterministic, $\abs{\sigma_r}$ is identically $0$.}.  However, provided either (i) $\sigma_r$ is identically degenerate or (ii) $\sigma_r$ is never degenerate (see Assumption \ref{A:incomplete_mkt} for a precise statement) we suitably modify the proofs in \cite{MR4086602} to verify in Theorem \ref{T:main_result} existence of a solution to the PDE which is also the value function. 

Having presented the main existence and verification result, we discuss the optimal policies in Sections \ref{SS:detwtg} and \ref{SS:opt_policy} respectively, making precise the heuristic in \eqref{E:qual}.  In Section \ref{S:indiff_pricing} we recall the concept of the a utility indifference price, and connect the price to the certainty equivalent.

In Section \ref{S:numerics} we perform a numerical application when the underlying factor process is CIR.  Here, in the complete market setting of Assumption \ref{A:complete_mkt} we display three very interesting results.  First, when the CDS market is present, rather than shorting the stock (as occurs absent the CDS market, see \cite{MR4086602}), the investor holds a stable equity position across a range of default intensities: see Figure \ref{fig:Optimal Defaultable Equity Positions}.  This implies the investor is using the CDS market as intended, to mitigate default risk.  Second, in Figure \ref{fig:Optimal CDS Positions} we show that the CDS position displays very little variation over both time and the state variable.   Indeed, the right-plot therein dramatically shows how stable the position is by plotting the minimal and maximal CDS positions (over the state space) as a function of time. This shows the CDS positions are implementable in practice (where there might not be a liquid market for dynamic CDS trading). Lastly, in Figure \ref{fig:RB} we plot the relative benefit of the CDS market (defined by the ratio $\textrm{CE}^{\textrm{CDS}}/\textrm{CE}^{\textrm{No CDS}} -1$) where ``No CDS'' means absent the CDS market. Especially for high default intensities and large positions in the defaultable bond, the relative benefit is quite large.

Section \ref{S:numerics} also considers an incomplete market example, where there are two equities, one of which can default. Despite market incompleteness, Figures \ref{fig:stochastic_UIP} and \ref{fig:stochastic_CDS_Positions} indicate the investor's ability to accurately hedge the defaultable bond.  Indeed, Figure \ref{fig:stochastic_UIP} shows the time zero indifference price is almost independent of the notional, and Figure \ref{fig:stochastic_CDS_Positions} shows the time zero CDS position grows in almost one-to-one correspondence with the notional.  Interestingly, in Figure \ref{fig:stochastic_CDS_Positions} we see the position in the defaultable equity varies very little with either the notional or the state variable, further indicating that hedging is being done in the CDS market.

This paper is organized as follows.  The model and optimal investment problem are  presented in Sections \ref{S:prob} and \ref{S:opt_invest}. The HJB equation for the certainty equivalent is identified in Section \ref{S:HJB_CE}. Section \ref{S:complete_mkt} presents results in the complete market setting, and Section \ref{S:incomplete_mkt} presents results in the incomplete market setting. Section \ref{S:indiff_pricing} discusses indifference pricing and Section \ref{S:numerics} contains the numerical example.  A conclusion follows in Section \ref{S:conclusion}. Proofs are contained in Appendices \ref{AS:HJB} -- \ref{AS:PropE9Lem}.

\section{Probabilistic setup and factor process}\label{S:prob}

There is a complete probability space $\probtripleg$ that supports two independent Brownian motions, $W$ and $B$, of respective dimensions $d$ and $k$, as well as an independent random variable $U\sim U(0,1)$. We denote by $\filt^{W}$ and $\filt^{W,B}$ the $\prob$ augmented natural filtrations of $W$ and $(W,B)$ respectively.  There is a process $X$ that represents the dynamic evolution of factors fundamental to the economy.  $X$ is driven by the Brownian motion $W$ and has dynamics
\begin{equation}\label{E:SDE}
    dX_s = b(s,X_s) ds + a(s,X_s) dW_s.
\end{equation}
$X$ takes values in a region $\OO\subset\reals^d$ and we assume
\begin{assumption}\label{A:region}\text{}
    \begin{enumerate}[(i)]
    \item $\OO \subset \mathbb{R}^d$ is an open, and there exists a sequence of open, bounded, connected sub-regions $\left\{ \OO_n \right\}$ with $\ol{\OO}_n \subset \OO_{n + 1},$ and $\OO = \bigcup_n \OO_n.$ For each $\OO_n,$ $\partial \OO \in C^{2, \beta}$.\footnote{See \cite[Section 3.2]{MR1326606} for a precise definition of $\partial \OO \in C^{2,\beta}$, and throughout $\beta \in (0,1]$ is fixed constant. The primary examples are $\OO = \reals^d$ with $\OO_n$ the ball of radius $n$, and $\OO = (0,\infty)$ with $\OO_n= (1/n,n)$.}
    \item $b \in C^{(1, 1)}([0, \infty) \times \OO; \reals^d)$ and $A \in C^{(1, 1)}([0, \infty) \times \OO; \sdpos)$\footnote{$\mathbb{S}^d_{+}$ is the set of symmetric non-negative definite $d\times d$ matrices, and $\sdpos$ is the strictly positive definite subset.}, and  for each fixed $(t,x)$, $a(t,x) = \sqrt{A(t,x)}$, the unique symmetric positive definite square root.   For any starting point $x\in\OO$ and time $t \geq 0$ there is a unique strong solution to \eqref{E:SDE} starting at $t$ with $X_t = x$, which we will write $X^{t,x}$.
    \end{enumerate}
\end{assumption}

\begin{remark} It is well known (see \cite[Chapter 5]{MR1121940}) that a unique strong solution taking values in $\OO=\reals^d$ will exist if $a,b$ are globally Lipshitz in space, locally uniformly in time, and of linear growth.  Additionally, in the time-homogeneous case, from \cite[Chapter IX]{MR1725357} the strong solution property will hold provided the process $X$ does not explode to the boundary of $\OO$ in finite time, and in the univariate setting there are necessary and sufficient conditions (see \cite[Theorem 5.1.5]{MR1326606}) for explosion to occur. 
\end{remark}

\section{The Optimal Investment Problem}\label{S:opt_invest}

Fix a starting time $t\geq 0$ and location $x\in\OO$. There are three assets available for investment: a money market, an equity market, and a rolling CDS market.  This latter market is relevant as there is a default time $\tau$ (or more appropriately named a credit-event time) which affects investment opportunities. We assume $\tau$ has $\filt^{W,B}$ intensity governed by a function $\gamma$ of time and state. To enforce this, and following the canonical reduced form construction (see \cite{bielecki2013credit}) we set\footnote{As $(t,x)$ are fixed, we omit their dependence so that, for example, $\tau^{t,x}$ is written $\tau$ and $X^{t,x}$ is written $X$.}
\begin{equation*}
    \tau= \inf \cbra{ s \geq t \such \int_t^s \gamma(u, X_u) du = -\log(U)}.
\end{equation*}
This implies for $s>t$
\begin{equation*}
    \condprobs{\tau > s}{\F^{W,B}_{\infty}} = \condprobs{\tau > s }{\F^{W,B}_s} = \xpn{-\int_t^s \gamma(u, X_u) du},
\end{equation*}
and hence $s \to \gamma(s,X_s), s\geq t$ is the $\filt^{W,B}$ default intensity for $\tau$ under $\prob$.  Given $\tau$ we define the default indicator process $H_\cdot = \mathbbm{1}_{\tau \leq \cdot}$, as well as the filtration $\filtg$ which is the $\prob$-augmentation of $W,B$ and $H$'s natural filtration.  Lastly, we remark (see \cite{bielecki2009credit}) that  $W,B$ remain Brownian motions in the $\filtg$ filtration and that $\filtg$ satisfies the usual conditions.

We next describe the markets. First, the money market pays a constant interest rate, which we set to $0$.  Second, the equity market has $k$ risky assets $S^e$ with dynamics on $[t,\infty)$
\begin{equation}\label{E:SE_dynamics}
    \frac{dS^e_{s}}{S^e_{s-}}  = \mathbbm{1}_{s \leq \tau} \left( \mu_e(s, X_s) ds + \sigma_e(s, X_s) dZ_s \right) - \ell_e(s, X_s) dH_s;\hspace{20pt} \langle Z, W \rangle_s  = \rho(s, X_s) ds.
\end{equation}
Here, $Z$ is the $k$-dimensional Brownian motion
\begin{equation}\label{E:Z_BM}
    Z_s = \int_t^s \rho(u,X_u)dW_u + \int_t^s \ol{\rho}(u,X_u)dB_u;\qquad s \geq t.
\end{equation}
$\rho$ is a $k\times d$ matrix valued function satisfying $\idmat{k} - \rho\rho' \in \mathbb{S}^k_{+}$, and $\ol{\rho} = \sqrt{\idmat{k} - \rho\rho'}$\footnote{Throughout, $'$ denotes transposition, $\idmat{p}$ is the $p$ dimensional identity matrix and $0_p$ is the $p$ dimensional $0$ vector.}.  The function $\mu_e$ represents the pre-default drift, and $\sigma_e$ the pre-default volatility.  Prior to default, $S^e$ is an It\^{o} process with dynamics governed by $X$, and whose shocks are partially correlated with those driving $X$.   Upon default, $S^e$ experiences a downward jump, the size of which is determined by the fractional loss vector-valued function $\ell_e$. Our assumptions on the default intensity and model coefficients are

\begin{assumption}\label{A:SE_reg}  As functions defined on $[0,\infty)\times \OO$ and taking values in the respective state spaces $(0,\infty)$, $\reals^k$, $\spos{k}$, $\reals^{k\times d}$ and $\reals^k$ we have $\gamma, \mu_e, \sigma_e, \rho$ and $\ell_e$ are all in $C^{(1,1)}$.  Additionally, the correlation function $\rho$ satisfies  $\idmat{k} - \rho \rho' \in \mathbb{S}^k_{+}$ and the loss function $\ell_e$ satisfies both $\ell_e'\ell_e > 0$ and $0 \leq \ell_e^{(i)}\leq 1, i=1,\dots, k$.

\end{assumption}

\begin{remark} $\ell_e'\ell_e > 0$ ensures at least one of the equities has a fractional loss upon default. Additionally, \cite[Lemma 1.7.3]{MR1326606} implies $\bar{\rho} \in C^{(1,1)}([0,\infty)\times \OO; \mathbb{S}^k_{+})$, inheriting $\rho$'s regularity.
\end{remark}

It is well known the process
\begin{equation*}
    M_s \dfn H_s - \int_t^{s\wedge \tau} \gamma(u,X_u)du;\qquad s\geq t,
\end{equation*}
is a  $\filtg$  local martingale, and hence we obtain the $\filtg$ semi-martingale decomposition for $S^e$
\begin{equation}\label{E:Se_smgle}
    \frac{dS^e_{s}}{S^e_{s-}} =  \mathbbm{1}_{s \leq \tau} \left(\left(\mu_e -\gamma\ell_e\right)(s, X_s) ds + \left(\sigma_e\rho\right)(s, X_s) dW_s + \left(\sigma_e\ol{\rho}\right)(s,X_s)dB_s\right)  - \ell_e(s, X_s)  dM_s.
\end{equation}

Lastly, we turn to the CDS market. Here, we use \cite{MR2474544} to associate a price process $S^r$ ($r$ stands for ``rolling'') to dynamic trading in a rolling CDS contract which fully indemnifies investors from losses in the event of default  over the time period $[t,\wt{T}]$ where $\wt{T}$ is CDS contract maturity\footnote{We write $\wt{T}$ for the CDS maturity as it need not coincide with the agent's investment horizon $T$ defined below, but we require $T <  \wt{T}$. We separate the times because there is no a-priori reason to believe the CDS market (for which $\wt{T} = 2, 5$ years are typical terms) has the same horizon as the individual investor.}.  In the rolling strategy,  at each time $s \in [t,\wt{T}]$ one enters a CDS contract and then unwinds at a short time $ds$ later.  The protection premium is paid continuously over $[s,s+ds]$, and should there be a default during this time, the holder is fully protected in her position. 

To obtain the wealth process dynamics for investment in the rolling CDS, \cite{MR2474544} assumes the market values securities using an exogenously given spot pricing measure $\tprob$. For the sake of consistency, $\tprob$ should be an equivalent local martingale measure for $S^e$, and this places restrictions on $\tprob$.  To see the restrictions, assume the investor has horizon $T < \wt{T}$. Using \eqref{E:Se_smgle}, one can show if a measure $\qprob$ is equivalent to $\prob$ on $\G_T$ with density
\begin{equation}\label{E:mm_den}
    \frac{d\qprob}{d\prob}\Big|_{\G_T} = \Ecal\left(\int \mbf{A}_u'dW_u + \mbf{B}_u'dB_u + \mbf{C}_u dM_u\right)_T,\footnote{$\Ecal(\cdot)$ is the Doleans-Dade stochastic exponential and $\mbf{A},\mbf{B},\mbf{C}$ are $\filtg$ predictable processes with $\mbf{C}>-1$.}
\end{equation}
then $\qprob$ is an equivalent local martingale measure for $S^e$ if and only if on $[t,\tau\wedge T]$
\begin{equation}\label{E:mpr_1}
    0 = \mu_e - \gamma\ell_e + \sigma_e\rho\mbf{A} + \sigma_e\bar{\rho}\mbf{B} - \gamma \ell_e \mbf{C}.
\end{equation}
First, consider when the equity market absent is default is complete, in that the number of factors $d$ equals the number of assets $k$ and $\rho \equiv \idmat{d}$. Here, $\tprob$ is identified by setting
\begin{equation}\label{E:tprob_mpr_c}
    \wt{\mbf{A}} = -\sigma_e^{-1}\left(\mu_e - \wt{\gamma}\ell_e\right);\qquad \wt{\mbf{B}} = 0;\qquad \wt{\mbf{C}} = \frac{\wt{\gamma}}{\gamma} - 1,
\end{equation}
where $\wt{\gamma}$, which we may exogenously specify, is the $\filt^{W,B}$ default intensity function under $\tprob$. In the general case, to enforce \eqref{E:mpr_1} we may exogenously choose both the $\tprob$ default intensity function $\wt{\gamma}$ and $W$ equity risk premia function $\tprobnu$ and then set
\begin{equation}\label{E:tprob_mpr_ic}
    \wt{\mbf{A}} = -\tprobnu;\qquad \wt{\mbf{B}} = -(\ol{\rho})^{-1}\left(\sigma_e^{-1}(\mu_e - \wt{\gamma}\ell_e) - \rho \tprobnu\right);\qquad \wt{\mbf{C}} = \frac{\wt{\gamma}}{\gamma} - 1.
\end{equation}
With this as motivation, we assume the following about $\tprob$.

\begin{assumption}\label{A:alt_gamma} $\tprob$ is equivalent to $\prob$ on $\G_S$ for each $S>0$, and under $\tprob$ for any interval $[t,S]$
    \begin{enumerate}[(i)]
        \item $\tau$ has $\filt^{W,B}$ intensity  process $s\to  \wt{\gamma}(s,X_s)$ where $\wt{\gamma}\in C^{(1,1)}([0,\infty)\times \OO; (0, \infty))$.
        \item $s\to \wt{W}_s \dfn W_s + \int_t^s \tprobnu(u,X_u)du$ is a $(\tprob,\filt^W)$ Brownian motion, where $\tprobnu\in C^{(1,1)}([0,\infty) \times \OO; (0, \infty))$.
    \end{enumerate}
\end{assumption}

\begin{remark} When $d=k$ and $\rho \equiv \idmat{d}$  from \eqref{E:tprob_mpr_c}, we have $\tprobnu = \sigma_e^{-1}\left(\mu_e - \wt{\gamma}\ell_e\right)$ while in the strictly incomplete case, $\tprobnu$ is an exogenous function. In each case, we assume $\wt{\nu}$ satisfies part (ii) of Assumption \ref{A:alt_gamma}.
\end{remark}

For the CDS maturity $\wt{T}\geq t$, Assumption \ref{A:alt_gamma} implies $X$ is non-explosive under $\tprob$ over $[t,\wt{T}]$, with dynamics
\begin{equation}\label{E:tprob_SDE}
    dX_s = \left(b - a\tprobnu\right)(s,X_s)ds + a(s,X_s)d\wt{W}_s;\qquad X_t = x.
\end{equation}
This, along with the non-negativity of $\wt{\gamma}$ allows us to define the functions on $[t,\wt{T}]\times \OO$
\begin{equation}\label{E:wtu_wtv_dfn}
    \wt{u}(s,y) \dfn 1 - \wtcondexpv{}{}{\e^{-\int_s^{\wt{T}} \wt{\gamma}(u,X_u)du}}{X_s =y};\quad \wt{v}(s,y) \dfn \wtcondexpv{}{}{\int_s^{\wt{T}} \e^{-\int_s^v \wt{\gamma}(u,X_u)du}dv}{X_s=y},
\end{equation}
where we have written $\wt{\mathbb{E}}$ instead of $\mathbb{E}^{\tprob}$. As can be deduced from \cite[pp. 2507]{MR2474544}, $\wt{\kappa}(s,X_s) = \wt{u}(s,X_s)/\wt{v}(s,X_s)$ is the fair CDS spread at time $s$ for a CDS which offers full protection over the horizon $[s,\wt{T}]$.  With this notation, we obtain from \cite[Lemma 2.4]{MR2474544} the following proposition, the proof of which is given in Appendix \ref{AS:opt_invest}.

\begin{proposition}\label{P:cds_dynamics}
Under Assumptions \ref{A:region}  and \ref{A:alt_gamma}, the per-dollar wealth process $R_{\cdot}$ associated with the rolling CDS strategy has dynamics on $[t,\wt{T}]$
\begin{equation}\label{E:cds_dynamics}
    d R_s =   \mathbbm{1}_{s \leq \tau}\left(\sigma_r(s,X_s)'a(s,X_s)\left(dW_s + \tprobnu(s,X_s)ds\right) - \wt{\gamma}(s,X_s)ds\right) + dH_s,
\end{equation}
where
\begin{equation}\label{E:sigR_def}
    \sigma_r(s,y) \dfn \wt{u}(s,y) \times \nabla_{y} \log\left(\frac{\wt{u}(s,y)}{\wt{v}(s,y)}\right).
\end{equation}

\end{proposition}

\begin{remark}Above, $\wt{u}$ is the value of the CDS protection, and  $a'\nabla_y \log(\wt{u}(s,y)/\wt{v}(s,y))$ is the volatility of the log CDS spread. Also, when the intensity function $\wt{\gamma}$ only depends on time, neither $\wt{u},\wt{v}$ depend on $y$, and hence $\sigma_r \equiv 0_d$ implying
\begin{equation*}
    dR_s = dH_s - \mathbbm{1}_{s\leq \tau} \wt{\gamma}(s)ds.
\end{equation*}
\end{remark}

Proposition \ref{P:cds_dynamics} allows us to associate with the CDS market a fictitious asset $S^r$ with dynamics over $[t,\wt{T}]$
\begin{equation*}
    \frac{\dif S^r_{s}}{S^r_{s-}} = \mathbbm{1}_{s \leq \tau} \left(\left(\sigma_r'a\tprobnu -\wt{\gamma} \right)(s, X_s)ds + \left(\sigma_ra\right)(s, X_s) d W_s \right) + d H_s,
\end{equation*}
and hence we combine $S^e, S^r$ into a  $k+1$ dimensional  process $S$ with dynamics
\begin{equation}\label{E:S_def}
\frac{dS_s}{S_{s-}} = \mathbbm{1}_{s\leq \tau}\left(\mu(s,X_s)ds + \sigma_W(s,X_s)dW_s + \sigma_B(s,X_s)dB_s\right) - \ell(s,X_s)dH_s,
\end{equation}
where
\begin{equation}\label{E:coeffs}
    \mu = \begin{pmatrix} \mu_e \\ \sigma_r' a \tprobnu - \wt{\gamma} \end{pmatrix}, \hspace{10pt} 
    \sigma_W = \begin{pmatrix} \sigma_e \rho \\ \sigma_r' a \end{pmatrix}, \hspace{10pt}  
    \sigma_B = \begin{pmatrix} \sigma_e \ol{\rho} \\ 0 \end{pmatrix}, \hspace{10pt} 
    \ell = \begin{pmatrix} \ell_e \\ -1 \end{pmatrix}.
\end{equation}

\subsection*{Martingale measures, wealth processes and acceptable trading strategies} Given the traded assets $S$, we define the class of equivalent local martingale measures, wealth processes and admissible strategies in the usual manner.  The investment horizon is $T<\wt{T}$, and the equivalent local martingale measures are
\begin{equation*}
        \M \dfn \cbra{ \qprob \such \qprob \sim \prob \textrm{ on } \G_T \textrm{ and } S \textrm{ is a } \qprob \textrm{ local martingale}}.
\end{equation*}
With an eye towards the optimal investment problem for an agent with CARA preferences, we set $\wt{M}$ as the subset with finite relative entropy
\begin{equation*}
    \wt{\M} \dfn \cbra{\qprob \in \M \such \relent{\qprob}{\prob} < \infty},
\end{equation*}
where for $\nu<<\nu$, $\relent{\nu}{\mu} = \expv{\nu}{}{\log(d\nu/d\mu)}$ is the relative entropy of $\nu$ with respect to $\mu$.

Trading strategies in $S$ are denoted by $\pi$, where for $j=1,...,k+1$, and $t\leq s \leq T$,  $\pi^j_s(\omega)$ is the dollar position in $S^j$ at time $s$ and scenario $\omega$.  We require $\pi\in\mcp(\filtg)$, the $\filtg$ predictable sigma-field, and hence $\pi$ coincides with a $\filt^{W,B}$ predictable process on the stochastic interval $[t,\tau]$ (see \cite{bielecki2009credit}).  We  write $\pi = (\theta,\delta)$ where $\theta$ is the position in $S^e$ and $\delta$ the position in $S^r$.  If $\pi$ is additionally $S$-integrable, the resultant wealth process for a time $t$ initial wealth $w$ is denoted $\We^{\pi}$ or $\We^{(\theta,\delta)}$ and has dynamics (omitting the $(s,X_s)$ function argument)
\begin{equation*}
    \begin{split}
        d\We^{\pi}_s = &  \pi_s'\left(\mathbbm{1}_{s \leq \tau} \left(\mu ds + \sigma_W dW_s + \sigma_B dB_s\right)-  \ell dH_s\right); \hspace{20pt} \We_t^{\pi} = w.
    \end{split}
\end{equation*}
With this notation, the admissible class of strategies $\A$ is
\begin{equation}\label{E:admiss}
\A = \cbra{ \pi \in \mcp(\filtg) \such \pi \textrm{ is $S$ integrable, }\We^{\pi} \textrm{ is a $\qprob$ supermartingale for all } \qprob\in \wt{\M}}.
\end{equation}

\subsection*{The agent, random endowment, and the optimal investment problem}

Having defined the market and trading strategies, we now turn to the agent, who derives utility from terminal consumption using the exponential or CARA utility function
\begin{equation*}
    U(w) \dfn - e^{-\alpha w};\qquad w\in \reals.
\end{equation*}
In addition to trading in $S$, the agent has a non-traded random endowment of the form
\begin{equation*}
    \phi(X_T)1_{\tau > T}  + \psi(\tau,X_{\tau})1_{\tau \leq T}.
\end{equation*}
Above, $\phi$ is a claim with payoff contingent upon no default by $T$. While the primary examples we have in mind are either no claim ($\phi \equiv 0$) or $q$ notional of a defaultable bond ($\phi\equiv q$), motivated by the discussion on exotic credit linked derivatives in \cite{schonbucher2003credit} we allow for payoffs which may depend upon $X_T$. Conversely, $\psi$ represents any ``payoff'' the investor may receive upon default.  The idea is that even though in our model investment stops at $\tau$, in reality there will be investment opportunities after $\tau$, and $\psi(\tau,X_{\tau})$ is the time $\tau$ value of future investment over $[\tau,T]$.  In Section \ref{S:numerics} will explicitly construct  $\psi$ using optimal investment results for affine stochastic volatility models as found in  \cite{MR2660149}, but for now we take $\psi$ as given.  Our assumptions on $\phi,\psi$ are

\begin{assumption}\label{A:phi_psi_alt} $\phi \in C^{2,\beta}(\OO,\reals)$\footnote{$C^{2,\beta}(\OO;\reals)$ consists of $C^{2}$ functions on $\OO$ whose first and second partial derivatives are $\beta$ H\"{o}lder continuous on each $\ol{\OO}_n$. Similarly, $C^{(1,1),\beta}([0,T]\times \OO;[0,\infty))$ requires  appropriate H\"{o}lder continuity on each $[0,T]\times\OO_n$.} is bounded from below with $\underline{\phi} \dfn 0 \wedge \inf_{x \in \OO} \phi(x)$. For each $n$ 
\begin{equation*}
    \sup_{t\leq T, x\in \ol{\OO}_n} \wtexpv{}{}{\phi(X^{t, x}_T)} <\infty.
\end{equation*}
$\psi \in C^{(1, 1)}([0, T] \times \OO; [0,\infty))$, with $\psi (T, \cdot) \equiv 0$. Either $\psi$ is bounded from above, or for each $n$
\begin{equation*}
        \sup_{t\leq T, x \in \ol{\OO}_n}\wtexpv{}{}{\int_t^T (\psi\wt{\gamma})(u, X^{t,x}_u)du} < \infty.
\end{equation*}
\end{assumption}

Recalling the starting time/location $(t,x)$ and wealth $w$, the agent's optimal investment problem is to identify
\begin{equation}\label{E:vf_tx}
    u(t, x, w) \dfn \sup_{\pi\in\A} \expvs{ - \xpn{-\alpha \left( \We^{\pi}_T + \phi(X_T) \mathbbm{1}_{\tau > T} + \Psi(\tau, X_{\tau}) \mathbbm{1}_{\tau \leq T} \right)}}.
\end{equation}
Lastly, we record the dual problem to \eqref{E:vf_tx}, as it is used throughout.
\begin{equation}\label{E:dual_problem}
    v(t, x) \dfn \inf_{\qprob \in \wt{\M}} \left( \relent{\qprob}{\prob} + \alpha \expv{\qprob}{}{\phi(X_T)1_{\tau > T}  + \psi(\tau,X_{\tau})1_{\tau \leq T}} \right).
\end{equation}
The dual problem clarifies our assumptions on $\phi,\psi$.  Indeed, Assumption \ref{A:phi_psi_alt} simply posits the existence of a martingale measure which integrates the claims $\phi,\psi$, with a slight strengthening to ensure the expected values are locally uniformly bounded in the starting points $(t,x)$.

\section{The Certainty Equivalent Hamilton Jacoby Bellman (HJB) Equation}\label{S:HJB_CE} In this section we identify the HJB equation and PDE for the certainty equivalent function.  After formally identifying the HJB equation, we will separate our presentation into two cases (see Assumptions \ref{A:complete_mkt}, \ref{A:incomplete_mkt} below), according to when the market absent default is complete or incomplete. Proofs of results in this section are given in Appendix \ref{AS:HJB_CE}.

Due to exponential preferences, the initial wealth $w$ factors out of \eqref{E:vf_tx} so that $u(t, x, w) = e^{-\alpha w} u(t, x, 0)$, and we define the certainty equivalent function 
\begin{equation} \label{certainty equivalent}    
    G(t, x) \coloneqq -\frac{1}{\alpha} \log \left( -u(t, x,0) \right).
\end{equation}
Next, define the instantaneous covariation matrices\footnote{$\Upsilon_e(s,X_s) = d\langle S^e,X\rangle_s/ds$ and $\Upsilon(s,X_s) = d\langle S,X\rangle_s/ds$.}
\begin{equation}\label{E:SigE_UpsE}
    \Sigma_e \dfn \sigma_e \sigma_e';\qquad \Upsilon_e \dfn \sigma_e \rho a';\qquad \Sigma \dfn \begin{pmatrix} \Sigma_e & \Upsilon_e\sigma_r \\ \sigma_r' \Upsilon_e' & \sigma_r' A \sigma_r \end{pmatrix};\qquad \Upsilon \dfn \begin{pmatrix} \Upsilon_e \\ \sigma_r' A \end{pmatrix}.
\end{equation}
Using the martingale optimality principle\footnote{A formal derivation is presented in Appendix \ref{AS:HJB}. However, we will use PDE and duality methods to rigorously verify that $G$ is the certainty equivalent function.}, the certainty equivalent function $G$ is expected to solve the PDE
\begin{equation}\label{E:G_HJB}
    \begin{split}
        0 &= G_t + LG -\frac{\alpha}{2}\nabla G' A \nabla G + \frac{\gamma}{\alpha} + \sup_{\pi} H(\pi,G,\nabla G),\qquad \phi = G(T,\cdot),
    \end{split}
\end{equation}
where $\nabla = \nabla_x$ is the gradient operator for $x$, $L$ is the extended generator for $X$ under $\prob$
\begin{equation*}
    L \dfn \frac{1}{2}\tr\left(A D^2\right) + b'\nabla,
\end{equation*}
and $D^2$ is the Hessian operator. The Hamiltonian is
\begin{equation}\label{E:hamil}
H(\pi,g,p) = \pi'\left(\mu - \alpha \Upsilon p\right) - \frac{\alpha}{2}\pi'\Sigma\pi - \frac{\gamma}{\alpha} e^{\alpha(g + \pi'\ell - \psi)}.
\end{equation}
In \eqref{E:G_HJB}, the first equation must hold for $0\leq t < T, x\in\OO$ and the second for $x\in\OO$.

Formally, the Hamiltonian  \eqref{E:hamil} coincides with that in \cite[Equation (4)]{MR4086602} and hence from \cite[Equations (7), (8)]{MR4086602} one expects the optimal policy function
\begin{equation}\label{E:hat_pi}
\wh{\pi} = \frac{1}{\alpha}\Sigma^{-1}\left(\mu - \alpha\Upsilon p - \frac{\plogf{g}{p}}{\ell'\Sigma^{-1}\ell}\ell\right),
\end{equation}
and reduced Hamiltonian
\begin{equation}\label{E:MF_hamil}
    \begin{split}
        \Hcal(g,p) \dfn \sup_{\pi} H(\pi,g,p) = \frac{1}{2\alpha}(\mu - \alpha\Upsilon p)'\Sigma^{-1}(\mu - \alpha\Upsilon p) - \frac{\textrm{PL}(g,p)^2 + 2\textrm{PL}(g,p)}{2\alpha\ell'\Sigma^{-1}\ell}.
    \end{split}
\end{equation}
Here, we first defined
 \begin{equation*}
\plog{y} \dfn (ye^{y})^{-1};\quad y > 0,
\end{equation*}
as the inverse of $y e^{y}$ on $y>0$\footnote{\textrm{PL} is called the ``Product-Log'' or ``Lambert-W'' function, and we summarize its properties in Appendix \ref{AS:PL}.} and then, abusing notation, set
 \begin{equation}\label{E:plf_def}
     \textrm{PL}(g,p) \dfn \plog{\gamma \ell'\Sigma^{-1}\ell e^{\alpha g - \alpha\Psi +\ell'\Sigma^{-1}(\mu-\alpha\Upsilon p)}}.
 \end{equation}
However, there is a crucial difference which prohibits us from directly importing the results of \cite{MR4086602}. Namely, therein it was assumed  (similarly to Assumption \ref{A:SE_reg} for $\Sigma_e$), that $\Sigma = \Sigma(t,x) \in \spos{k}$ for all $t\leq T, x\in\OO$. Presently, $\Sigma$ from \eqref{E:SigE_UpsE} is the instantaneous covariation matrix for both the equity and CDS markets, and may not be strictly positive definite.  Indeed, one can show $\Sigma$ is not invertible when
\begin{equation*}
    \rho'\rho = \idmat{d},\qquad \textrm{ or } \qquad \sigma_r = 0_d.
\end{equation*}
Fortunately, it turns out that degeneracy of $\Sigma$ does not pose a problem, provided we separate analysis into two cases, corresponding to when the market is complete or not.

\section{Complete market}\label{S:complete_mkt}  In the first case, we analyze when the $(S^e,S^r)$ market is complete. Throughout, Assumptions \ref{A:region}, \ref{A:SE_reg}, \ref{A:alt_gamma} and \ref{A:phi_psi_alt} are in force. As a first step towards enforcing completeness, we assume the number of assets $k$ equals the number of factors $d$ and the correlation matrix function $\rho$ is identically equal to $\idmat{d}$.  In view of \eqref{E:Z_BM} the Brownian motion $B$ is irrelevant and we remove $B$ entirely by setting $\filtg$ as the $\prob$ augmentation of $W$ and $H$'s natural filtration\footnote{This is not technically required, but allows to assert, for example, that the market trading in $S^e$ absent default is complete with unique martingale measure $\tprob$ without worrying about the dynamics of $B$ under any martingale measure, as ultimately they are irrelevant to the optimal investment problem.}.  As is clear from \eqref{E:tprob_mpr_c}, $\tprob$ is the unique martingale measure for $S^e$ absent default. And, provided we make one additional assumption, $\tprob$ will be the unique martingale measure for $(S^e,S^r)$ and have finite relative entropy with respect to $\prob$ for all starting points $(t,x)$.  To state the assumption, define the functions 
\begin{equation}\label{E:E_barA_def}
\Efnalt(t,x) \dfn 1 + (\sigma_r'a\sigma_e^{-1}\ell_e)(t,x),
\end{equation}
and
\begin{equation}\label{E:Qval_complete}
Q_c(t,x)\dfn  \left(\frac{1}{2} \abs{\sigma_e^{-1}(\mu_e-\wt{\gamma}\ell_e)}^2 + \wt{\gamma} \left(\frac{\gamma}{\wt{\gamma}} - \log\left(\frac{\gamma}{\wt{\gamma}}\right) - 1\right)\right)(t,x).
\end{equation}
With these definitions, we assume
\begin{assumption}\label{A:complete_mkt}  The number of assets $k$ equals the number of factors $d$, the correlation matrix $\rho$ function is identically $\idmat{d}$, and $\filtg = \filt^{W,H}$. Additionally, $\Efnalt(t,x) \neq 0$ on $[0,T]\times\OO$ and for each $n$
\begin{equation*}
    \sup_{t\leq T, x\in\OO_n} \wtexpv{}{}{\int_t^T Q_c(u,X^{t,x}_u)du} < \infty.
\end{equation*}
\end{assumption}

\subsection{On Assumption \ref{A:complete_mkt}} Let us discuss the condition on $\Efnalt$ needed to ensure completeness of the $(S^e,S^r)$ market. First, for any measure $\qprob \in \M$ with density process $\Ecal\left(\int_t^\cdot \mbf{A}_s'dW_s + \int_t^{\cdot} \mbf{C}_s dM_s\right)$, using \eqref{E:S_def} with $\rho\equiv \idmat{d}$ (and hence $\ol{\rho}\equiv 0_{d\times d}$) the market price of risk equations are
\begin{equation*}
    0 = \left(\begin{array}{c} \mu_e - \gamma \ell_e \\ \sigma_r'a\sigma_e^{-1}(\mu_e - \wt{\gamma} \ell_e) - \wt{\gamma} + \gamma \end{array}\right) + \left(\begin{array}{c} \sigma_e \\ \sigma_r'a\end{array}\right)\mbf{A} - \gamma \mbf{C}\left(\begin{array}{c} \ell_e \\ -1\end{array}\right).
\end{equation*}
The top equation gives $\mbf{A} = -\sigma^{-1}_e\left(\mu_e - \gamma(1+\mathbf{C})\ell_e\right)$.  Plugging this into the bottom equation, and using \eqref{E:E_barA_def}, $\mbf{C}$ must solve
\begin{equation*}
    0 = \Efnalt\left(\gamma(1+\mbf{C}) - \wt{\gamma}\right).
\end{equation*}
Thus, if $\Efnalt\neq 0$ on $[t,T]\times \OO$, then we must have $\mbf{C} = \wt{\gamma}/\gamma - 1$ which implies $\qprob = \tprob$ and hence $(S^e,S^r)$ market completeness.  But, if $\Efnalt$ can degenerate (with positive $\textrm{Leb}_{[t,T]}\times\prob$ probability), there are many solutions $\mbf{C}$, and hence the market is not complete.  Alternatively, from the trading strategy perspective, for $\pi = (\theta,\delta)$ the corresponding wealth process has $\tprob$ dynamics
\begin{equation*}
    d\We^{\pi}_s = 1_{s\leq \tau}\left(\sigma_e(s,X_s)'\theta_s+ \delta_s a(s,X_s)'\sigma_r(s,X_s)\right)d\wt{W}_s + \left(\delta_s - \theta_s'\ell_e(s,X_s)\right) d\wt{M}_s.
\end{equation*}
By translating $\theta \to -\delta (\sigma'_e)^{-1}a'\sigma_r + \theta$ the resultant dynamics are
\begin{equation*}
    d\We^{(\theta,\delta)}_s = 1_{s\leq \tau}\theta_s'\sigma_e(s,X_s)d\wt{W}_s - \theta_s'\ell_e(s,X_s) d\wt{M}_s +  \delta_s\Efnalt(s,X_s) d\wt{M}_s = \theta_s'\frac{dS^e_s}{S^e_{s-}} +  \delta_s\Efnalt(s,X_s) d\wt{M}_s.
\end{equation*}
Thus, we may equate investment in $(S^e,S^r)$ with investment in $(S^e,S^{\wt{r}})$ where 
\begin{equation*}
    \frac{dS^{\wt{r}}_{s}}{S^{\wt{r}}_{s-}} = \Efnalt(s,X_s)d\wt{M}_s = \Efnalt(s,X_s)\left(dH_s - 1_{s\leq \tau}\wt{\gamma}(s,X_s) ds\right).
\end{equation*}
If $\Efnalt\neq  0$ we can always hedge against default, but if $\Efnalt$ degenerates then investment in $S^{\wt{r}}$ does not offer default protection.  Thus, to hedge against default risk, we need $\Efnalt\neq 0$.  Lastly, note that $\Efnalt = 1$ when $\sigma_r  =0$, so that degeneracy of $\sigma_r$ does not lead to incompleteness.

Next, let us discuss the condition on $Q_c$. First, as the map $z \to z - 1 -\log(z)$ is non-negative on $(0,\infty)$ we know $Q_c\geq 0$. Next, for a given starting point $(t,x)$ calculation shows that if $\wtexpv{}{}{\int_t^T Q_c(u,X^{t,x}_u)du} < \infty$ then
\begin{equation*}
    \relent{\tprob}{\prob}(t,x)  = \wtexpv{}{}{\int_t^T e^{-\int_t^u \wt{\gamma}(v,X^{t,x}_v)dv} Q_c(u,X^{t,x}_u)du} < \wtexpv{}{}{\int_t^T Q_c(u,X^{t,x}_u)du} < \infty.
\end{equation*}
As such, our assumption $Q_c$ (along with Assumption \ref{A:phi_psi_alt}) is essentially the minimal one needed to ensure the optimal investment problem is well posed: that the unique martingale measure has finite relative entropy for all starting points.   We have only slightly strengthened this assumption, by removing the $\textrm{exp}(-\int_t^u \wt{\gamma}(v,X^{t,x}_v)dv)$ term, and requiring finite-ness locally uniformly in $(t,x)$. 

\subsection{The Hamiltonian and HJB equation}

We now identify the reduced Hamiltonian \eqref{E:MF_hamil} and HJB PDE \eqref{E:G_HJB}, the latter of which linearizes,  as expected due to market completeness.  To state the result we define the extended generator $\wt{L}$ of $X$ under $\tprob$ (see \eqref{E:tprob_SDE}) when $\nu = \sigma_e^{-1}(\mu_e - \wt{\gamma}\ell_e)$  
\begin{equation}\label{E:wtL_def}
\wt{L} \dfn \frac{1}{2}\tr\left(AD^2\right) + \nabla '\left(b- a\sigma_e^{-1}\left(\mu_e - \wt{\gamma}\ell_e\right)\right).
\end{equation}

\begin{proposition}\label{P:general_hamil_11}
Under Assumption \ref{A:complete_mkt}, $\Hcal$ from \eqref{E:MF_hamil} takes the form
\begin{equation}\label{E:Hcal_is_ok_1}
    \begin{split}
        \Hcal_{c}(g,p) &= \frac{1}{\alpha}(Q_c-\gamma) + \wt{\gamma}(\psi-g) + \frac{\alpha}{2}p'Ap - p'a\sigma_e^{-1}(\mu_e-\wt{\gamma}\ell_e).
    \end{split}
\end{equation}
Writing $\pi = (\theta,\delta)$ where $\theta$ is the equity position and $\delta$ the CDS position, the candidate optimal policy functions are (recall \eqref{E:E_barA_def}) 
\begin{equation}\label{E:gen_opt_strat_1}
    \begin{split}
    \wh{\delta}_c(g,p) &= \frac{1}{\alpha\Efnalt}\left(\ell_e'\Sigma_e^{-1}\left(\mu_e - \wt{\gamma}\ell_e - \alpha\sigma_e a' p\right) + \log\left(\frac{\gamma}{\wt{\gamma}}\right) - \alpha(\psi-g) \right),\\
    \wh{\theta}_c(g,p) &= \frac{1}{\alpha}\Sigma_e^{-1}\left(\mu_e - \wt{\gamma}\ell_e - \alpha \sigma_e a' \left(p  + \wh{\delta}_c(g,p)\sigma_r\right)\right).
    \end{split}
\end{equation}
The PDE \eqref{E:G_HJB} specifies to the linear parabolic PDE
\begin{equation}\label{E:PDE_1}
    \begin{split}
        0 &= G_t + \wt{L}G  + \wt{\gamma}(\psi-G) + \frac{1}{\alpha} Q_c;\qquad \phi = G(T,\cdot).
    \end{split}
\end{equation}

\end{proposition}

\begin{remark} \label{R:Feynman-Kac}  
The PDE \eqref{E:PDE_1} is expected to admit the solution
\begin{equation*}
    \begin{split}
        G(t,x) &= \wtexpv{}{}{1_{\tau>T}\phi(X^{t,x}_T) + 1_{\tau\leq T}\psi(\tau,X^{t,x}_{\tau})} + \frac{1}{\alpha}\relent{\tprob}{\prob}(t,x),\\
        &= \wtexpv{}{}{e^{-\int_t^T \wt{\gamma}(u,X^{t,x}_u)du}\phi(X^{t,x}_T) + \int_t^T e^{-\int_t^v \wt{\gamma}(u,X^{t,x}_u)du}\left(\frac{1}{\alpha}Q_c + \wt{\gamma}\psi\right)(v,X^{t,x}_v)dv}.
    \end{split}
\end{equation*}
The first equality is expected from the general duality theory, and the second follows as $\wt{\gamma}$ is the $\filt^W$ default intensity function for $\tau$ under $\tprob$.   The second equality is also expected using Feynman-Ka\v{c}.  While there are certain assumptions needed to ensure the above two equalities (see for example \cite{Heath-Schweizer} for the Feynman-Ka\v{c} method or \cite{MR1891730} for the duality method) we will prove existence of solutions to \eqref{E:PDE_1} which are verified to be the certainty equivalent using the general results of Appendix \ref{AS:main_result}, which are valid in both the complete and incomplete settings.
\end{remark}

\subsection{Optimal Policies}\label{SS:detwtg}  Here we analyze the optimal policies, as a different (yet intuitive)  phenomena arises when compared to the optimal policies in \cite{MR4086602}, where the investor does not have access to the CDS market. Using \eqref{E:gen_opt_strat_1}, one can show  $\wh{\theta}_c,\wh{\delta}_c$ satisfy the relationship
\begin{equation*}
    \wh{\delta}_c =  \ell_e'\wh{\theta}_c  + g-\psi + \frac{1}{\alpha}\log\left(\frac{\gamma}{\wt{\gamma}}\right).
\end{equation*}
On the right side above, $\ell_e'\wh{\theta}_c$ is the loss in equity wealth upon default. The quantity $g-\psi$ is the loss of wealth due to the termination of investment opportunities, as $g$ is the indirect utility from future trading if default has not yet occurred, and $\psi$ is the payment upon default. Thus, if default occurs the investor ``loses'' the value from additional trading $g$, but ``gains'' the payment $\psi$.  

To understand the  $(1/\alpha)\log(\gamma/\wt{\gamma})$ term, first recall that $\tprob$ is the unique martingale measure. Next, assume the investor is at time $s$, with $X_s = y$, $G=g$, and default has not occurred. The investor is worried about default in the next instant $ds$, and she knows upon default she will lose $\ell_e'\wh{\theta} + g-\psi$, so she takes a position of this size in the rolling CDS to cover this loss.

However, she also wants to optimally invest over the next instant, where nothing changes except the possibility of default. As her trading strategies must be predictable, she cannot adjust her CDS position after the fact.  Therefore, according the standard optimal investment theory, she seeks a position in the CDS which ensures her marginal utility  is proportional to the state price density given her information. As $\tprob$ is the unique martingale measure, the additional position $\delta$ ensures the first order optimality conditions $e^{-\alpha\delta} = d\tprob/ d\prob |_{\cbra{\tau>s,\F^{W,B}_s}}$. This gives
\begin{equation*}
    e^{-\alpha \delta} = \frac{d\tprob}{d\prob}\big|_{\cbra{\tau>s,\F^{W,B}_s}} \approx \frac{\wtcondprobs{\tau \leq s + ds}{\tau > s, \F^{W,B}_s}}{\condprobs{\tau \leq s + ds}{\tau > s, \F^{W,B}_s}} \approx \frac{\wt{\gamma}_s}{\gamma_s},
\end{equation*}
and hence the appearance of the term $(1/\alpha)\log(\gamma/\wt{\gamma})$.

To summarize, while naively one might conjecture the rolling CDS position should only cover losses in the equity position, this is not true. Instead, the CDS position covers not only the loss in equity upon default but also the effective loss due to the stoppage of trade. As each of these losses may be exactly estimated prior to default, the position is also adjusted to satisfy the standard optimality conditions over the next instant, accounting for the investor's information set.  And, as shown in \cite[Section 2.4]{MR4086602}, the ratio $\wt{\gamma}/\gamma$ can be interpreted as the risk premia due to default, and hence the agent sets her position to ensure marginal utility is equal to the credit risk premia.

\subsection{Existence and Verification}  We conclude this section with the existence and verification theorem in the complete market setting.

\begin{theorem}\label{T:main_result_complete} Under Assumption \ref{A:complete_mkt}, the certainty equivalent $G$ is in $C^{1,2}((0,T)\times \OO;\reals)$ and satisfies the PDE \eqref{E:PDE_1}. The optimal equity and rolling CDS strategies are given in \eqref{E:gen_opt_strat_1}, evaluated at $(s,X_s)$ for $s\in[t,T]$.
\end{theorem}

\section{Incomplete Market Case}\label{S:incomplete_mkt}  In this section, we treat the general case where the number of assets and factor need not coincide, and where the independent Brownian motion $B$ in \eqref{E:S_def} is present.  Here, the analysis is more involved, and we enforce stronger assumptions on the model coefficients. In order to state our main assumption, recall that $\tprob$ is defined using from \eqref{E:tprob_mpr_ic}, where $\tprobnu,\wt{\gamma}$ are in Assumption \ref{A:alt_gamma}.  Given this, similarly to \eqref{E:Qval_complete} define 
\begin{equation}\label{E:Qval_incomplete}
Q_i(t,x)\dfn  \bigg(\frac{1}{2} \abs{\tprobnu}^2 + \frac{1}{2}\abs{\ol{\rho}^{-1}(\sigma_e^{-1}(\mu_e - \wt{\gamma}\ell_e)-\rho\tprobnu)}^2 + \wt{\gamma} \left(\frac{\gamma}{\wt{\gamma}} - \log\left(\frac{\gamma}{\wt{\gamma}}\right) - 1\right)\bigg)(t,x).
\end{equation}
We then assume
\begin{assumption}\label{A:incomplete_mkt}
There is an $\eps_1>0$ such that $(1-\eps_1)\idmat{k} - \rho\rho'\in \mathbb{S}^k_{+}$ on $[0,T]\times\OO$. The spot pricing measure $\tprob$ is defined though \eqref{E:tprob_mpr_ic} where $\wt{\gamma},\wt{\nu}$ are in Assumption \ref{A:alt_gamma}. For $\sigma_r$ in \eqref{E:sigR_def}, either $\abs{\sigma_r}\equiv 0$ on $[0,T]\times\OO$ or $\abs{\sigma_r} > 0$ on $[0,T]\times\OO$.  Additionally,
\begin{enumerate}[(i)] 
\item  There is $\eps_2 > 0$ such that for each $n$,
\begin{equation*}
    \sup_{t\leq T, x \in \ol{\OO}_n} \wtexpv{}{}{\e^{\eps_2 \int_t^T Q_i(u,X^{t,x}_u) du}} < \infty.
\end{equation*}
\item  There is $p > 1$ such that a strong solution to the SDE $dX_t = (b+(p-1)a\tprobnu)(t,X_t)dt + a(t,X_t)dW_t$ exits for each starting time $t\leq T$ and $x\in\OO$.  Writing $X^{(p),t,x}$ as the solution, we have for each $n$,
\begin{equation*}
    \sup_{t\leq T, x \in \ol{\OO}_n} \expvs{\e^{\frac{1}{2}p(p - 1)\int_t^T \abs{\tprobnu(u,X^{(p),t,x}_u)}^2 d u}}  < \infty.
\end{equation*}
\end{enumerate}
\end{assumption}

\subsection{On Assumption \ref{A:incomplete_mkt}} That $\rho\rho' \leq (1-\eps_1)\idmat{k}$ clearly implies the market absent default is incomplete.   As for $\sigma_r$, we want to allow the default intensity $\wt{\gamma}$ under $\tprob$ to be deterministic, at which point $\sigma_r$ vanishes everywhere on the state space. However, when the intensity $\wt{\gamma}$ depends on $X$ as well, we do not want $\sigma_r$ to vanish on the interior of the state space \footnote{If $\abs{\sigma_r} = 0$ somewhere, but not everywhere, on the interior, we cannot use \cite[Theorem 11.3(b)]{MR1465184} as the quantity $B_{\infty}$ there-in is infinite.  This removes the key gradient bound for local solutions to \eqref{E:G_HJB}, needed to ensure a global solution exists.}.   Therefore, we assume $\abs{\sigma_r} \neq 0$.  This assumption can always be verified in examples (see Section \ref{S:numerics}). Additionally, from \eqref{E:sigR_def} we see that $\abs{\sigma_r} = 0$ precisely when $\abs{\nabla h} = 0$ for $h = \wt{u}/\wt{v}$.  Using \eqref{E:wtu_wtv_dfn} (see also the proof of Proposition \ref{P:cds_dynamics} in Appendix \ref{AS:opt_invest} below), one can show $h$ solves the PDE
\begin{equation*}
    0 = h_t + \wt{L}h + \nabla h'A \frac{\nabla \wt{v}}{\wt{v}} + \frac{1}{\wt{v}} h + \frac{\wt{\gamma}}{\wt{v}};\qquad \wt{\gamma}(\wt{T},\cdot) = h(\wt{T},\cdot)
\end{equation*}
where $\wt{L}$ is the extended generator of $X$ under $\tprob$. This implies $h$ admits a Feynman-Ka\v{c} representation, and the problem of studying when such functions have non-degenerate gradient is well known.  See, for example, \cite{anderson2008equilibrium, MR3131287, MR3590708}. In these articles, the key condition is that $\nabla \wt{\gamma}(\wt{T},\cdot)$ does not degenerate.  If this is the case, then under certain technical restrictions, the non-degeneracy is transferred to $h$ and hence $\sigma_r$.   The exponential integrability and non-explosion conditions are needed in order to obtain locally uniform bounds on the local solutions to the HJB equation obtained in Appendix \ref{AS:main_result} below.  Here, we note that by taking $\eps_2 \approx 0$ and $p\approx 1$ we are only requiring the existence of some exponential moment, and in specific examples, this can be checked.

\subsection{The Hamiltonian and HJB equation} When Assumption \ref{A:incomplete_mkt} holds,  a challenge seemingly arises in identifying the reduced Hamiltonian, as  $\Sigma$ from \eqref{E:SigE_UpsE} is not invertible when $\sigma_r = 0_d$.  However, it tunrs out that degeneracy of $\sigma_r$ does not pose a problem. Below we express our results in terms of $\sigma_r$ (otherwise omitting $(t,x)$ function arguments), and recall $Q_c$ from \eqref{E:Qval_complete}. 

\begin{proposition}\label{P:general_hamil_2}
Under Assumption \ref{A:incomplete_mkt}, $\Hcal$ from \eqref{E:MF_hamil} takes the form
\begin{equation}\label{E:Hcal_is_ok_2}
    \begin{split}
        \Hcal_i(g,p) &= \frac{1}{\alpha}(Q_c-\gamma)  + \wt{\gamma}(\psi-g)  + \frac{\alpha}{2}p'\Upsilon_e'\Sigma_e^{-1}\Upsilon_e p - p'\Upsilon_e'\Sigma_e^{-1}(\mu_e-\wt{\gamma}\ell_e) + R_{\Hcal}(\sigma_r,g,p),
    \end{split}
\end{equation}
where $R_{\Hcal}$ is given in \eqref{E:GR_def_new} (see also \eqref{E:GR_def}) below and satisfies $\R_{\Hcal}(0_d,g,p) = 0$. The optimal policy functions are 
\begin{equation}\label{E:gen_opt_strat_2}
    \begin{split}
    \wh{\delta}_i(g,p) &= \frac{1}{\alpha}\left(\ell_e'\Sigma_e^{-1}\left(\mu_e - \wt{\gamma}\ell_e - \alpha \Upsilon_e p \right) + \log\left(\frac{\gamma}{\wt{\gamma}}\right) - \alpha(\psi-g) + R_{\delta}(\sigma_r,g,p)\right),\\
    \wh{\theta}_i(g,p) &= \frac{1}{\alpha}\Sigma_e^{-1}\left(\mu_e - \wt{\gamma}\ell_e - \alpha \Upsilon_e\left( p  + \wh{\delta}_i(g,p)\sigma_r\right)\right),
    \end{split}
\end{equation}
where $R_{\delta}$ is defined in \eqref{E:Rdelta_def_new} (see also \eqref{E:Rdelta_def}) below and satisfies $R_{\delta}(0_d,g,p) = 0$. Lastly, the PDE in \eqref{E:G_HJB} specifies to the semi-linear parabolic Cauchy PDE
\begin{equation}\label{E:PDE_2}
    \begin{split}
        0 &= G_t + \ol{L}G  +\wt{\gamma}(\psi-G) + \frac{1}{\alpha}Q_c -\frac{\alpha}{2}\nabla G'a(\idmat{d}-\rho'\rho)a'\nabla G  + R_{\Hcal}(\sigma_r,G,\nabla G),\\
        \phi &= G(T,\cdot),
    \end{split}
\end{equation}
where
\begin{equation*}
\ol{L} \dfn \frac{1}{2}\tr\left(AD^2\right) + \nabla '\left(b- a\rho'\sigma_e^{-1}\left(\mu_e - \wt{\gamma}\ell_e\right)\right),
\end{equation*}

\end{proposition}

\begin{remark} When $k=d$ and $\rho = \idmat{d}$, note that $\ol{L}$ coincides with $\wt{L}$ from \eqref{E:wtL_def} and $\Hcal_i = \Hcal_c + \R_{\Hcal}$ where $\Hcal_c$ is from  \eqref{E:Hcal_is_ok_1}.  Also, it is not as obvious, but when the additional condition of Assumption \ref{A:complete_mkt} hold, $R_{\Hcal}$ vanishes and the PDE \eqref{E:PDE_2} reduces to that in \eqref{E:PDE_1}.   
\end{remark}

\subsection{The Optimal Policies}\label{SS:opt_policy}  In the general case, the optimal policies satisfy a (qualitatively) similar relationship as in the complete market case.  To show this,  assume  $\Sigma$ is invertible (the result holds even when $\Sigma$ degenerates, but the notation using \eqref{E:gen_opt_strat_2} is much more cumbersome). Here, the optimal $\wh{\pi} = (\wh{\theta},\wh{\delta})$ is from \eqref{E:hat_pi}, and using \eqref{E:coeffs}, \eqref{E:SigE_UpsE} and \eqref{E:plf_def} we deduce
\begin{equation*}
    \begin{split}
        \wh{\delta}(g,p) &= \wh{\theta}(g,p)'\ell_e - \wh{\pi}(g,p)'\ell,\\
        &= \wh{\theta}(g,p)'\ell_e + g-\psi  + \frac{1}{\alpha}\left(\plogf{g}{p} - \ell'\Sigma^{-1}\left(\mu-\alpha\Upsilon p\right) - \alpha(g-\psi)\right),\\
        &= \wh{\theta}(g,p)'\ell_e + g-\psi  + \frac{1}{\alpha}\log\left(\frac{\gamma \ell'\Sigma^{-1}\ell}{\plogf{g}{p}}\right),
    \end{split}
\end{equation*}
where the last equality follows from the identity $\plog{xe^y}-y = -\log(\plog{xe^y}/x)$.  As shown in \cite[Section 2.4]{MR4086602}, the map
\begin{equation*}
    s \to \frac{\plogf{G}{\nabla G}}{\ell'\Sigma^{-1}\ell}(s,X_s)
\end{equation*}
is the $\filt^{W,B}$ default intensity of $\tau$ under the dual optimal measure $\wh{\qprob}$.  Therefore, just as in the complete market case, the position in the CDS accounts for (i) the loss in the equity market should default occur, (ii) the effective loss due to the inability to trade should default occur, (iii) and the inverse marginal utility of the credit risk premia associated to the dual optimal measure.

\subsection{Existence and Verification} As in the complete market case, we conclude with the existence and verification result. Here, unlike in the complete market case where by default $\tprob$ was the dual optimal measure, we additionally explicitly identify the dual optimal measure $\wh{\qprob}\in\tM$.

\begin{theorem}\label{T:main_result}
Under Assumption \ref{A:incomplete_mkt}, the certainty equivalent $G$ is in $C^{1,2}((0,T)\times\OO; \reals)$ and solves the PDEs \eqref{E:PDE_2}. The optimal trading equity and CDS trading strategies $(\wh{\theta}, \wh{\delta})$ are given in \eqref{E:gen_opt_strat_2}. Lastly, the density process 
\begin{equation}\label{E:opt_mm_dens}
    \wh{Z}_s \dfn \e^{- \alpha \left( \We^{(\wh{\theta}, \wh{\delta})}_s + 1_{\tau > s} G(s, X_s)  + 1_{\tau \leq s} \psi(\tau,X_{\tau}) -G(t, x) \right) }, \quad t \leq s \leq T,
\end{equation}
defines a measure $\wh{Q} \in \wt{\M}$ that solves the dual problem \eqref{E:dual_problem}.
 
\end{theorem}

\section{Indifference Pricing for Defaultable Bonds}\label{S:indiff_pricing} As an application of our main results, we may price a defaultable zero coupon bond through the principle of utility indifference. This approach accounts for both market incompleteness and investor preferences.  To identify the formulas, let us denote by $u(\cdot;q)$ and $G(\cdot;q)$ the value function and certainty equivalent function respectively when $\phi(x) \equiv q$.  We then seek a price function $p(\cdot;q)$ such that the investor is indifferent between (i) not owning the defaultable bond and (ii) paying $qp(\cdot;q)$ to own $q$ notional of the defaultable bond.  Here, by indifferent we mean that the indirect utility the investor obtains trading in the equity-CDS market is the same in both cases.  Mathematically, from \eqref{E:vf_tx} we require for $t\leq T$, $x\in\OO$ and $w\in\reals$ that
\begin{equation*}
    u(t,x,w;0) = u(t,x,w-qp(t,x,w;q);q),
\end{equation*}
or, in terms of the certainty equivalents
\begin{equation*}
    p(t,x,w;q) = p(t,x;q) = \frac{1}{q}\left(G(t,x;q)-G(t,x;0)\right).
\end{equation*}
Note the price is independent of the initial capital, as is well known. In the complete case of Section \ref{S:complete_mkt}, using Remark \ref{R:Feynman-Kac} we find 
\begin{equation*}
    p(t,x;q) = p(t,x) = \wtexpv{}{}{1_{\tau^{t,x}>T}} = \wtexpv{}{}{e^{-\int_t^T \wt{\gamma}(u,X^{t,x}_u)du}}.
\end{equation*}
This is simply the unique arbitrage free price for one unit of the defaultable bond.  The more interesting case is when the market is incomplete.

Lastly, we can easily incorporate recovery into the above pricing. Indeed, with (time and state-dependent) recovery, the bond payoff becomes $q1_{\tau > T} + q R(\tau,X_{\tau})1_{\tau \leq T}$ and we can absorb the recover function $R$ into the payoff $\psi$.  The indifference price takes the same form as above.

\section{Numerical Application}\label{S:numerics}

In this section we consider when $X$ follows a CIR process and investment opportunities are affine functions.  This updates \cite[Section 4.2]{MR4086602} to when the investor may trade in the rolling CDS.  As we will see, unlike in \cite{MR4086602} the investor does not short the defaultable stock, rather she increases her position in the CDS contract.  As such, the CDS market provides the investor a means to hedge her default risk.

The horizon is $T = 1$ and we start at $t=0$. For fixed $x\in \OO = (0,\infty)$,  $X = X^{x}$ has dynamics
\begin{equation*}
    \dif X_s = \kappa(\theta - X_s) \dif s + \xi \sqrt{X_s} \dif W_s,\ X_0 = x,
\end{equation*}
where we impose the Feller's condition $2 \kappa \theta \geq \xi^2$ to ensure $X\in \OO$ for all times. For the sake of comparison, we use the same parameters in \cite{MR4086602} for $X$, $\kappa = 0.25$, $\theta = 0.06$ and $\xi = 0.1$. Lastly, we set the absolute risk aversion coefficient of investor $\alpha = 3$.

\subsection{Complete market}\label{SS:Ex_Comp}  In the first example, there is one equity with drift and volatility functions
\begin{equation*}
    \mu_e(s, y) = y \sigma \nu; \quad \sigma_e(s, y) = \sqrt{y} \sigma,
\end{equation*}
where $\nu = 4.0762$ and $\sigma = 0.9762$. The loss proportion upon default is a constant, $\ell_e(y) = \ell = 0.5$, and as there is only one asset, the post-default certainty equivalent $\psi \equiv 0$. The investor holds $q$ units of a defaultable bond. The default intensities are linear functions of the factor process with 
\begin{equation*}
    \gamma(s, y) = y \gamma;\quad \wt{\gamma}(s, y) = y \wt{\gamma}; \quad  \wt{\gamma} = 1.5\gamma,
\end{equation*}
where we choose $\gamma$ such that at the long term mean level $\theta$ of $X$ the one-year default probability is $1 - e^{- \gamma \theta} = 3\%$\footnote{This gives $\gamma = 0.5076$ and implicitly approximates $\gamma\int_0^1 X_u du \approx \gamma X_0$. Without this approximation, one can show using \eqref{E:wtD_def}, but computing under the physical measure, that $\gamma = 0.5080$.}. The expiration date of the CDS contract is $\wt{T} = 2$. Using \eqref{E:tprob_SDE}, the dynamics for $X$ under $\tprob$ are
\begin{equation*}
    dX_t = \wt{\kappa}\left(\wt{\theta}-X_t\right)dt + \xi\sqrt{X_t}d\wt{W}_t;\qquad \wt{\kappa} = \kappa + \xi\left(\nu - \frac{\wt{\gamma}}{2 \sigma}\right);\quad \wt{\theta} = \frac{\kappa\theta}{\wt{\kappa}}.
\end{equation*}
While Feller's condition holds under $\tprob$, to ensure $X\in\OO$ we also need $\wt{\kappa}>0$, but under our parameter assumptions we have $\wt{\kappa} = 0.6186$ and $\wt{\theta} = 0.0242$. Next, as it affects the optimal policies $\wh{\theta},\wh{\delta}$ in \eqref{E:gen_opt_strat_1} (if not the PDE in \eqref{E:PDE_1}), we calculate $\sigma_r$ from \eqref{E:sigR_def}, starting with $\wt{u}$ and $\wt{v}$ from \eqref{E:wtu_wtv_dfn}. Each of these depends on the function,
\begin{equation}\label{E:wtD_def}
    \wt{D}(s, v, y) \dfn \wtcondexpv{}{}{\e^{ -\wt{\gamma}\int_s^v X_u du } }{X_s = y}.
\end{equation}
The Markov property implies (abusing notation) $\wt{D}(s,v,y) = \wt{D}(v-s,y)$ and it is well known (see \cite{duffie1996yield}, \cite{MR1793362}, and \cite{MR1846525}) that $\wt{D}(u,y)  = \wt{A}(u) \e^{ - \wt{B}(u) y}$ for certain explicitly identifiable functions $\wt{A},\wt{B}$ with $\wt{B} >0$.  The formulas for $\wt{u}$, $\wt{v}$ and hence $\sigma_r$ are obtained using $\wt{D}$ as $ \wt{u}(s, y) = 1 - \wt{D}(\wt{T}-s, y)$ and $\wt{v}(s, y) = \int_s^{\wt{T}} \wt{D}(v-s, y) d v$. Lastly, one can show for $\wt{T} > 1$ that $\sigma_r \neq 0$, and in fact one has the bounds
\begin{equation*}
    \wt{D}(\wt{T}-s, y) \wt{B}(\wt{T} - s) \leq \sigma_r(s, y) \leq \wt{B}(\wt{T} - s).
\end{equation*}
Assumption \ref{A:phi_psi_alt} holds as $\phi$ is constant and $\psi = 0$. For Assumption \ref{A:complete_mkt}, first we have $\Efnalt \geq 1 > 0$ from \eqref{E:E_barA_def} as $\sigma_r$, $a$, $\sigma_e$, and $\ell_e$ are all positive scalars. Next, from \eqref{E:Qval_complete} we see $Q_c(s,y) = Q_c y$ for a certain positive constant $Q_c$. It is well known (see \cite[Equation (3.23)]{MR1846525}) that 
\begin{equation*} 
\wtexpv{}{}{X^{0, x}_u} = x e^{- \wt{\kappa}u} + \wt{\theta} (1 - e^{- \wt{\kappa} u } ) \leq \max(x, \wt{\theta}) 
\end{equation*}
for CIR processes, and hence Assumption Assumption \ref{A:complete_mkt} holds.

Figure \eqref{fig:Optimal Defaultable Equity Positions} compares the optimal equity positions (for difference face $q$) at $t=0$ in our model (left plot) where the CDS market is pesent, to the model of \cite{MR4086602} where the CDS market is absent (right plot). As we see, the presence of the CDS market enables the investor to avoid shorting the stock as a hedge against default risk, and as such provides a natural hedging instrument\footnote{In the model of \cite{MR4086602} the investor shorts the stock for low state variables due to mean-reversion, which implies the likelihood of default will increase in the future. As such, shorting the stock provides default protection.}.

Figure \eqref{fig:Optimal CDS Positions} shows the optimal CDS positions. The left plot shows the optimal CDS positions at time $0$ for different notional $q$ of the defaultable bond. Here we see the optimal CDS position increases along with $q$. The right plot fixes $q=1$ and for each $t\in [0,1]$ shows the range $[\min_{x>0}\wh{\delta}(t,x), \max_{x>0} \wh{\delta}(t,x)]$.  As such, we see the investor's position on the CDS displays remarkable little variation over the state space.  This in turn implies that the optimal position is relatively static in the CDS, and hence well-approximated by static CDS positions.

Lastly, Figure  \eqref{fig:RB} shows the relative benefit of the market with CDS versus that without CDS, defined as $\textrm{CE}^{\text{CDS}}/\textrm{CE}^{\text{no CDS}} - 1$, where $\text{CE}^{\text{CDS}}$ is the certainty equivalent with CDS, and $\textrm{CE}^{\text{no CDS}}$ is the certainty equivalent without CDS, computed in \cite{MR4086602}.  Especially as the defaultable bond position and default intensity increase, the CDS market provides the investor with a subtstantial benefit over the market without CDS.

\begin{figure}[ht]
\centering
\begin{subfigure}[b]{0.45\textwidth}
  \centering
  \includegraphics[width = \textwidth]{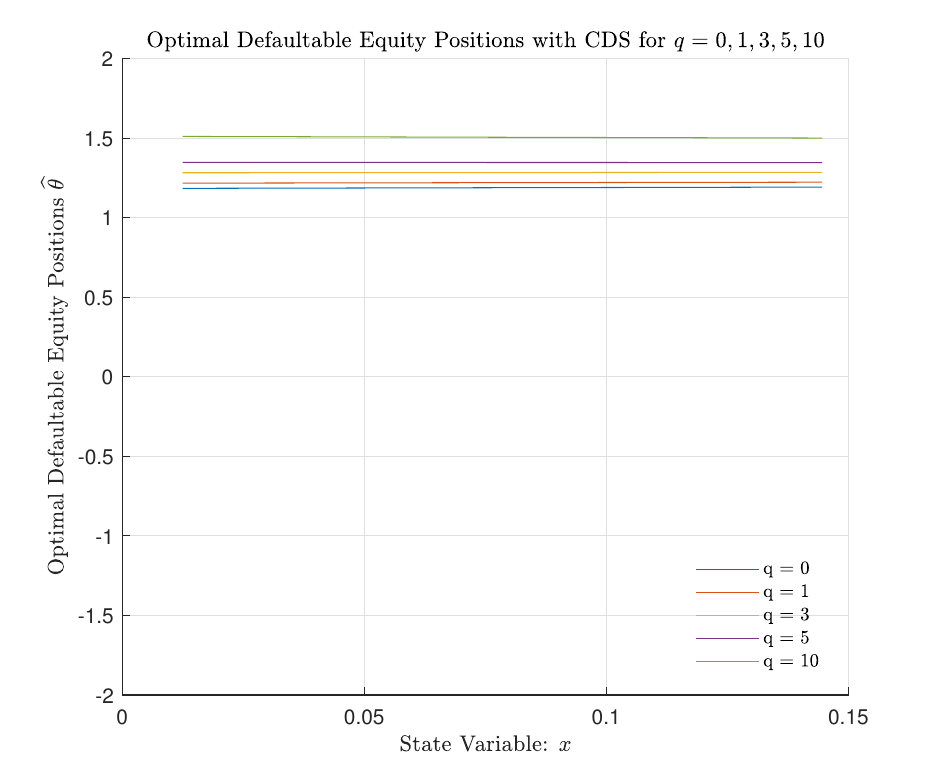}
\end{subfigure}
\hfill
\begin{subfigure}[b]{0.45\textwidth}
  \centering
  \includegraphics[width = \textwidth]{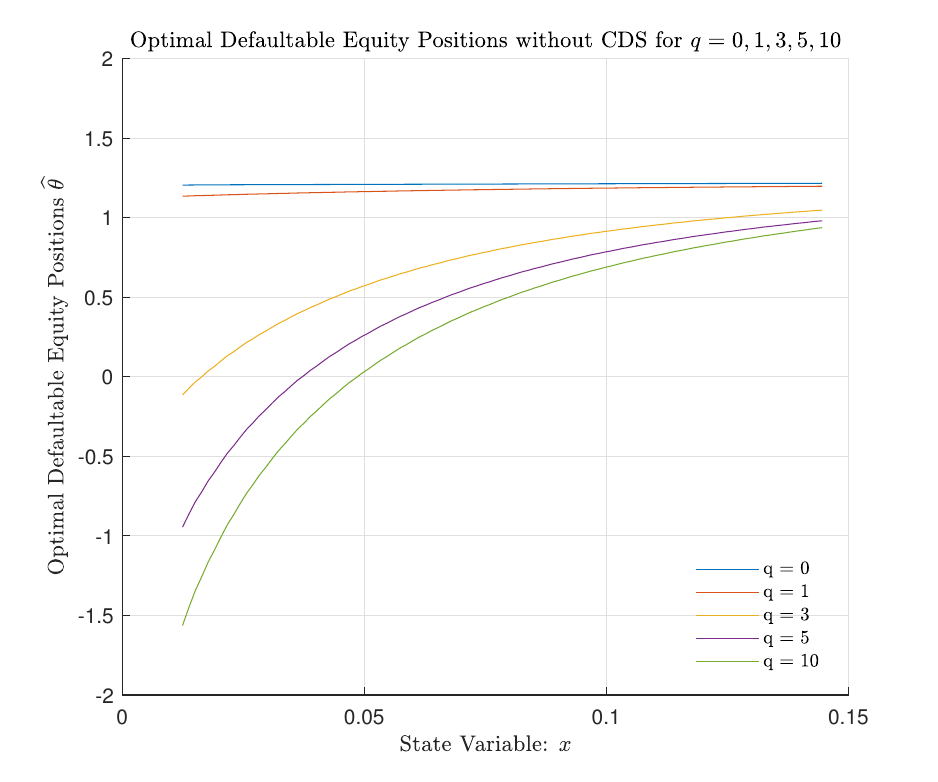}
\end{subfigure}
\caption{Time zero equity positions as a function of the state variable, for different face $q$ in the defaultable bond in complete model of Section \ref{SS:Ex_Comp}. The left plot is in the presence of the CDS market. The right plot is in the absence of the CDS market.}
\label{fig:Optimal Defaultable Equity Positions}
\end{figure}

\begin{figure}[ht]
\centering
\begin{subfigure}[b]{0.45\textwidth}
  \centering
  \includegraphics[width = \textwidth]{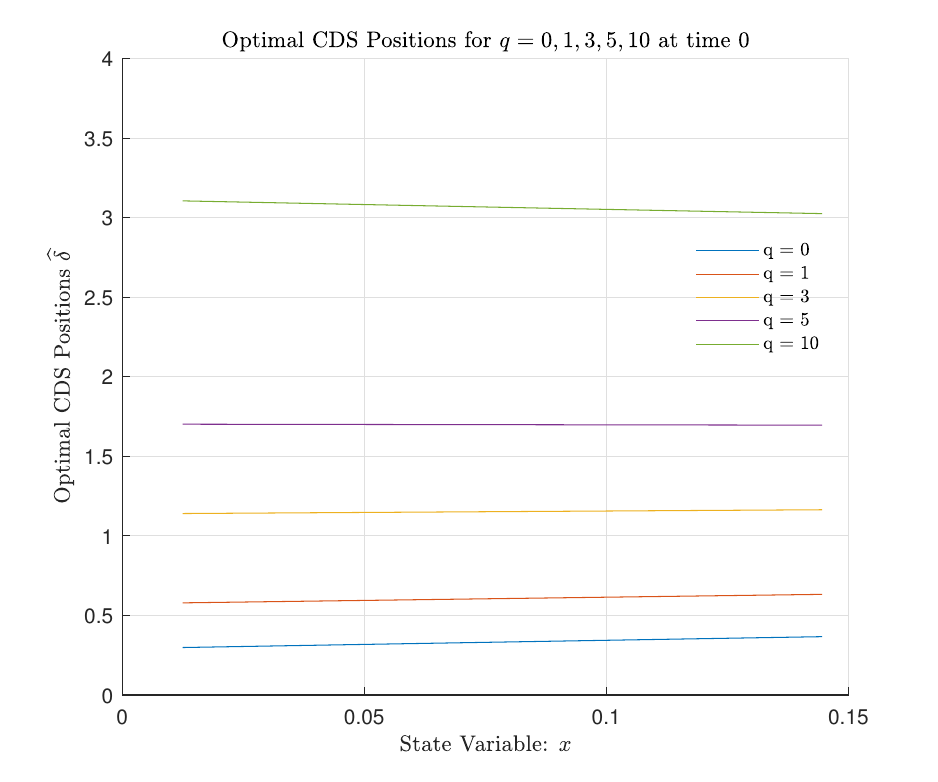}
\end{subfigure}
\hfill
\begin{subfigure}[b]{0.45\textwidth}
  \centering
  \includegraphics[width = \textwidth]{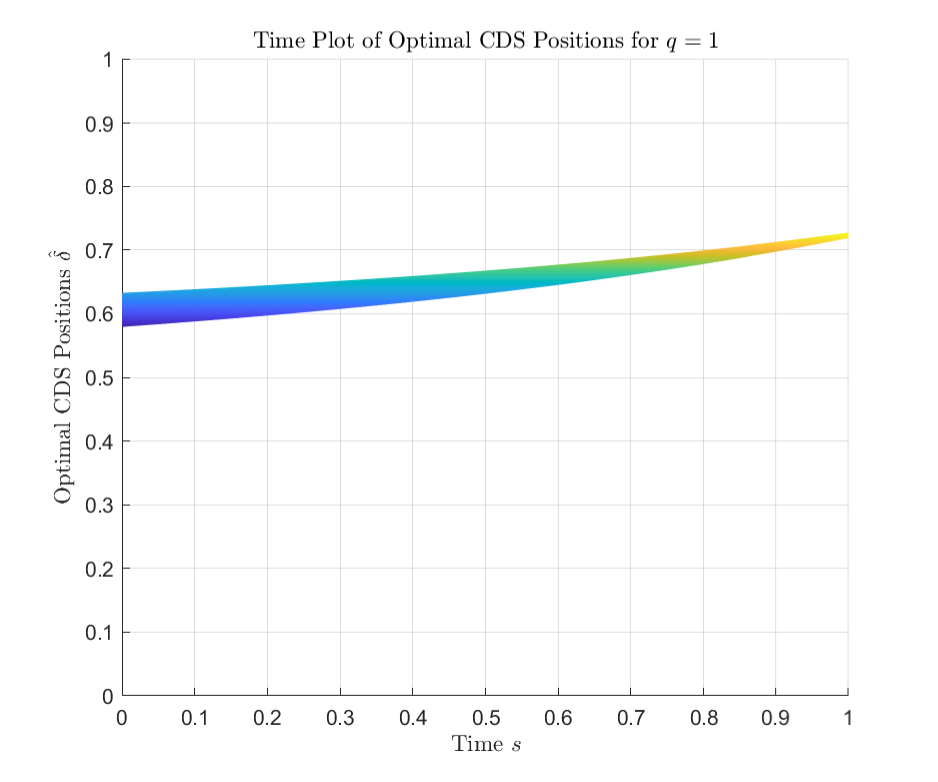}
\end{subfigure}
\caption{CDS positions.  The left plot shows the time $0$ cds position as function of the state variable, for different face $q$ of the defaultable bond in the complete model of Section \ref{SS:Ex_Comp}.  The right plot shows the range (over the state variable) of possible CDS positions as a function of time,  for $q=1$ face if the defaultable bond.}
\label{fig:Optimal CDS Positions}
\end{figure}

\begin{figure}[ht]
\centering
\includegraphics[width = .6\textwidth, height = 6cm]{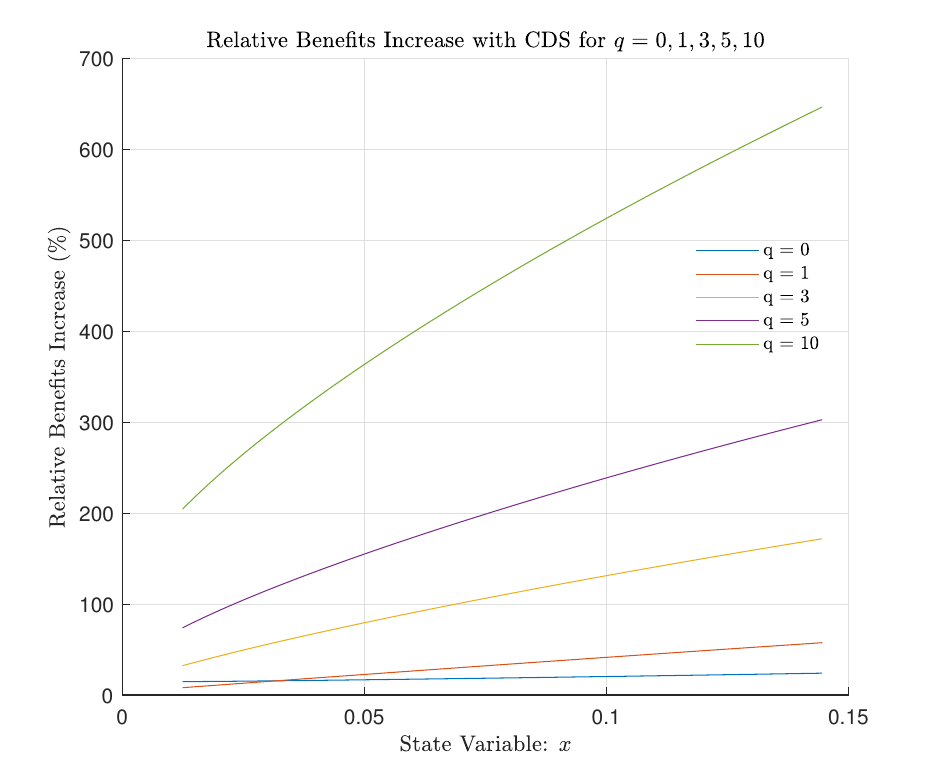}
\caption{The relative benefit $\textrm{CE}^{\textrm{CDS}}/\textrm{CE}^{\textrm{no CDS}} - 1$ at time $0$ as a function of the state variable for different face $q$ in the defaultable bond in the complete model of Section \ref{SS:Ex_Comp}.}
\label{fig:RB}
\end{figure}

\subsection{Incomplete Market}\label{SS:Ex_Incomp} In the incomplete market example there are two equities, and for given $\nu\in\reals^2$ and $\sigma\in \spos{2}$, the equity drift and volatility functions are
\begin{equation*}
    \mu_e(s,y) = y\sigma\nu;\quad \sigma_e(s,y) = \sqrt{y}\sigma.
\end{equation*}
The correlation is a constant $\rho\in\reals^2$ satisfying $\rho' \rho < 1$, and the loss proportion is $\ell_e(s, y)  = \ell(y) e_2$ where $e_2 = (0, 1)'$ and $\ell$ is a function specified below. This implies that at $\tau$ the first asset does not default, while the second does default, suffering the proportional loss $\ell(y)$ if $X_\tau = y$. We use the same values of $\nu$, $\sigma = \sqrt{\Sigma}$ and $\rho$ as in \cite{MR4086602} ,
\begin{equation*}
    \nu = 
    \begin{pmatrix}
        2.235 \\
        3.672
    \end{pmatrix};\quad 
    \Sigma =
    \begin{pmatrix}
        0.277 & 0.310 \\
        0.310 & 0.953
    \end{pmatrix};\quad
    \rho = 
    \begin{pmatrix}
        -0.530 \\
        -0.320
    \end{pmatrix}.    
\end{equation*}


The investor holds $q$ units of a defaultable bond, paying $1$ in the event $\tau > 1$.  As for $\psi$, unlike in the single-asset case, there is a non-zero post-default certainty equivalent from trading, as one may trade the non-defaulted first equity over $[\tau,1]$. To explicitly compute $\psi$, we use \cite {MR2660149,MR2932547} and \cite[Section 4.2]{MR4086602} which shows that $\psi(s, y) = \psi_g(s) + \psi_h(s) y$ for explicitly identifiable functions $\psi_g,\psi_h$.  The default intensities take the same form as in the complete market case
\begin{equation*}
    \gamma(s, y) = y \gamma;\quad \wt{\gamma}(s, y) = y \wt{\gamma}; \quad  \frac{\wt{\gamma}}{\gamma} = 1.5,
\end{equation*}
and the loss proportion is a constant, $\ell(y) = 0.5$. We assume the spot pricing measure $W$ risk-premia function is
\begin{equation*}
    \tprobnu(s, y) = \left(\rho' \sigma_e^{-1}(\mu_e - \wt{\gamma} l_e)\right)(s,y) = \sqrt{y} \tprobnu,\quad \tprobnu = \rho' \nu - \frac{\wt{\gamma}}{2}( \rho' \sigma^{-1} e_2 )\footnote{This corresponds to the minimal martginale measrure ignoring the jump stochastic exponential.}.
\end{equation*}
Under $\tprob$, the dynamics of $X$ becomes
\begin{equation*}
    d X_s = \wt{\kappa}(\wt{\theta} - X_s) d s + \xi \sqrt{X_s} d \wt{W}_s,\quad X_t = x,
\end{equation*}
where $\wt{\kappa} = \kappa + \xi \tprobnu$ and $\wt{\theta} = \kappa\theta/\wt{\kappa}$. As $\wt{\kappa} = 0.0177 > 0$ and Feller's condition holds, $X$ is still a CIR process under $\tprob$ with $X \in \OO$.  The CDS contract expiriy is $\wt{T} = 2$ and $\sigma_r$ is constructed in the same manner as in the complete market case. Lastly, though the verification is lengthy, one can show Assumptions \ref{A:phi_psi_alt} and \ref{A:incomplete_mkt} hold.

Figure \ref{fig:stochastic_UIP} plots the time zero indifference price as a function of the state variable for $q=1,3,5,10$ face of the defaultable bond.  Here, we see the prices display little variation across the position size.  This indicates that even though the market is incomplete, the investor is still able to accurately hedge the defaultable bond, and hence the position size does not significatnly influence the price.   Figure \ref{fig:stochastic_CDS_Positions} also indicates the investor's ability to hedge the defaultable bond, by showing time zero optimal positions in the non-defaultable equity (top-left plot), defaultable equity (top-right plot) and CDS (bottome plot) for different bond notational positions.   From the  figures (note the axis scaling) we see the investor's position in the defaultable equity is rather insensitive to the bond notational, as the notional primarily affects the
CDS position, with a lesser affect on the non-defaultable equity position.  This further indicates hedging is being done through the CDS, and the hedging strategies are mostly insenstive to the default intensity.

\begin{figure}[ht]
\centering
\includegraphics[width = 10cm, height=6cm]{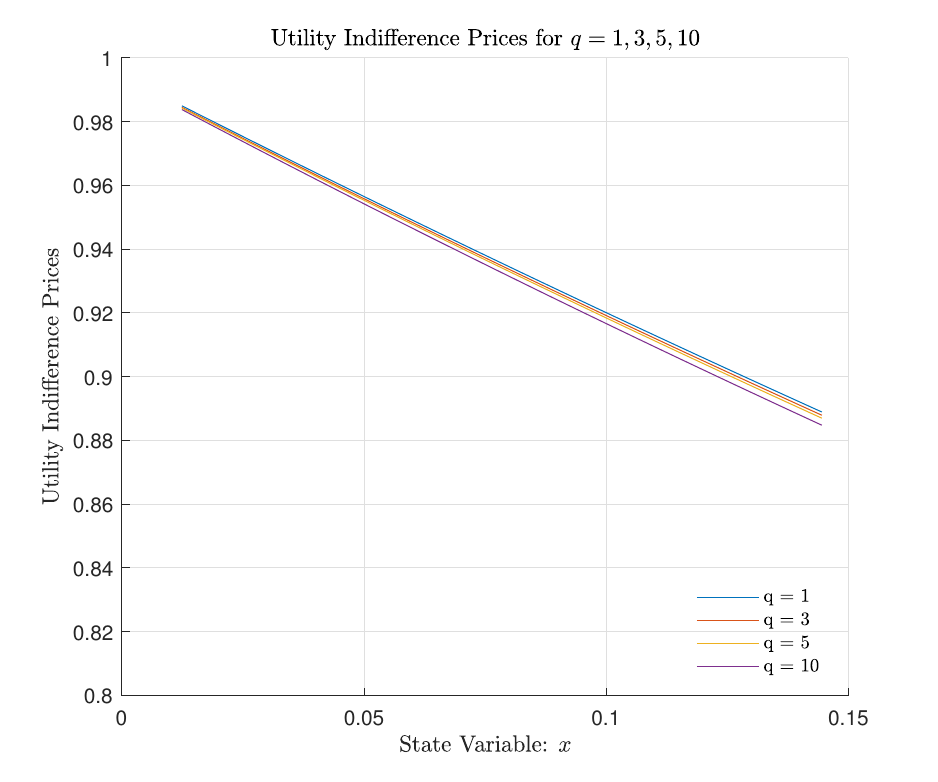}
\caption{Time $0$ utility indifference price as a function of the state variable for different face $q$ in the defaultable bond, in the incomplete model of Section \ref{SS:Ex_Incomp}.}
\label{fig:stochastic_UIP}
\end{figure}

\begin{figure}[ht]
\centering
\begin{subfigure}[b]{0.475\textwidth}
  \centering
  \includegraphics[width = \textwidth]{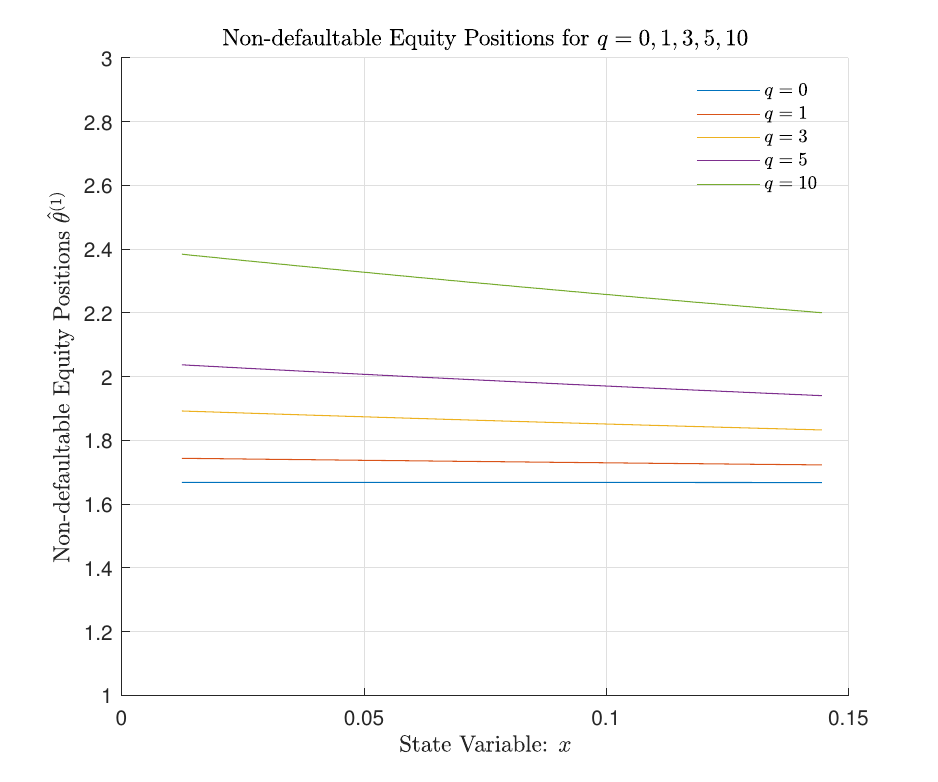}
\end{subfigure}
\hfill
\begin{subfigure}[b]{0.475\textwidth}
  \centering
  \includegraphics[width = \textwidth]{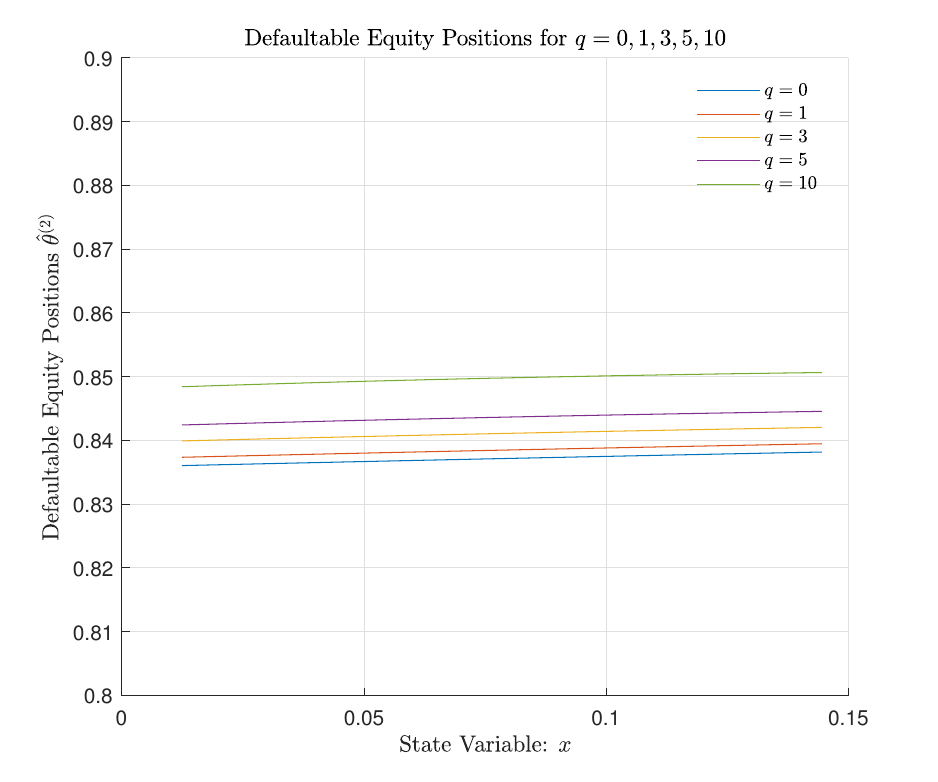}
\end{subfigure}
\vskip\baselineskip
\begin{subfigure}[b]{0.475\textwidth}
  \centering
  \includegraphics[width = \textwidth]{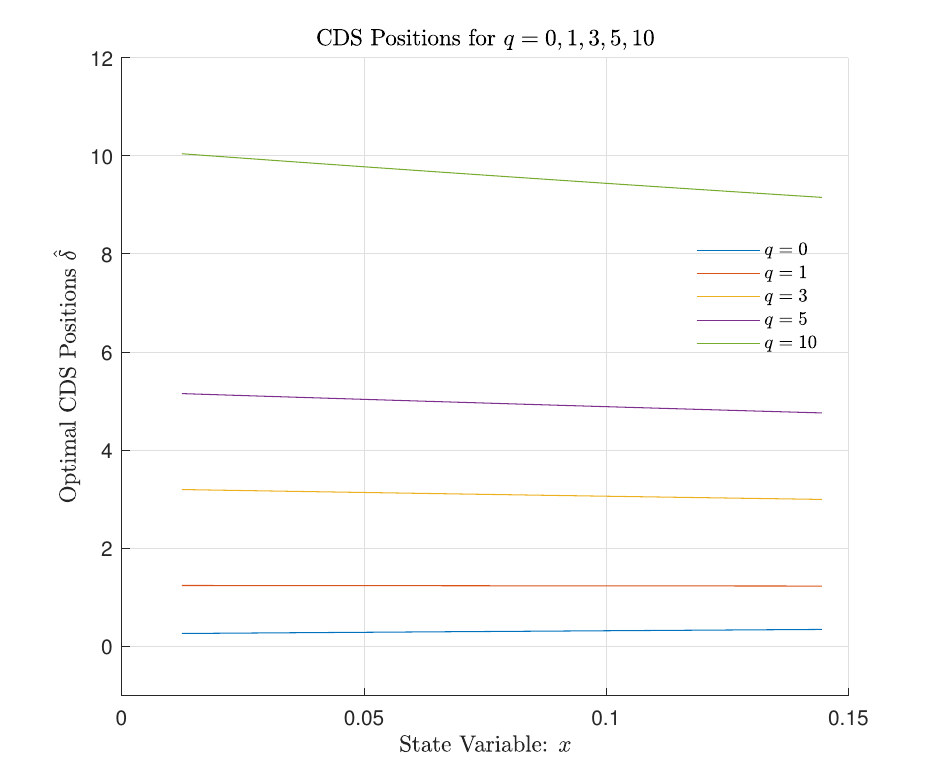}
\end{subfigure}
\caption{Time zero equity and CDS positions as a function of the state variable, for different face $q$ in the defaultable bond in the incomplete model of Section \ref{SS:Ex_Incomp}. The top left plot is for non-defaultable equity, the top right plot is for the defaultable equity, and the bottom plot is for the CDS.}
\label{fig:stochastic_CDS_Positions}
\end{figure}

\section{Conclusion}\label{S:conclusion}

In this paper, we considered the optimal investment problem in a model where the risky assets may default, but where the investor can hedge default risk by dynamically trading in a CDS market.  This updates the setting of \cite{MR4086602} to allow for CDS trading.  Under general conditions, we show that if the equity market absent default is complete, then the equity-CDS market accounting for default is complete as well. Furthermore, we show the optimal CDS position does not just cover equity losses upon default. Rather it additionally covers against losses due to the stoppage of trade.  Numerically, we find that the investor is using the CDS as the primary vehicle to hedge against default, as one would expect.  Furthermore, CDS positions are nearly static, and hence easily implementable in practice. Lastly, we show the investor's indirect utility significantly increases when the CDS market is present, despite the relatively tame strategies.  We hope this work reinforces the use of CDS trading to hegdge against default, even when a dynamic market is not present.

\bibliographystyle{alpha}
\bibliography{reference}

\newcommand{\etalchar}[1]{$^{#1}$}
\begin{thebibliography}{DGR{\etalchar{+}}02}

\bibitem[AR08]{anderson2008equilibrium}
Robert~M Anderson and Roberto~C Raimondo.
\newblock Equilibrium in continuous-time financial markets: Endogenously
  dynamically complete markets.
\newblock {\em Econometrica}, 76(4):841--907, 2008.

\bibitem[BC18]{MR3846288}
Lijun Bo and Agostino Capponi.
\newblock Portfolio {C}hoice with {M}arket--{C}redit-{R}isk {D}ependencies.
\newblock {\em SIAM J. Control Optim.}, 56(4):3050--3091, 2018.

\bibitem[BJ06]{bielecki2006portfolio}
Tomasz~R Bielecki and Inwon Jang.
\newblock Portfolio optimization with a defaultable security.
\newblock {\em Asia-Pacific Financial Markets}, 13:113--127, 2006.

\bibitem[BJR08]{MR2474544}
Tomasz~R. Bielecki, Monique Jeanblanc, and Marek Rutkowski.
\newblock Pricing and trading credit default swaps in a hazard process model.
\newblock {\em Ann. Appl. Probab.}, 18(6):2495--2529, 2008.

\bibitem[BJR09]{bielecki2009credit}
Tomasz~R Bielecki, Monique Jeanblanc, and Marek Rutkowski.
\newblock Credit risk modeling. osaka university csfi lecture notes series,
  2009.

\bibitem[BM01]{MR1846525}
Damiano Brigo and Fabio Mercurio.
\newblock {\em Interest rate models---theory and practice}.
\newblock Springer Finance. Springer-Verlag, Berlin, 2001.

\bibitem[BPT10]{buraschi2010correlation}
A.~Buraschi, P.~Porchia, and F.~Trojani.
\newblock {Correlation risk and optimal portfolio choice}.
\newblock {\em Journal of Finance}, 65(1):393--420, 2010.

\bibitem[BR13]{bielecki2013credit}
Tomasz~R Bielecki and Marek Rutkowski.
\newblock {\em Credit risk: modeling, valuation and hedging}.
\newblock Springer Science \& Business Media, 2013.

\bibitem[BWY10]{MR2779555}
Lijun Bo, Yongjin Wang, and Xuewei Yang.
\newblock An optimal portfolio problem in a defaultable market.
\newblock {\em Adv. in Appl. Probab.}, 42(3):689--705, 2010.

\bibitem[CFL14]{capponi2014dynamic}
Agostino Capponi and Jos{\'e}~E Figueroa-L{\'o}pez.
\newblock Dynamic portfolio optimization with a defaultable security and
  regime-switching.
\newblock {\em Mathematical Finance}, 24(2):207--249, 2014.

\bibitem[CLHH05]{MR2178034}
Netzahualc{\'o}yotl Casta{\~n}eda-Leyva and Daniel Hern{\'a}ndez-Hern{\'a}ndez.
\newblock Optimal consumption-investment problems in incomplete markets with
  stochastic coefficients.
\newblock {\em SIAM J. Control Optim.}, 44(4):1322--1344 (electronic), 2005.

\bibitem[Dab14]{dabadghao2014dynamic}
Shaunak~S Dabadghao.
\newblock {\em Dynamic portfolio optimization with Credit Default Swaps}.
\newblock PhD thesis, Purdue University, 2014.

\bibitem[DGR{\etalchar{+}}02]{MR1891730}
Freddy Delbaen, Peter Grandits, Thorsten Rheinl{\"a}nder, Dominick Samperi,
  Martin Schweizer, and Christophe Stricker.
\newblock Exponential hedging and entropic penalties.
\newblock {\em Math. Finance}, 12(2):99--123, 2002.

\bibitem[DK96]{duffie1996yield}
Darrell Duffie and Rui Kan.
\newblock A yield-factor model of interest rates.
\newblock {\em Mathematical finance}, 6(4):379--406, 1996.

\bibitem[DPS00]{MR1793362}
Darrell Duffie, Jun Pan, and Kenneth Singleton.
\newblock Transform analysis and asset pricing for affine jump-diffusions.
\newblock {\em Econometrica}, 68(6):1343--1376, 2000.

\bibitem[Dur19]{MR3930614}
Rick Durrett.
\newblock {\em Probability---theory and examples}, volume~49 of {\em Cambridge
  Series in Statistical and Probabilistic Mathematics}.
\newblock Cambridge University Press, Cambridge, fifth edition, 2019.

\bibitem[Fri64]{MR0181836}
Avner Friedman.
\newblock {\em Partial differential equations of parabolic type}.
\newblock Prentice-Hall, Inc., Englewood Cliffs, N.J., 1964.

\bibitem[GR12]{MR2932547}
Paolo Guasoni and Scott Robertson.
\newblock Portfolios and risk premia for the long run.
\newblock {\em Ann. Appl. Probab.}, 22(1):239--284, 2012.

\bibitem[HN89]{HN1989}
Stewart~D. Hodges and Anthony Neuberger.
\newblock Optimal replication of contingent claims under transactions costs.
\newblock {\em Review of Futures Markets}, 8:222--239, 1989.

\bibitem[HS00]{Heath-Schweizer}
D.~Heath and M.~Schweizer.
\newblock Martingales versus {PDEs} in finance: an equivalence result with
  examples.
\newblock {\em Journal of Applied Probability}, 37:947--957, 2000.

\bibitem[IR20]{MR4086602}
Tetsuya Ishikawa and Scott Robertson.
\newblock Optimal investment and pricing in the presence of defaults.
\newblock {\em Math. Finance}, 30(2):577--620, 2020.

\bibitem[KMK10]{MR2660149}
Jan Kallsen and Johannes Muhle-Karbe.
\newblock Utility maximization in affine stochastic volatility models.
\newblock {\em Int. J. Theor. Appl. Finance}, 13(3):459--477, 2010.

\bibitem[KO96]{kim1996dnp}
T.S. Kim and E.~Omberg.
\newblock {Dynamic Nonmyopic Portfolio Behavior}.
\newblock {\em Review of Financial Studies}, 9(1):141--161, 1996.

\bibitem[KP14]{MR3131287}
Dmitry Kramkov and Silviu Predoiu.
\newblock Integral representation of martingales motivated by the problem of
  endogenous completeness in financial economics.
\newblock {\em Stochastic Process. Appl.}, 124(1):81--100, 2014.

\bibitem[KS91]{MR1121940}
Ioannis Karatzas and Steven~E. Shreve.
\newblock {\em Brownian motion and stochastic calculus}, volume 113 of {\em
  Graduate Texts in Mathematics}.
\newblock Springer-Verlag, New York, second edition, 1991.

\bibitem[KS06]{MR2206349}
Hidehiro Kaise and Shuenn-Jyi Sheu.
\newblock On the structure of solutions of ergodic type {B}ellman equation
  related to risk-sensitive control.
\newblock {\em Ann. Probab.}, 34(1):284--320, 2006.

\bibitem[Lie96]{MR1465184}
Gary~M. Lieberman.
\newblock {\em Second order parabolic differential equations}.
\newblock World Scientific Publishing Co., Inc., River Edge, NJ, 1996.

\bibitem[Lin06]{MR2212266}
Vadim Linetsky.
\newblock Pricing equity derivatives subject to bankruptcy.
\newblock {\em Math. Finance}, 16(2):255--282, 2006.

\bibitem[Liu07]{liu2007pss}
J.~Liu.
\newblock {Portfolio Selection in Stochastic Environments}.
\newblock {\em Review of Financial Studies}, 20(1):1, 2007.

\bibitem[MZ04]{MR2048829}
Marek Musiela and Thaleia Zariphopoulou.
\newblock An example of indifference prices under exponential preferences.
\newblock {\em Finance Stoch.}, 8(2):229--239, 2004.

\bibitem[O{\v{Z}}09]{MR2489605}
Mark~P. Owen and Gordan {\v{Z}}itkovi{\'c}.
\newblock Optimal investment with an unbounded random endowment and
  utility-based pricing.
\newblock {\em Math. Finance}, 19(1):129--159, 2009.

\bibitem[Pha09]{MR2533355}
Huy{\^e}n Pham.
\newblock {\em Continuous-time stochastic control and optimization with
  financial applications}, volume~61 of {\em Stochastic Modelling and Applied
  Probability}.
\newblock Springer-Verlag, Berlin, 2009.

\bibitem[Pin95]{MR1326606}
Ross~G. Pinsky.
\newblock {\em Positive harmonic functions and diffusion}, volume~45 of {\em
  Cambridge Studies in Advanced Mathematics}.
\newblock Cambridge University Press, Cambridge, 1995.

\bibitem[PS08]{MR2574236}
Philip Protter and Kazuhiro Shimbo.
\newblock No arbitrage and general semimartingales.
\newblock In {\em Markov processes and related topics: a {F}estschrift for
  {T}homas {G}. {K}urtz}, volume~4 of {\em Inst. Math. Stat. (IMS) Collect.},
  pages 267--283. Inst. Math. Statist., Beachwood, OH, 2008.

\bibitem[RY99]{MR1725357}
Daniel Revuz and Marc Yor.
\newblock {\em Continuous martingales and {B}rownian motion}, volume 293 of
  {\em Grundlehren der Mathematischen Wissenschaften [Fundamental Principles of
  Mathematical Sciences]}.
\newblock Springer-Verlag, Berlin, third edition, 1999.

\bibitem[Sch01]{MR1865021}
Walter Schachermayer.
\newblock Optimal investment in incomplete markets when wealth may become
  negative.
\newblock {\em Ann. Appl. Probab.}, 11(3):694--734, 2001.

\bibitem[Sch03]{schonbucher2003credit}
Philipp~J Sch{\"o}nbucher.
\newblock {\em Credit derivatives pricing models: models, pricing and
  implementation}.
\newblock John Wiley \& Sons, 2003.

\bibitem[Sch17]{MR3590708}
Daniel~C. Schwarz.
\newblock Market completion with derivative securities.
\newblock {\em Finance Stoch.}, 21(1):263--284, 2017.

\bibitem[SZ07]{sircar2007utility}
Ronnie Sircar and Thaleia Zariphopoulou.
\newblock Utility valuation of credit derivatives: Single and two-name cases.
\newblock In {\em Advances in Mathematical Finance}, pages 279--301. Springer,
  2007.

\bibitem[Wac02]{wachter2002pac}
J.~Wachter.
\newblock {Portfolio and consumption decisions under mean-reverting returns: An
  exact solution for complete markets}.
\newblock {\em Journal of Financial and Quantitative Analysis}, 37(1):63--91,
  2002.

\end{thebibliography}

\appendix 

\section{HJB Derivation}\label{AS:HJB}

Here we informally derive \eqref{E:G_HJB}, but we stress that we will use other methods to verify solutions to \eqref{E:G_HJB} are the certainty equivalent. Fix $x\in \OO$ and $t\leq T$, and  for $\pi \in \A$ from \eqref{E:admiss} define
\begin{equation*}
    J(\pi) \dfn \expvs{U\left(\We^{\pi}_T + \phi(X_T)1_{\tau > T} + \psi(\tau,X_{\tau})1_{\tau\leq T}\right)}.
\end{equation*}
Using the jump in wealth at default we obtain
\begin{equation*}
    J(\pi) = \expvs{U\left(\We^{\pi}_{(T\wedge\tau)-} - \pi_{\tau}'\ell(\tau,X_{\tau})1_{\tau\leq T} + \phi(X_T)1_{\tau > T} + \psi(\tau,X_{\tau})1_{\tau\leq T}\right)}.
\end{equation*}
As $\pi\in\mcp(\filtg)$, we know that both $\pi$ and the left limit process associated to $\We^{\pi}$ coincide with $\filt^{W,B}$ predictable processes on $[0,\tau]$, and hence we use the intensity function $\gamma$ to obtain
\begin{equation*}
    \begin{split}
        J(\pi) &= \mathbb{E}\bigg[\int_{t}^{\infty} U\bigg(\We^{\pi}_{(T\wedge u)-} - \pi_{u}'\ell(u,X_u) 1_{u\leq T} + \phi(X_T)1_{u > T} + \psi(u,X_u)1_{u\leq T}\bigg)\\
        &\qquad\qquad \times \gamma(u,X_u) e^{-\int_0^u \gamma(v,X_v) dv}du\bigg],\\
        &=\mathbb{E}\bigg[\int_t^T U\left(W^{\pi}_u - \pi_u'\ell(u,X_u)  + \psi(u,X_u)\right)\gamma(u,X_u) e^{-\int_t^u \gamma(v,X_v) dv}du\\
        &\qquad\qquad + e^{-\int_t^T \gamma(v,X_v) dv}U\left(W^{\pi}_{T} + \phi(X_T)\right)\bigg].
    \end{split}
\end{equation*}
Using again that $\pi$ coincides with an $\filt^{W,B}$ predictable process prior to default, and writing $X= X^{t,x}$, we are left with the control problem of identifying $v(t,x,w) \dfn \sup_{\pi \in \mcp(\filt)} \wt{J}(t,x,w;\pi)$ where
\begin{equation*}
    \begin{split}
        \wt{J}(t,x,w,\pi) &= \mathbb{E}\bigg[\int_t^T U\left(W^{\pi}_u - \pi_u'\ell(u,X^{t,x}_u)  + \psi(u,X^{t,x}_u)\right)\gamma(u,X^{t,x}_u) e^{-\int_t^u \gamma(v,X^{t,x}_v) dv}du\\
        &\qquad\qquad + e^{-\int_t^T \gamma(v,X^{t,x}_v) dv}U\left(W^{\pi}_{T} + \phi(X^{t,x}_T)\right) \bigg| \We^{\pi}_t = w\bigg].
    \end{split}
\end{equation*}
This is a standard control problem, and from (for example) \cite{MR2533355} we obtain the HJB equation for $v$, suppressing function arguments,
\begin{equation*}
    0  = v_t + \max_{\pi}\bra{-\gamma v + \pi'\left(v_w \mu +  v_w \Upsilon \nabla\right) + \frac{1}{2}v_{ww} \pi'\Sigma \pi + Lv - \gamma e^{-\alpha(w-\pi'\ell + \psi)}},
\end{equation*}
with boundary condition $v(T,x,w) = e^{-\alpha(w + \phi(x))}$.  Writing $v(t,x,w) = -e^{-\alpha(w+G(t,x))}$ and simplifying gives the HJB equation \eqref{E:G_HJB} for $G$.

\section{Proofs from Section \ref{S:opt_invest}}\label{AS:opt_invest}

\begin{proof}[Proof of Proposition \ref{P:cds_dynamics}]

Recall the extended generator $\wt{L}$ in \eqref{E:wtL_def} for $X$ under $\tprob$. By \cite[Theorem 1, Lemma 2]{Heath-Schweizer} and our given assumptions, we know that $\wt{u}, \wt{v}$ respectively solve the PDEs
\begin{equation*}
    \begin{split}
        0 & = \wt{u}_s + \wt{L}\wt{u} -\wt{\gamma}(\wt{u}-1),\hspace{30pt} \wt{u}(T,\cdot) = 0,\\
        1 &= \wt{v}_s + \wt{L}\wt{v} - \wt{\gamma}\wt{v},\hspace{60pt} \wt{v}(T,\cdot) = 0.
    \end{split}
\end{equation*}
Next, using the notation of \cite{MR2474544} we have $r_s \equiv 0$, $B_s \equiv 1$, $\delta_s \equiv 1$, $M_s = H_s - \int_t^{\tau \wedge s} \wt{\gamma}(u,X_u)du$, $G_s = e^{-\int_t^s \wt{\gamma}(u,X_u)}$ and $\kappa(s,T) = \wt{u}(s,X_s)/\wt{v}(s,X_s)$.  Additionally,
\begin{equation*}
    m^1_s = 1 - G_s(1-\wt{u}_s);\qquad  m^2_s = G_s \wt{v}_s + \int_t^s G_u du.
\end{equation*}
The dynamics in \eqref{E:cds_dynamics} now follow from \cite[Lemma 2.4]{MR2474544} by direct computations.

\end{proof}


\section{Proofs from Section \ref{S:HJB_CE}}\label{AS:HJB_CE}

We begin with Proposition \ref{P:general_hamil_11}.  Throughout, we suppress $(t,x)$ keeping only the dependence upon $(p,g)$ and $\sigma = \sigma_r(t,x)$ explicit. We also write $\pi =(\theta,\delta)$. As identification of the optimal equity policy function $\wh{\theta}$ for fixed $(p,g,\delta)$ is the same under both Assumptions \ref{A:complete_mkt}, \ref{A:incomplete_mkt}, we start with this identification. To ease the notational burden we  define
\begin{equation}\label{E:ACE_def}
    \begin{split}
        \Afn(\sigma) &\dfn \sigma'\bAfn\sigma;\hspace{20pt}\Bfn(p,\sigma) \dfn -\wt{\gamma} + \sigma'\bBfn(p);\hspace{20pt} \Cfn \dfn \ell_e'\Sigma_e^{-1}\ell_e;\\
        \Dfn(g,p) &\dfn \Cfn \times \gamma e^{\alpha(g-\psi) + \ell_e'\Sigma_e^{-1}\left(\mu_e - \alpha\Upsilon_e p\right)}; \hspace{20pt} \Efn(\sigma) \dfn 1 + \sigma'\bEfn,
    \end{split}
\end{equation}
where
\begin{equation}\label{E:BD_def}
    \begin{split} 
         \bAfn &\dfn A - \Upsilon_e' \Sigma_e^{-1} \Upsilon_e;\hspace{10pt}\bBfn(p) \dfn a\tprobnu - \Upsilon_e' \Sigma_e^{-1} \mu_e - \alpha \bAfn p\hspace{10pt} \bEfn \dfn \Upsilon_e' \Sigma_e^{-1} \ell_e.
    \end{split}
\end{equation}
Note from \eqref{E:SigE_UpsE} and \eqref{E:E_barA_def} that $\Efn(\sigma_r) = \Efnalt$ when $k=d,\rho = \idmat{d}$. Now, fix $\delta$ and consider maximizing the right side of \eqref{E:hamil} over $\theta$. Using \eqref{E:SigE_UpsE} and noting that $\Sigma_e$ is invertible, direct computation yields the optimal equity policy function
\begin{equation}\label{E:opt_theta_1}
    \wh{\theta}(p,g,\delta) = \frac{1}{\alpha}\Sigma_e^{-1}\left(\mu_e -\alpha \Upsilon_e p - \frac{1}{\Cfn}\plog{\Dfn(g,p) e^{-\alpha\Efn(\sigma_r) \delta}}\ell_e - \alpha\delta \Upsilon_e \sigma_r\right).
\end{equation}
Plugging this into the right side of \eqref{E:hamil} and simplifying yields (recall \eqref{E:MF_hamil})
\begin{equation}\label{E:mid_hcal}
    \begin{split}
        \Hcal(g,p) &= \frac{1}{2\alpha}\left(\mu_e - \alpha\Upsilon_e p\right)'\Sigma_e^{-1}\left(\mu_e - \alpha\Upsilon_e p\right)\\
        &\qquad + \sup_{\delta}\bigg( -\frac{1}{2}\alpha \Afn(\sigma_r) \delta^2 + \Bfn(p,\sigma_r) \delta - \frac{1}{2\alpha\Cfn}\bigg(\plog{\Dfn(g,p) e^{-\alpha\Efn(\sigma_r)\delta}}^2\\
        &\qquad\qquad + 2\plog{\Dfn(g,p) e^{-\alpha \Efn(\sigma_r) \delta}}\bigg)\bigg).
    \end{split}
\end{equation} 

\subsection{Proof of Proposition \ref{P:general_hamil_11}}

When Assumption \ref{A:complete_mkt} holds, $\Afn(\sigma_r) =0$ and $\Bfn(\sigma,p) = -\wt{\gamma}\Efn(\sigma)$ (because $\tprobnu = \sigma_e^{-1}(\mu_e-\wt{\gamma}\ell_e)$).  This implies we are trying to maximize
\begin{equation*}
     -\wt{\gamma}\Efn(\sigma_r) \delta - \frac{1}{2\alpha\Cfn}\left(\plog{\Dfn(g,p) e^{-\alpha\Efn(\sigma_r)\delta}}^2 + 2\plog{\Dfn(g,p) e^{-\alpha \Efn(\sigma_r) \delta}}\right).
\end{equation*}
Using \eqref{E:plog_deriv} below, direct computations give the first order conditions and  solution
\begin{equation}\label{E:p1_foc_opt_delta}
    \wt{\gamma}\Cfn = \plog{\Dfn(g,p) e^{-\alpha\Efn(\sigma_r)\wh{\delta}}} \quad \Longrightarrow \quad \wh{\delta} = \frac{1}{\alpha\Efn(\sigma_r)}\left(-\wt{\gamma}\Cfn  + \log\left(\frac{\Dfn(g,p)}{\wt{\gamma}\Cfn}\right)\right),
\end{equation}
which is well defined because $\Efn(\sigma_r)\neq 0$.  Using $\wh{\delta}$ in \eqref{E:opt_theta_1} yields the formula for $\wh{\theta}$ in \eqref{E:gen_opt_strat_1}, and plugging in for $\Cfn$ and $\Dfn(g,p)$ yields the formula for $\wh{\delta}$ in \eqref{E:gen_opt_strat_1}. Continuing, plugging for $\wh{\delta}$ in \eqref{E:mid_hcal} gives
\begin{equation*}
    \begin{split}
        \Hcal(g,p) &= \frac{1}{2\alpha}\left(\mu_e - \alpha\Upsilon_e p\right)'\Sigma_e^{-1}\left(\mu_e - \alpha\Upsilon_e p\right)+ \frac{\wt{\gamma}^2 \Cfn}{2\alpha} - \frac{\wt{\gamma}}{\alpha} - \frac{\wt{\gamma}}{\alpha}\log\left(\frac{\Dfn(g,p)}{\wt{\gamma}\Cfn}\right).
    \end{split}
\end{equation*}
Plugging in again for $\Cfn,\Dfn(g,p)$ and using \eqref{E:Qval_complete} gives \eqref{E:Hcal_is_ok_1}.  The PDE \eqref{E:PDE_1} follows from \eqref{E:G_HJB} using $\wt{L}$ from \eqref{E:wtL_def} and \eqref{E:SigE_UpsE} at $\rho = \idmat{d}$ which shows $\Upsilon_e'\Sigma_e^{-1}\Upsilon_e = A$.

\subsection{Proof of Proposition \ref{P:general_hamil_2}}
We use the notation of \eqref{E:ACE_def}, \eqref{E:BD_def} as well as the optimal equity policy function in \eqref{E:opt_theta_1} for fixed $\delta$ and the Hamiltonian in \eqref{E:mid_hcal}.  Define the residual function
\begin{equation}\label{E:GR_def_new}
    \begin{split}
        R_{\Hcal}(\sigma,g,p) &\dfn \sup_{\delta}\bigg( -\frac{1}{2}\alpha \Afn(\sigma) \delta^2 + \Bfn(p,\sigma) \delta - \frac{1}{2\alpha\Cfn}\bigg(\plog{\Dfn(g,p) e^{-\alpha\Efn(\sigma)\delta}}^2\\
        &\qquad\qquad + 2\plog{\Dfn(g,p) e^{-\alpha \Efn(\sigma) \delta}}\bigg)\bigg) - \frac{\wt{\gamma}^2 \Cfn}{2\alpha} + \frac{\wt{\gamma}}{\alpha} + \frac{\wt{\gamma}}{\alpha}\log\left(\frac{\Dfn(g,p)}{\wt{\gamma}\Cfn}\right).
    \end{split}
\end{equation}
From \eqref{E:mid_hcal} we see
\begin{equation*}
    \begin{split}
        \Hcal(g,p) = \frac{1}{2\alpha}\left(\mu_e - \alpha\Upsilon_e p\right)'\Sigma_e^{-1}\left(\mu_e - \alpha\Upsilon_e p\right) + \frac{\wt{\gamma}^2 \Cfn}{2\alpha} - \frac{\wt{\gamma}}{\alpha} - \frac{\wt{\gamma}}{\alpha}\log\left(\frac{\Dfn(g,p)}{\wt{\gamma}\Cfn}\right) + R_{\Hcal}(\sigma,g,p).
    \end{split}
\end{equation*}
Using \eqref{E:Qval_complete}, a direct computation shows 
\begin{equation*}
    \begin{split}
        &\frac{1}{2\alpha}\left(\mu_e - \alpha\Upsilon_e p\right)'\Sigma_e^{-1}\left(\mu_e - \alpha\Upsilon_e p\right) + \frac{\wt{\gamma}^2 \Cfn}{2\alpha} - \frac{\wt{\gamma}}{\alpha} - \frac{\wt{\gamma}}{\alpha}\log\left(\frac{\Dfn(g,p)}{\wt{\gamma}\Cfn}\right)\\
        &\qquad = \frac{1}{\alpha}(Q_c-\gamma) + \wt{\gamma}(\psi-g)  + \frac{\alpha}{2}p'\Upsilon_e'\Sigma_e^{-1}\Upsilon_e p - p'\Upsilon_e'\Sigma_e^{-1}(\mu_e-\wt{\gamma}\ell_e).
    \end{split}
\end{equation*}
Therefore, the PDE \eqref{E:PDE_2} is obtained by substituting in for $\Hcal(G,\nabla G)$, and noting $\Upsilon_e'\Sigma_e^{-1}\Upsilon_e = a\rho'\rho a'$ and $\Upsilon_e'\Sigma_e^{-1} = a\rho'\sigma_e^{-1}$. We now show $R_{\Hcal}(0_d,p,g) = 0$. If $\sigma_r = 0_d$ then $\Afn(0_d) = 0, \Bfn(p, 0_d) = -\wt{\gamma}, \Efn(0_d) = 1$ and we solve
\begin{equation}\label{E:sup_delt_sigr_0}
    \sup_{\delta}\left( - \wt{\gamma}\delta - \frac{1}{2\alpha\Cfn}\left(\plog{\Dfn(g,p) e^{-\alpha\delta}}^2 + 2\plog{\Dfn(g,p) e^{-\alpha \delta}}\right)\right).
\end{equation}
Similarly to \eqref{E:p1_foc_opt_delta}, the first order conditions and optimizer are
\begin{equation}\label{E:opt_delta_sigr_0}
    \wt{\gamma}\Cfn = \plog{\Dfn(g,p) e^{-\alpha\wh{\delta}}} \quad \Longrightarrow \wh{\delta} = \frac{1}{\alpha}\left(-\wt{\gamma}\Cfn  + \log\left(\frac{\Dfn(g,p)}{\wt{\gamma}\Cfn}\right)\right).
\end{equation}
Plugging in for $\wh{\delta}$ in \eqref{E:sup_delt_sigr_0} gives the value
\begin{equation*}
    \frac{\wt{\gamma}^2 \Cfn,}{2\alpha} - \frac{\wt{\gamma}}{\alpha} - \frac{\wt{\gamma}}{\alpha}\log\left(\frac{\Dfn(g,p)}{\wt{\gamma}\Cfn}\right),
\end{equation*}
so that from \eqref{E:GR_def_new} we see $R_{\Hcal}(0_d,g,p) = 0$. Note from \eqref{E:opt_theta_1} and \eqref{E:opt_delta_sigr_0} we also obtain the formula for $\wh{\theta}_i$ in \eqref{E:gen_opt_strat_2} when $\sigma_r = 0$.  Next, consider when $\abs{\sigma_r}>0$.  Using \eqref{E:plog_deriv} below the first order conditions in \eqref{E:GR_def_new} are
\begin{equation*}
    0 = -\alpha\Afn(\sigma_r) \wh{\delta} + \Bfn(p,\sigma_r) + \frac{\Efn(\sigma_r)}{\Cfn}\plog{\Dfn(g,p) e^{-\alpha\Efn(\sigma_r)\wh{\delta}}}. 
\end{equation*}
This has unique solution
\begin{equation}\label{E:sigmaR_not_0_foc}
    \wh{\delta}(g,p) = \frac{\Bfn(p, \sigma_r)}{\alpha\Afn(\sigma_r)} +\frac{\Efn(\sigma_r)}{\alpha\Afn(\sigma_r)\mbf{K}_6(\sigma_r)}\plog{\frac{\Dfn(g,p)\mbf{K}_6(\sigma_r)}{\Cfn} e^{-\frac{\Efn(\sigma_r)\Bfn(p,\sigma_r)}{\Afn(\sigma_r)}}}.
\end{equation}
where
\begin{equation*}
\mbf{K}_6(\sigma) = \frac{\Afn(\sigma)\Cfn + \Efn(\sigma)^2}{\Afn(\sigma)}
\end{equation*}
In view of the equation for $\wh{\delta}$ in \eqref{E:opt_delta_sigr_0} we can write
\begin{equation}\label{E:wgdel_sigr_pos}
    \wh{\delta}(g,p) = \frac{1}{\alpha}\left(-\wt{\gamma}\Cfn  + \log\left(\frac{\Dfn(g,p)}{\wt{\gamma}\Cfn}\right)\right) + \wh{R}_{\delta}(\sigma_r,g,p), 
\end{equation}
where 
\begin{equation}\label{E:Rdelta_def}
    \begin{split}
        &\wh{R}_{\delta}(\sigma,g,p)  \dfn \frac{1}{\alpha}\left(\wt{\gamma}\Cfn - \log\left(\frac{\Dfn(g,p)}{\wt{\gamma}\Cfn}\right)\right) + \frac{\Bfn(p,\sigma)}{\alpha\Afn(\sigma)}\\
        &\qquad + \frac{\Efn(\sigma)}{\alpha\Afn(\sigma)\mbf{K}_6(\sigma)}\plog{\frac{\Dfn(g,p)\mbf{K}_6(\sigma)}{\Cfn} e^{-\frac{\Efn(\sigma_r)\Bfn(p,\sigma_r)}{\Afn(\sigma_r)}}}.
    \end{split}
\end{equation}
Plugging in for $\Cfn,\Dfn(g,p)$ we obtain the formula for $\wh{\delta}_i$ in \eqref{E:gen_opt_strat_2}, and the formula for $\wh{\theta}_i$ in \eqref{E:gen_opt_strat_2} follows using the first order condition \eqref{E:sigmaR_not_0_foc}. Next, plugging in $\wh{\delta}$ from \eqref{E:sigmaR_not_0_foc} into \eqref{E:GR_def_new} and simplifying leaves the explicit formula for $R_{\Hcal}$
\begin{equation}\label{E:GR_def}
    \begin{split}
        R_{\Hcal}(\sigma,g,p) &\dfn  - \frac{\wt{\gamma}^2\Cfn}{2\alpha} + \frac{\wt{\gamma}}{\alpha} + \frac{\wt{\gamma}}{\alpha}\log\left(\frac{\Dfn(g,p)}{\wt{\gamma}\Cfn}\right) + \frac{\Bfn(p,\sigma)^2}{2\alpha \Afn(\sigma)}\\
        &\qquad - \frac{1}{2\alpha\mbf{K}_6(\sigma)}\bigg( \plog{\frac{\Dfn(g,p)\mbf{K}_6(\sigma)}{\Cfn} e^{-\frac{\Efn(\sigma_r)\Bfn(p,\sigma_r)}{\Afn(\sigma_r)}}}^2 + 2\plog{\dots}\bigg)
    \end{split}
\end{equation}
where the second product log function is evaluated at the same argument as the first. Now, the formula for $\wh{\delta}$ in \eqref{E:opt_delta_sigr_0} was defined when $\abs{\sigma_r} = 0$, and the formula for $\wh{\delta}$ in \eqref{E:wgdel_sigr_pos} (and hence $\wh{R}_{\delta}$ in \eqref{E:Rdelta_def}) was defined when $\abs{\sigma_{r}} >0$.  If we combine the two cases we obtain
\begin{equation*}
    \wh{\delta}(g,p) =  \frac{1}{\alpha}\left(-\wt{\gamma}\Cfn  + \log\left(\frac{\Dfn(g,p)}{\wt{\gamma}\Cfn}\right)\right) + R_{\delta}(\sigma_r,g,p),
\end{equation*}
where for all $\sigma$ we define
\begin{equation}\label{E:Rdelta_def_new}
R_{\delta}(\sigma,g,p) = 1_{\abs{\sigma}>0} \wh{R}_{\delta}(\sigma,g,p).
\end{equation}
This  clearly shows $R_{\delta}(0_d,g,p) = 0$\footnote{In fact, one can show $\wh{R}_\delta$ may be continuously extended to $0_d$ by setting $\wh{R}_{\delta}(0_d,g,p) = 0$, but as in Assumption \ref{A:incomplete_mkt} we either assume $\sigma_r \equiv 0_d$ or $\abs{\sigma}_r > 0$ continuity as $0$ is not needed for our main results to go through.} and finishes the result.


\section{Properties of the Product Log Function}\label{AS:PL}

In this section we collect facts about the product log function $\plog{}$, defined as the inverse of $y e^y$ on $(-1,\infty)$. First, direct calculation using $\plog{z}e^{\plog{z}} = y$ shows that
\begin{equation}\label{E:plog_deriv}
\partial_z \plog{z} = \frac{\plog{z}}{z(1+\plog{z})}.
\end{equation}

Next, we state a lemma  used to provide gradient estimates when solving the PDE in \eqref{E:PDE_2}.

\begin{lemma}\label{L:plog_quad}
For $z\in\reals$ and $K>0$ we have $\plog{Ke^z}^2 + 2\plog{Ke^z}\leq 2K + z^2$ with equality if and only if $K=z$.
\end{lemma}

\begin{proof}[Proof of Lemma \ref{L:plog_quad}]

Set $h(K,z) = 2k + z^2 - \plog{Ke^z}^2 - 2\plog{Ke^z}$. From \eqref{E:plog_deriv} we see that $\partial_z h(K,z) = 2(z - \plog{Ke^z})$ and $\partial_{zz} h(K,z) = 2/(1+\plog{Ke^z}) > 0$. Thus, for $K$ fixed, $h(K,z)$ has a unique minimum (in $z$) at $z = \plog{Ke^z}$ or equivalently when $z=K$, and when $z=K$ we have $h(K,z) =0$.

\end{proof}


\section{Proof of Theorem \ref{T:main_result}}\label{AS:main_result}

Analogous to the proof of \cite[Theorem 2.10]{MR4086602}, the outline for proving Theorem \ref{T:main_result} is
\begin{enumerate}
    \item Under Assumption \ref{A:complete_mkt} (respectively \ref{A:incomplete_mkt}), identify a localized version of the PDE \eqref{E:PDE_1} (resp. \eqref{E:PDE_2}),  and use classic results on quasi-linear PDEs to prove existence of solutions. To treat \eqref{E:PDE_1} and \eqref{E:PDE_2} in a unified manner, we use \eqref{E:G_HJB} to express both as
    \begin{equation}\label{E:G_HJB_proof}
        0 = G_t + LG - \frac{\alpha}{2}\nabla G'A \nabla G + \frac{\gamma}{\alpha} + \Hcal(G,\nabla G),\qquad \phi = G(T,\cdot),
    \end{equation}
    where under Assumption \ref{A:complete_mkt} we take $\Hcal = \Hcal_c$ from \eqref{E:Hcal_is_ok_1}, and under Assumption \ref{A:incomplete_mkt} we take $\Hcal = \Hcal_i$ from \eqref{E:Hcal_is_ok_2}.

    \item Define a localized optimal investment problem for which the localized PDE is verified to be the certainty equivalent. 
    
    \item Obtain a global solution to the PDE by showing local uniformly boundedness of the local PDEs.

    \item Show that global solution is the certainty equivalent to the global optimal investment problem. Identify the optimal trading strategies in equity and CDS market, and equivalent local martingale measure.
\end{enumerate}

Below, several of the steps almost exactly follow those in \cite{MR4086602}, but several other steps require different proofs or approaches.  To simplify the presentation, all references to \cite{MR4086602} are given in italics, and wherever possible, we will describe how one may use a corresponding result in \cite{MR4086602}.

\subsection*{Mollifiers and H\"{o}lder  Space}

The mollifiers are denoted $\cbra{\chi_n}$. They are precisely constructed in \textit{Appendix A.1}, and for each $n$, $\chi_n \in C^{\infty}(\OO; [0,1])$ with $\chi_n = 1$ on $\OO_{n - 1}$, $\chi_n = 0$ on $\ol{\OO}_n^c$,  and  $\chi_n > 0$ on $\OO_n$. For every $0 < \eta \in \reals$, there exists a unique decomposition $\eta = k + \beta$ such that $k$ is a nonnegative integer and $\beta \in(0, 1]$. The H\"{o}lder spaces $H_{\eta,Q}$ and $H_{\eta,Q,loc}$ for a given region $Q$ in $\reals^{1+d}$ are defined exactly as $H_{k + \beta,Q}$ and $H_{k + \beta,Q,loc}$ respectively in \textit{Appendix A.1}. For $\eta = 1$ and bounded domain $K \in \reals^{1 + d}$, the space $H_{1, K}$ consists of functions that are Lipschitz continuous in $K$. In particular, $C^{(1, 1)}(K; \reals) \subset H_{1, K} \subset H_{\beta', K}$ for all $\beta' \in (0, 1)$ and bounded domains $K \in \reals^{1 + d}$.

\subsection*{Localized PDE} In view of \eqref{E:G_HJB_proof}, we introduce the following localized PDE defined on $[0, T] \times \OO_n$, analogous to the one in \textit{Appendix A.2}. 
\begin{equation}\label{E:G_HJB_local}
    \begin{split}
        0 &= G^n_s + L G^n -\frac{\alpha}{2}\nabla (G^n)' A \nabla (G^n) + \chi_n \left( \frac{\gamma}{\alpha} + \Hcal(G^n, \nabla G^n )\right),\quad \chi_n \phi = G^n(T,\cdot).
    \end{split}
\end{equation}
In accordance with \cite{MR1465184}, we adjust \eqref{E:G_HJB_local} with the following notations. First, define $v^n(s, x) \dfn G^n(T - s, x)$ and $\Omega_n \dfn (0, T - t) \times \OO_n$. Next, set $\Gamma_n \dfn B \Omega_n \cup S\Omega_n \cup \C \Omega_n$ as the parabolic boundary \footnote{See \cite[Section II.1]{MR1465184} for a definition of parabolic boundary.} of $\Omega_n$ where 
\begin{equation*}
    B \Omega_n \dfn \{0\} \times \OO_n;\quad S \Omega_n \dfn (0, T - t) \times \partial \OO_n;\quad C \Omega_n = \{0\} \times \partial \OO_n. 
\end{equation*}
The definition of $\Gamma_n \in H_{\eta}$ for some $\eta \geq 1$ can be found in \cite[Section IV.7]{MR1465184}. Finally, define $\varphi^n(s, x) \dfn \chi_n(x) \phi(x)$ for all $(s, x) \in \ol{\Omega}_n$ as the boundary condition on $\Gamma_n$. Then, the PDE for $v^n$ is
\begin{align}
    0 = P v^n & \dfn -v^n_s + \frac{1}{2}\tr\left( A D^2 v^n \right) + \ca^n(v^n, \nabla v^n), \label{E:localized_PDE} \\
    \ca^n(g, p) & \dfn b' p - \frac{\alpha}{2} p' A p + \chi_n \left( \frac{\gamma}{\alpha} + \Hcal(g, p)\right). \label{E:checka_def}
\end{align}
with $v^n = \varphi^n$ on $\Gamma_n$. 

Lastly, we briefly discuss the smoothness of $\Gamma_n$ and $\varphi^n$ here. For $\Gamma_n$, given $B \Omega_n = \{ 0 \} \times \Omega_n \subseteq \{ 0 \} \times \reals^d$ and $S \Omega_n = (0, T - t) \times \partial \OO_n$ with $\partial \OO_n \in C^{2, \beta}$ by Assumption \ref{A:region}, one can easily verify that $\Gamma_n \in H_{\eta}$ for $1 \leq \eta \leq 2 + \beta$. For $\varphi^n$, the independence of $s$ and $\chi_n \phi \in C^{2, \beta}(\OO; \reals)$ by Assumption \ref{A:phi_psi_alt} imply $\varphi^n \in H_{\eta, \Omega_n}$ for $1 \leq \eta \leq 2 + \beta$.

\begin{proposition}\label{P:localied_PDE_existence}
There exists a unique solution $v^n \in H_{2 + \beta, \Omega_n}$ to \eqref{E:localized_PDE}.
\end{proposition}
\begin{proof}[Proof of Proposition \ref{P:localied_PDE_existence}]
Analogous to the proof of \textit{Proposition A.1}, the existence of a solution to localized PDE \eqref{E:localized_PDE} is based on \cite[Theorem 12.16]{MR1465184}. To invoke this theorem, we need to verify the following:
\begin{enumerate}[(1)]
    \item  $\Gamma_n \in H_{1 + \beta'}$ and $\varphi^n \in H_{1 + \beta', \Omega_n}$ for some $\beta' \in (0, 1)$.
    
    \item $A^{i j} \in H_{1, K}$ for all bounded subsets $K$ of $\Omega_n$, and $\ca^n(g, p) \in H_{\beta, K}$ for all bounded subsets $K$ of $\Omega_n \times \reals \times \reals^d$.

    \item There exists a constant $C(n)$ such that $g \ca^n(g, 0) \leq C(n)(1 + g^2)$ for $(s, x) \in [0, T] \times \OO_n$.

    \item For any given interval $[g_1, g_2]$,
    \begin{equation*}
        \limsup_{\abs{p} \uparrow \infty} \sup_{\substack{(s, x) \in [0, T] \times \ol{\OO}_n, \\ g \in [g_1, g_2]}} \frac{\abs{\ca^n(g, p)}}{\abs{p}^2} < \infty.
    \end{equation*}
    \item
    \begin{equation*}
        \limsup_{\abs{p} \uparrow \infty}  \sup_{(s, x) \in [0, T] \times \ol{\OO}_n} \frac{\abs{\nabla_x A_{i j}}}{\abs{p}^2} < \infty,
    \end{equation*} 
    
\end{enumerate}

Since $1 \leq 1 + \beta' \leq 2 + \beta$, the previous discussion on the smoothness of $\Gamma_n$ and $\varphi_n$ implies part (1).  By Assumption \ref{A:region}, $A^{i j} \in C^{(1, 1)}(K; \reals) \subset H_{1, K}$ for all bounded subsets $K \in \Omega_n$. As $b$, $A$, $\chi_n$, $\gamma$ and $\Hcal$ are all continuously differentiable functions with respect to their variables, it follows that $\ca^n$ is continuously differentiable with respect to $(s, x, g, p)$. This implies that $\ca^n \in H_{1, K} \subseteq H_{\beta, K}$ for all bounded subsets $K$ of $\Omega_n \times \reals \times \reals^d$ and part (2) follows. Parts (3) and (4) are shown in Lemmas \ref{L:g_checka_bnd},  \ref{L:p_checka_bnd} respectively below. As for part (5), $A$ being independent of $p$ implies $\abs{\nabla_x A_{i j}}$ is of order 1 with respect to $\abs{p}$, giving the result. 

Putting all this together, \cite[Theorem 12.16]{MR1465184} yields a solution $v^n \in H^{- 1 - \beta'}_{2 + \beta, \Omega_n}$\footnote{The space $H^{-1-\beta'}_{2+\beta,\Omega_n}$ is defined in \cite[Chapter IV.1]{MR1465184} and is written $H^{(-1-\beta')}_{2+\beta,\Omega_n}$.}. By the remark at the end of \cite[Theorem 12.16]{MR1465184}, this solution is unique as $A$ is independent of $g$, and $\ca$ is Lipschitz with respect to $g$ and $p$. 

Lastly, we need to verify that $\varphi^n$ satisfies the compatibility condition of the first order, $P \varphi^n = 0$ on $C \Omega_n$. Indeed, $\nabla \varphi^n = 0$ and $\chi_n = 0$ on $C \Omega_n$ imply $\ca = 0$ on $C \Omega_n$. Moreover, as $\varphi^n_s = 0$ and $D^2 \varphi^n = 0$ on $C \Omega_n$, it follows $P \varphi^n = 0$ on $C \Omega_n$. Then, since $\Gamma_n \in H_{2 + \beta}$ and $\varphi^n \in H_{2 + \beta}$, we can further conclude that $v^n \in H_{2 + \beta, \Omega_n}$.

\end{proof}
\begin{remark}\label{R:gradient_bnd}
    Analogous to \cite[Remark A.2]{MR4086602}, $\abs{G^n}_{2 + \beta, \Omega_n}$ implies $$\sup_{(s, x) \in [0, T] \times \ol{\OO}_n} \abs{\nabla G^n} < \infty.$$
\end{remark}

\subsection{Localized optimal investment problem} We now define a localized optimal investment problem for the localized PDE defined above.  To do so, fix $t\leq T, x\in\OO$ and take $n$ large enough so that $x\in\OO_n$. Define the localized default time
\begin{equation}\label{E:tau_n_def}
    \tau^n = \inf \cbra{ s \geq t \such \int_t^s (\chi_n \gamma)(u, X_u) \dif u = -\log(U)}.
\end{equation}
and the localized default indicator process $H^n_{\cdot} \coloneqq \mathbbm{1}_{\tau^n \leq \cdot}$. To the define the localized equity process, by enlarging the probability space, assume there is a $d$-dimensional Brownian motion $\wh{W}$ independent of $U,W$ and $B$. The filtration  $\filtg^n$ is the $\prob$ augmented version of $\filt^{W, B,\wh{W}} \vee \filt^{H^n}$. The equity market in this local investment problem has the following price dynamics:
\begin{equation*}
    \frac{dS_{s}^{e,n}}{S_{s-}^{e,n}}  = 1_{s \leq \tau^n} \Bigl( (\chi_n \mu_e)(s, X_s) ds + (\sqrt{\chi_n} \sigma_e)(s, X_s) dZ^n_s \Bigr) - \ell_e(s, X_s) dH^n_s,
\end{equation*}
where
\begin{equation*}
    d Z^n_s = \sqrt{\chi_n} \rho d W_s + \ol{\rho} d B_s + \sqrt{1 - \chi_n} \rho d \wh{W}_s.
\end{equation*}
The fictitious asset $S^{r,n}$ related to the CDS market in this local investment problem has dynamics
\begin{align*}    \frac{\dif S^{r,n}_{s}}{S^{r,n}_{s-}} = 
    & 1_{s \leq \tau^n} \bigl( \chi_n ( \sigma_r' a \tprobnu -\wt{\gamma} )(s, X_s)ds + \chi_n ( \sigma_r' a )(s, X_s) d W_s \\  
    & \qquad + \sqrt{\chi_n (1 - \chi_n)} ( \sigma_r' a  )(s, X_s) d \wh{W}_s \bigr) + d H^n_s. \nonumber
\end{align*}
The joint market dynamics is thus
\begin{align*}
\frac{dS^n_s}{S^n_{s-}} = 
& 1_{s\leq \tau^n}\bigl( (\chi_n \mu)(s,X_s)ds + (\chi_n \sigma_W)(s,X_s)dW_s + ( \sqrt{\chi_n} \sigma_B )(s,X_s)d B_s \\
& \qquad + ( \sqrt{\chi_n(1 - \chi_n)} \sigma_W )(s, X_s) d \wh{W}_s \bigr) - \ell(s,X_s)dH^n_s. \nonumber
\end{align*}


Next, define $\zeta^n \coloneqq \inf \left\{ \left. s \geq t \right| X_s \in \partial \OO_n \right\}$ as the first time $X$ exits $\OO_n$. The local investment problem has horizon $T\wedge\zeta^n$, and we define $\M^n$ as the equivalent local martingale measures on $\G^n_{T \wedge \zeta^n}$, $\wt{\M}^n$ the subset with finite relative entropy with respect to $\prob$ on $\G^n_{T \wedge \zeta^n}$, and  $\A^n$ the class of $\filtg^n$ predictable trading strategies $\pi^n$ whose wealth process $\We^{\pi^n}$ is a $\mathbb{Q}^n$ supermartingale for all $\mathbb{Q}^n \in \wt{\M}$. Here, investment stops at $T\wedge\tau^n\wedge \zeta^n$ and $\We^{\pi^n}$ has dynamics
\begin{equation}\label{E:loc_wealth_dynamics}
    \begin{split}
    d \We^{\pi^n}_s = & 1_{s \leq \tau^n \wedge \zeta^n} \left( \pi^n_s \right)' \bigl( (\chi_n \mu) (s, X_s) d s + (\chi_n \sigma_W)(s, X_s)' d W_s + ( \sqrt{\chi_n} \sigma_B ) (s, X_s) d B_s \\
    & \qquad + ( \sqrt{\chi_n(1 - \chi_n)} \sigma_W) (s, X_s) d \wh{W}_s \bigr) - \mathbbm{1}_{s \leq \zeta^n } \left( \pi^n_s \right)' l(s, X_s) d H^n_s.
    \end{split}
\end{equation}
The random endowment is $(\chi_n \phi)(X_{T \wedge \zeta^n}) \mathbbm{1}_{\tau^n > T \wedge \zeta^n} + \Psi(\tau^n , X_{\tau^n}) \mathbbm{1}_{\tau^n \leq T \wedge \zeta^n}$. The value function is
\begin{equation} \label{local_utility_function}
    u^n(t, x) \coloneqq \sup_{\pi^n \in \A^n} \E \left[ - \e^{- \alpha \left( \We^{\pi^n}_{T \wedge \zeta^n} + (\chi_n \phi)(X_{T \wedge \zeta^n}) \mathbbm{1}_{\tau^n > T \wedge \zeta^n} + \Psi(\tau^n , X_{\tau}^n) \mathbbm{1}_{\tau^n \leq T \wedge \zeta^n} \right) } \right].
\end{equation}

Analogously to \eqref{certainty equivalent}, and with an eye towards \eqref{E:G_HJB_local} we define the certainty equivalent $G^n(t,x) = -(1/\alpha)\log(-u^n(t,x))$, which analogously to \eqref{E:G_HJB} is expected to solve
\begin{equation*}
    0 = G^n_t + LG - \frac{\alpha}{2}\nabla (G^n)'A\nabla G^n + \frac{\chi_n\gamma}{\alpha} +  \sup_{\pi} H^n(\pi,G^n,\nabla G^n),\quad \chi_n\phi = G^n(T,\cdot).
\end{equation*}
To obtain \eqref{E:G_HJB_local} we must first obtain the local Hamiltonian, which has the same form as the right side of \eqref{E:hamil}, but for the localized assets.  As such, we note the local analog of \eqref{E:SigE_UpsE} is
\begin{equation*}
    \begin{split}
        \Sigma^n_{e} &= d\langle S^{e,n}, S^{e,n}\rangle / ds = \chi_n \Sigma_e;\quad \Upsilon^n_{e} = d\langle S^{e,n}, X \rangle /ds = \chi_n \Upsilon_e,\\
        \Sigma^n &= d\langle S^n, S^n \rangle /ds = \chi_n\Sigma;\quad \Upsilon^n = d\langle S^n, X\rangle / ds = \chi_n \Upsilon.
    \end{split}
\end{equation*}
Note that $\chi_n$ factors our of each quantity above.  Therefore, from \eqref{E:hamil} we obtain $H^n = \chi_n H$ and hence \eqref{E:G_HJB_local} follows.

Next, we recall \textit{Lemmas A.3 and A.4} specified to our notation. The first one is regarding the structure of $\qprob^n \sim \prob$ on $\G^n_{T \wedge \zeta^n}$, and the second one is regarding when $\qprob^n \in \M^n$. Below, our notation follows \eqref{E:mm_den}.

\begin{lemma}\label{L:local_ELLM_density_rep}
    Any measure $\qprob^n \sim \prob$ on $\G^n_{T \wedge \zeta^n}$ has representation
    \begin{equation}\label{E:local_elmm}
        \left. \frac{d \qprob^n}{d \prob} \right|_{T \wedge \zeta^n} = \Ecal\left(\int (\mbf{A}^n_s) 'dW_s + (\mbf{B}^n_s)' d B_s + \mbf{C}^n_s d M_s+ (\mbf{D}^n_s)' d \wh{W}_s \right)_t^{T \wedge \zeta^n}
    \end{equation}
    where $\mbf{A}^n$, $\mbf{B}^n$, $\mbf{C}^n$, and $\mbf{D}^n$ are $\filtg^n$ predictable processes. Additionally, $(W^{\qprob^n}, B^{\qprob^n}, \wh{W}^{\qprob^n})$ is a $\qprob^n$ Brownian motion over $[t,T \wedge \zeta^n]$ where $dW^{\qprob^n}_s = dW_s - \mbf{A}^n_s d s$, $dB^{\qprob^n}_{s} = dB_{s} -  \mbf{B}^n_s ds$, $d\wh{W}^{\qprob^n}_{s}  = d\wh{W}_{s} - \mbf{D}^n_s ds$, and $dM^{\qprob^n}_{s} = dM_{s} - 1_{s\leq \tau^n}(\chi_n \gamma)(s, X_s) \mbf{C}^n_s ds$ is a $\qprob^n$ martingale stopped at $T \wedge \zeta^n$.
\end{lemma}

\begin{lemma}
Let $\qprob^n \sim \prob$ on $\G^n_{T \wedge \zeta^n}$, and $\mbf{A}^n$, $\mbf{B}^n$, $\mbf{C}^n$, $\mbf{D}^n$ be in Lemma \ref{L:local_ELLM_density_rep}. Then $\qprob^n \in \M^n$ if and only if for $\prob \times \text{Leb}_{[t, T]}$ almost all $(\omega, s)$ on the set $\{ s \leq \tau^n \wedge T \wedge \zeta^n \}$, we have (suppressing the $(s,X_s)$ function arguments)
\begin{align}\label{E:risk_neutral_drift_condition}
    0 = 
    & \chi_n(\mu - \gamma \ell) + \chi_n \sigma_W \mbf{A}^n_s + \sqrt{\chi_n} \sigma_B \mbf{B}^n_s + \sqrt{\chi_n(1 - \chi_n)} \sigma_W \mbf{D}^n_s - \chi_n \gamma \mbf{C}^n_s \ell.
\end{align}
\end{lemma}
Now, we give the counterpart to \textit{Proposition A.5}. It shows that $G^n$ from Proposition \ref{P:localied_PDE_existence} is the certainty equivalent to \eqref{local_utility_function}. 

\begin{proposition}\label{P:localized_verification}
There is a unique solution $G_n \in H_{2 + \beta', \Omega_n}$ to the PDE in \eqref{E:G_HJB_local} that takes the form $G^n(t, x) = -\frac{1}{\alpha} \log \bigl( -u^n(t, x) \bigr)$, for $u^n$ defined in \eqref{local_utility_function}. The optimal policy functions are given in \eqref{E:gen_opt_strat_1} and \eqref{E:gen_opt_strat_2} except with $G$ and $\nabla G$ replaced by $G^n$ and $\nabla G^n$. The optimal martingale density process for $s \in [t, T \wedge \zeta^n]$ is
\begin{equation}\label{E:optimal_loc_martingale_density_local}
    \wh{Z}^n_s \dfn \wh{Z}^{\wh{\qprob}^n}_s = \e^{-\alpha \Bigl( \We^{\wh{\pi}^n} + 1_{\tau^n > s} G^n(s, X_s) + 1_{\tau^n \leq s} \Psi(\tau^n , X_{\tau^n}) - G^n(t, x) \Bigr)}.
\end{equation}
\end{proposition}
\begin{proof}[Proof of Proposition \ref{P:localized_verification}]
Due to the existence of the Brownian motion $\wh{W}$ in the localized optimal investment problem, we need to change the proof of \textit{Proposition A.5} slightly. Denoting $\wh{\We}^n \dfn \We^{\wh{\pi}_n}$, we need to show (1) $\wh{\We}^n$ is a $\qprob^n$ martingale for all $\qprob^n \in \wt{\M}^n$; (2) $\wh{Z}^n$ given by \eqref{E:optimal_loc_martingale_density_local} is a strictly positive martingale and $\wh{\qprob}^n$ defined by $\wh{Z}^n_{T \wedge \zeta^n}$ is in $\wt{\M}^n$. We begin with proof of (1). For $\qprob^n \in \wt{\M}^n$, by \eqref{E:risk_neutral_drift_condition} we know
\begin{align*}
    d \wh{\We}^n_s = 
    & \mathbbm{1}_{s \leq \tau^n \wedge \zeta^n} \Bigl( (\chi_n \sigma_W)(s, X_s)' d W^{\qprob^n}_s + ( \sqrt{\chi_n} \sigma_B ) (s, X_s) d B^{\qprob^n}_s \\
    & \qquad + ( \sqrt{\chi_n(1 - \chi_n)} \sigma_W) (s, X_s) d \wh{W}^{\qprob^n}_s \Bigr) - 1_{s \leq \zeta^n } \left( \pi^n_s \right)' l(s, X_s) d M^{\qprob^n}_s.
\end{align*}
Hence,
\begin{equation*}
    \left[ \wh{\We}^{\pi^n} \right]_{T \wedge \zeta^n} = \int_{t}^{T \wedge \zeta^n \wedge \tau^n} \bigl( \chi_n (\wh{\pi}^n)' \Sigma \wh{\pi}^n \bigr)(s,X_s) ds + 1_{\tau^n \leq T \wedge \zeta^n} \bigl( l' \wh{\pi}^n \bigr)^2(\tau^n,X_{\tau^n}).
\end{equation*}
Remark \ref{R:gradient_bnd} shows $\abs{\wh{\pi}^n}$ is bounded on $[t, T] \times \ol{\OO}_n$. As $\Sigma$ and $l$ are bounded on $[t, T] \times \ol{\OO}_n$ and $\chi_n$ is bounded on $\ol{\OO}_n$, we conclude that  
$\left[ \wh{\We}^{\pi^n} \right]_{T \wedge \zeta^n}$ is $\qprob^n$ almost surely bounded by a constant depending on $n$, which implies that $\wh{\We}^{\pi^n}$ is a $\qprob^n$ martingale.

Next, we move to prove (2). Using the arguments in the proof of \textit{Proposition A.5}, we know that $\wh{Z}^n$ has dynamics on $[t, T \wedge \zeta^n]$ of 
\begin{equation*}
    \frac{d \wh{Z}^n_s}{\wh{Z}^n_{s -}} = \mathbbm{1}_{s \leq \tau^n} \bigl( (\mbf{A}^n_s)' d W_s + (\mbf{B}^n_s)' d B_s + (\mbf{D}^n_s)' d \wh{W}_s \bigr) + \mbf{C}^n_s d M^n_s,
\end{equation*}
where $\mbf{A}^n_s = -\alpha \left( \chi_n \sigma_W' \wh{\pi}^n + a' \nabla G^n \right)_s$, $\mbf{B}^n_s = -\alpha \left( \sqrt{\chi_n} \sigma_B' \wh{\pi}^n \right)_s$, $\mbf{D}^n_s = - \alpha \left( \sqrt{\chi_n(1 - \chi_n)} \sigma_W' \wh{\pi}^n \right)_s$, and $\mbf{C}^n_s = \left( \e^{\alpha G^n - \alpha \psi + \alpha (\wh{\pi}^n)' l} - 1 \right)_s$. Remark \ref{R:gradient_bnd} shows $\abs{\mbf{A}^n_s}$, $\abs{\mbf{B}^n_s}$, $\abs{\mbf{C}^n_s}$, $\abs{\mbf{D}^n_s}$ are all bounded and $\filt^W$ predictable. Also, there is also a $\eps_n$ such that $\mbf{D}^n_s > - (1 - \eps_n)$. By \cite[Theorem 9]{MR2574236}, we know that $\wh{Z}^n$ is a strictly positive martingale. Then we can apply the same argument in the proof of \textit{Proposition A.5 and Lemma C.7}\footnote{The proof therein assumes $\Sigma$ does not degenerate, but calculation shows the result still holds even wtih degeneracy. Also, there is a typo in the second to last line of the proof of \textit{Lemma C.7}, as $+\gamma\ell$ should be $-\gamma\ell$.} to conclude $\wh{\qprob}^n$ defined by $\wh{Z}^n_{T \wedge \zeta^n}$ is in $\wt{\M}^n$.

\end{proof}


\subsection{Unwinding the localization: analytical results}
Now, we need to unwind the existence result of local PDE to the global one. 

\begin{proposition}\label{P:analytical_unwind_lower_bnd}
Let $G^n$ from Proposition \ref{P:localied_PDE_existence} and \ref{P:localized_verification} and $\underline{\phi}$ from Proposition \ref{A:phi_psi_alt}. Then $G^n(s, y) \geq \underline{\phi}$ on $[0, T] \times \ol{\OO}_n$.
\end{proposition}

\begin{proof}[Proof of Proposition \ref{P:analytical_unwind_lower_bnd}]
Under Assumption \ref{A:incomplete_mkt} with $\abs{\sigma_r} > 0$, we can apply the same argument in \textit{Proposition A.6} to prove this result. Therefore, we only need to show $G^n$ is bounded from below by $\underline{\phi}$ on $[0, T] \times \ol{\OO}_n$ under Assumption \eqref{A:complete_mkt} or under Assumption \ref{A:incomplete_mkt} with $\abs{\sigma_r} \equiv 0$. Note that in both cases, we have
\begin{equation*}
    \check{a}^n(g, 0) = \chi_n \left( \frac{\gamma}{\alpha} + \Hcal_c(g, 0) \right).
\end{equation*}
Assume $G^n$ has the minimum in $[0, T) \times \OO_n$ at $(s_0, y_0)$. If $s_0 > 0$, then $G^n_s(s_0, y_0)= 0$; otherwise, $G^n_s(s_0, y_0) \geq 0$. Also, $\nabla G^n(s_0, y_0) = 0$. By the ellipticity of $A$ in $\OO_n$, we know that at $(s_0, y_0)$,
\begin{equation*}
    \begin{split}
    0 \geq & \chi_n \left( \frac{\gamma}{\alpha} + \Hcal_c(G^n, 0)\right)\\
    \geq & \frac{\chi_n}{2\alpha}\left(\mu_e - \wt{\gamma}\ell_e \right)'\Sigma_e^{-1}\left(\mu_e - \wt{\gamma}\ell_e \right) + \frac{\wt{\gamma}}{\alpha}\left(\frac{\gamma}{\wt{\gamma}} - 1 - \log\left(\frac{\gamma}{\wt{\gamma}}\right)\right)- \wt{\gamma}(G^n - \psi) \\
    \geq & - \wt{\gamma}(G^n - \psi).
    \end{split}
\end{equation*}
Therefore, we have $G^n(s_0, y_0) \geq \psi \geq 0 \geq \underline{\phi}$ by Assumption \ref{A:phi_psi_alt}. Moreover, as $G^n(T, \cdot) = \chi_n \phi \geq 0 \geq \underline{\phi}$ and $G^n(t, \partial \OO_n) = 0 \geq \underline{\phi}$ for every $t \in [0, T]$, we can conclude that $G^n(s, y) \geq \underline{\phi}$ on $[0, T] \times \ol{\OO}_n$.
\end{proof}


\begin{proposition}\label{P:analytical_unwind}
Let $G^n$ from Proposition \ref{P:localied_PDE_existence} and \ref{P:localized_verification}. Assume for each $k \in \mathbb{N}$ that
\begin{equation}\label{E:Gn_bnd_assumption}
    \sup_{n \geq k + 1} \sup_{0 \leq s \leq T, y \in \ol{\OO}_k} G^n(s, y) = C(k) < \infty.
\end{equation}
Then there exists a solution $G$ to \eqref{E:PDE_2}. In particular, there is a subsequence (still labeled $n$) such that $G^n$ converges to $G$ in $H_{2 + \beta, (0, T) \times \OO, loc}$.
\end{proposition}

\begin{proof}[Proof of Proposition \ref{P:analytical_unwind}]
The proof of \textit{Proposition A.7} is divided into $4$ steps. Here, we need only verify the quantities $A_k^{\infty}$, $B_k^{\infty}$, $C_k^{\infty}$, and $D_k^{\infty}$ defined in \textit{Step $1$, equations (A.16), (A.17)} are all finite with $A_k^{\infty} = 1$ and $C_k^{\infty} = 0$. Then we can apply the same arguments in the remaining part of \textit{Step $1$}, with \cite[Theorem 11.3(b)]{MR1465184}, Proposition \ref{P:analytical_unwind_lower_bnd}, and Equation \eqref{E:Gn_bnd_assumption} to show
\begin{equation*}
    \sup_{n \geq k + 2} \sup_{(s, y) \in [0, T] \times \ol{\OO}_k} \abs{\nabla G^n(s, y)} < \infty.
\end{equation*} 
Using this bound, we can apply the same arguments in \textit{Steps 2-4} to finish the proof. To define $A_k^{\infty}$, $B_k^{\infty}$, $C_k^{\infty}$, and $D_k^{\infty}$ we first define the  the Bernstein function
\begin{equation*}
\Ecal(s, y, p) \dfn \frac{1}{2} p' A(s, y) p;\qquad (s,y,p)\in [0,T]\times \OO \times \reals^d.
\end{equation*}
Under our assumptions, for each $k$ there exists $0 < \lambda_k < \Lambda_k$ such that $\lambda_k p' p \leq \Ecal(s, y, p) \leq \Lambda_k p' p$ on $[0, T] \times \ol{\OO}_k \times \reals^d$. Next, we define the operators $\delta(p)$, and $\ol{\delta}(p)$, which act on functions $f$ of $(s,y,g,p)$ by 
\begin{equation*}
    \delta(p)[f] \dfn  f_g + \frac{1}{p' p} p' \nabla_x f; \qquad \ol{\delta}(p)[f] \dfn p' \nabla_p f.
\end{equation*}
Lastly, we note that for $n \geq k + 1$, $\chi_n(x) = 1$ on $\ol{\OO}_k$, and hence in \eqref{E:checka_def} we write  $\ca(g, p)$ for $\ca^n(g, p)$ as it no longer depends upon $n$. With all this preparation, \textit{equations (A.16), (A.17)} take the form
\begin{equation}\label{E:ABCD_def}
    \begin{split}
    A_k& \dfn  \frac{1}{\Ecal} \Bigl( \frac{p' p}{8 \lambda_k} \sum_{i, j = 1}^{d} \bigl( \ol{\delta}(p)[A^{i j}] \bigr)^2 + \bigl( \ol{\delta}(p) - 1 \bigr) [\Ecal] \Bigr)\\
    B_k& \dfn \frac{1}{\Ecal} \Bigl( \delta(p) [\Ecal] +  \bigl( \ol{\delta}(p) - 1 \bigr)[\check{a}] \Bigr),\\
    C_k& \dfn \frac{1}{\Ecal} \Bigl( \frac{p' p}{8 \lambda_k} \sum_{i, j = 1}^{d} \bigl( \delta(p) [A^{i j}] \bigr)^2 + \delta(p) [\check{a}] \Bigr),\\
    D_k& \dfn \frac{1}{\Ecal} \Bigl( \Lambda_k p' p + \abs{p} \bigl( \abs{\nabla_p \Ecal} + \abs{\nabla_p \check{a}} \bigr) \Bigr).
    \end{split}
\end{equation}
Next, for any constant $C(k) > 0$, and $Y \in \cbra{A,B,C,D}$ we define
\begin{equation}\label{E:ABCD_k_infty}
    \begin{split}
        Y^\infty_k & \dfn \limsup_{\abs{p} \uparrow \infty} \sup_{\substack{(s, y) \in [0, T] \times \ol{\OO}_k, \\ g \in [-C(k), C(k)]}} \abs{Y_k(s, y, g, p)}.
    \end{split}
\end{equation}
As mentioned above, our goal is to verify 
\begin{equation*}
    A_k^{\infty} = 1, \quad B^{\infty}_k < \infty, \quad C^{\infty}_k = 0, \quad D^{\infty}_k < \infty.
\end{equation*}
As $A=A(s,y)$ does not depend on $p$ and $(\ol{\delta}(p)-1)[\Ecal] = \Ecal$ we obtain $A_k \equiv  1$, and hence $A^{\infty}_k = 1$. As for the other quantities, we show $B^{\infty}_k < \infty$, $C^{\infty}_k = 0$, and $D^{\infty}_k < \infty$ in Lemmas \ref{L:BCD_k_infty_finite_1} and \ref{L:BD_k_infty_finite_2} below. As the rest of the proof holds by repeating the steps in the proof of \textit{Proposition A.7}, we obtain the result.
\end{proof}


\subsection{Unwinding the localization: probabilistic results}\text{}

We now unwind the localization from a probabilistic perspective, culminating in Proposition \ref{P:unwind_prob}, which will finish the proof of Theorem \ref{T:main_result}.  A key result from \cite{MR4086602} which we use throughout, and which holds in our setting with no changes to the proof\footnote{Presently we have the independent Brownian motion $\wh{W}$ which was absent in \cite{MR4086602}. However, it is clear from \eqref{E:loc_wealth_dynamics} and the proof of \textit{Proposition A.9} that the integral with respect to $\wh{W}$ vanishes in probability as $n\to\infty$.}, is \textit{Proposition A.9}, which we repeat here for ease of reference. 

\begin{proposition}\label{P:unwind_prob_0}
Assume \eqref{E:Gn_bnd_assumption} and let $G$ be the solution to \eqref{E:G_HJB} obtained in Proposition \ref{P:analytical_unwind}. Define $\wh{\pi}$ using \eqref{E:hat_pi} (more precisely \eqref{E:gen_opt_strat_1}, \eqref{E:gen_opt_strat_2}) and $\wh{Z}$ as in \eqref{E:opt_mm_dens}. Then
\begin{enumerate}[(i)]
\item $\wh{Z}$ defines a measure $\wh{\qprob} \in \tM$.
\item $\We^{\wh{\pi}}$ is a $\wh{\qprob}$ sub-martingale.
\item For $u$ defined in \eqref{E:vf_tx} we have $G(t,x) \geq -(1/\alpha)\log(-u(t,x,0))$.
\end{enumerate}
\end{proposition}

Net, we introduce the inequalities that will appear frequently below. For $x > 0$, $y \in \reals$, $K > 0$ and $p > 1$, $q = \frac{p}{p - 1}$,
\begin{equation}\label{E:inequalities_prob_result}
    x y \leq \frac{1}{K} \bigl( x \log x + \e^{K y} \bigr),\quad \e^{y} \leq \frac{1}{p} \e^{p y} + \frac{1}{q} \e^{q y}.
\end{equation}

\begin{proposition}\label{P:unwind_prob}
Under Assumption \ref{A:complete_mkt} the conclusions of Theorem \ref{T:main_result_complete} follow.  Similarly, under Assumption \ref{A:incomplete_mkt} the conclusions of Theorem \ref{T:main_result} follow.
\end{proposition}
\begin{proof}[Proof of Proposition \ref{P:unwind_prob}] To ease the notational burden, define
\begin{equation}\label{E:wh_an_bn_def}
    \wh{a}_n \dfn T\wedge \zeta^{n-1};\qquad \wh{b}_n \dfn T\wedge \zeta^n.
\end{equation}
According to the outline in the proof of \textit{Proposition A.11} the result in each respective case will follow provided (for a starting time $t$ and location $x$)
\begin{enumerate}[(1)]
\item For each $k$ and $n\geq k+1$, find a probability measure $\ol{\qprob}^n \in \wt{\M}^n$ such that
\begin{equation}\label{E:dual_quant}
    \frac{1}{\alpha} \expvs{Z^{\ol{\qprob}^n}_{\wh{b}_n} \log Z^{\ol{\qprob}^n}_{\wh{b}_n}} + \expvs{Z^{\ol{\qprob}^n}_{\wh{b}_n} \Bigl( 1_{\tau^n > \wh{b}_n}(\chi_n \phi)(X_{\wh{b}_n}) + 1_{\tau^n \leq \wh{b}_n} \psi(\tau^n, X_{\tau^n}) \Bigr)},
\end{equation}
is bounded from above by a constant $C(k)$ when $x\in\ol{\OO}_k$, $t\leq T$ and $n\geq k+1$. By duality, this will establish \eqref{E:Gn_bnd_assumption} and allow us to invoke Propositions \ref{P:analytical_unwind} and \ref{P:unwind_prob_0}.

\item For each $\qprob \in \wt{\M}$, adjust the density process $Z^{\qprob}$ in $( \tau \wedge \wh{a}_n, \tau \wedge \wh{b}_n]$ using $Z^{\ol{\qprob}^n}$ to obtain a new density process $Z^n$, and show that $\qprob^n$ defined by $Z^n$ is in $\wt{\M}^n$. Next, show that as $n \to \infty$
\begin{equation*}
    \begin{split}
        &\expvs{Z^n_{\wh{b}_n} \log Z^n_{\wh{b}_n}} \to \expvs{Z^{\qprob}_{T \wedge \tau}\log \left( Z^{\qprob}_{T \wedge \tau} \right)},\\
        &\expvs{Z^n_{\wh{b}_n} \bigl(1_{\tau^n > \wh{b}_n} (\chi_n \phi)(X_{\wh{b}_n}) + 1_{\tau^n \leq \wh{b}_n} \psi(\tau^n,X_{\tau^n}) \bigr) } \to \expvs{Z^{\qprob}_{T\wedge\tau} \left( 1_{\tau > T} \phi(X_T) + 1_{\tau \leq T} \psi(\tau,X_{\tau}) \right)}.
    \end{split}
\end{equation*}
This will give us the upper bound on $G$,
\begin{equation*}
    \begin{split}
        G(t, x) &\leq \inf_{\qprob \in \wt{\M}} \Bigl( \frac{1}{\alpha} \expvs{Z^{\qprob}_{\tau \wedge T}\log \left( Z^{\qprob}_{\tau \wedge T} \right)} + \expvs{Z^{\qprob} \left( 1_{\tau > T} \phi_T + 1_{\tau \leq T} \psi_\tau \right)} \Bigr),\\
        &= -\frac{1}{\alpha}\log(-u(t,x,0)),
    \end{split}
\end{equation*}
where the last equality follows according to the standard duality theory.  In view of Proposition \ref{P:unwind_prob_0} we conclude that $G$ is the certainty equivalent, and $\wh{\pi},\wh{\qprob}$ are optimal, as \eqref{E:opt_mm_dens} verifies the first order optimality conditions, and  $\wh{\qprob}\in\tM$ implies $\We^{\wh{\pi}}$ is a $\wh{\qprob}$ super-martingale, hence martingale by Proposition \ref{P:unwind_prob_0}. This will yield Theorems \ref{T:main_result_complete} and \ref{T:main_result}.
\end{enumerate}

We start with step $(1)$ and use the notation of \eqref{E:local_elmm}. The idea is to set $\ol{\qprob}^n = \tprob$, but we have to make slight adjustments due to the localization (in particular because of the martingale $M^n$). First, under Assumption \ref{A:complete_mkt} we set (see \eqref{E:tprob_mpr_c})
\begin{equation}\label{E:bar_qprob_n_1}
    \ol{\mbf{A}}^n = \ol{\mbf{A}} = -\tprobnu = -\sigma_e^{-1}(\mu_e - \wt{\gamma}\ell_e);\quad \ol{\mbf{B}}^n \equiv 0;\quad \ol{\mbf{C}}^n = \ol{\mbf{C}} =  \frac{\wt{\gamma}}{\gamma}-1;\quad \ol{\mbf{D}}^n \equiv 0.
\end{equation}
The market price of risk equations \eqref{E:risk_neutral_drift_condition} are easily seen to hold. Alternatively, under Assumption \ref{A:incomplete_mkt}, we set
\begin{equation}\label{E:bar_qprob_n_2}
    \ol{\mbf{A}}^n = \ol{\mbf{A}} = -\tprobnu ;\quad \ol{\mbf{B}}^n =  -\sqrt{\chi}_n(\ol{\rho})^{-1}\left(\sigma_e^{-1}(\mu_e -\wt{\gamma}\ell_e) - \rho\tprobnu\right);\quad \ol{\mbf{C}}^n = \ol{\mbf{C}} = \frac{\wt{\gamma}}{\gamma}-1;\quad \ol{\mbf{D}}^n \equiv 0.
\end{equation}
The market price of risk equations \eqref{E:risk_neutral_drift_condition} again hold. With these assignments, we now show \eqref{E:dual_quant}. 
By first conditioning on $\F^{W,H^n}_{\tau^n\wedge \wh{b}_n}$ we obtain
\begin{equation*}
    \begin{split}
        &\expv{}{}{Z^{\ol{\qprob}^n}_{\wh{b}_n} \Bigl( 1_{\tau^n > \wh{b}_n}(\chi_n \phi)(X_{\wh{b}_n}) + 1_{\tau^n \leq \wh{b}_n} \psi(\tau^n, X_{\tau^n}) \Bigr)} = \expv{\ol{\qprob}^n}{}{\Bigl( 1_{\tau^n > \wh{b}_n}(\chi_n \phi)(X_{\wh{b}_n}) + 1_{\tau^n \leq \wh{b}_n} \psi(\tau^n, X_{\tau^n}) \Bigr)},\\
        & \qquad = \expv{\ol{\qprob}^n}{}{(\chi_n\phi)(X_{\wh{b}_n})e^{-\int_t^{\wh{b}_n}(\chi_n\wt{\gamma})(v,X_v)dv} + \int_t^{\wh{b}_n} \psi(u,X_u)(\chi_n\wt{\gamma})(u,X_u)e^{-\int_t^u (\chi_n\wt{\gamma})(v,X_v)dv}du}, \\
        & \qquad = \wtexpv{}{}{(\chi_n\phi)(X_{\wh{b}_n})e^{-\int_t^{\wh{b}_n}(\chi_n\wt{\gamma})(v,X_v)dv} + \int_t^{\wh{b}_n} \psi(u,X_u)(\chi_n\wt{\gamma})(u,X_u)e^{-\int_t^u (\chi_n\wt{\gamma})(v,X_v)dv}du}.
    \end{split}
\end{equation*}
The last equality follows as $\ol{\qprob}^n$ and $\tprob$ agree on $\F^{W}_{\wh{b}_n}$. As $\chi_n$ vanishes on $\partial \OO_n$ and $\chi_n,\wt{\gamma}\geq 0$ the first term is bounded above by
\begin{equation*}
    \wtexpv{}{}{1_{\zeta^n>T}(\chi_n\phi)(X_{\wh{b}_n})e^{-\int_t^{\wh{b}_n}(\chi_n\wt{\gamma})(v,X_v)dv}} \leq -\ul{\phi} + \wtexpv{}{}{\phi(X^{t,x}_T)},
\end{equation*}
where we brought back in the $(t,x)$ dependency.  Therefore, from part (1) of Assumption \ref{A:phi_psi_alt}, for $t\leq T$ and $x\in \ol{\OO_k}$ this term is bounded from above by some $c(k)$. As for the second term, if $\psi$ is bounded from above by a constant $K$ then so is the expected value.  Else, using $0\leq \chi_n\leq 1$ and $\wt{\gamma}\geq 0$ we deduce
\begin{equation*}
    \wtexpv{}{}{\int_t^{\wh{b}_n} \psi(u,X_u)(\chi_n\wt{\gamma})(u,X_u)e^{-\int_t^u (\chi_n\wt{\gamma})(v,X_v)dv}du} \leq \wtexpv{}{}{\int_t^T (\psi\wt{\gamma})(u,X^{t,x}_u)du}.
\end{equation*}
This term is also bounded from above by a constant $C(k)$ when $t\leq T$ and $x\in\ol{\OO}_k$ by Assumption \ref{A:phi_psi_alt}. It remains to bound the relative entropy. To this end, and because $\ol{\mbf{D}}^n$ vanishes in each case, we write
\begin{equation*}
    \ol{Z}^n_{\cdot} = \Ecal \Bigl(  \int_t\ol{\mbf{A}}_u' d W_u + (\ol{\mbf{B}}^n_u)' d B_u + \ol{\mbf{C}}_u d M^n_u \Bigr)_{\cdot},
\end{equation*}
Using the identity $\Ecal \left( U + V + [U, V] \right) = \Ecal \left( U \right) \Ecal \left( V \right)$ we can write $\ol{Z}^n = \ol{Z}^{\mbf{A}}\ol{Z}^{\mbf{B},n}\ol{Z}^{\mbf{C},n}$ where
\begin{equation}\label{E:bZ_ABD_n}
    \ol{Z}^{\mbf{A}}_{\cdot} = \Ecal \Bigl( \int_t^{\cdot} \ol{\mbf{A}}_u' d W_u \Bigr);\quad \ol{Z}^{\mbf{B},n}_{\cdot} = \Ecal \Bigl( \int_t^{\cdot} (\ol{\mbf{B}}^n_u)' d B_u \Bigr);\quad 
    \ol{Z}^{\mbf{C}, n}_{\cdot} = \Ecal \Bigl( \int_t^{\cdot} \ol{\mbf{C}}_u d M^n_u \Bigr).
\end{equation}
Below, we will verify the identity
\begin{equation}\label{E:loc_rel_ent_ident}
    \begin{split}
        &\expvs{Z^{\ol{\qprob}^n}_{\wh{b}_n} \log Z^{\ol{\qprob}^n}_{\wh{b}_n}} = \mathbb{E}\bigg[\ol{Z}^{\mbf{A}}_{\wh{b}_n}\bigg(\log\left(\ol{Z}^{\mbf{A}}_{\wh{b}_n}\right) + \int_t^{\wh{b}_n} \bigg(\frac{1}{2}|\ol{\mbf{B}}^n_u|^2 \\
        &\qquad + (\chi_n\wt{\gamma})(u,X_u) e^{-\int_t^u (\chi_n\wt{\gamma})(v,X_v)dv}\left(\frac{\gamma}{\wt{\gamma}} - \log\left(\frac{\gamma}{\wt{\gamma}}\right) -1 \right)(u,X_u)\bigg)du\bigg)\bigg].
    \end{split}
\end{equation}
Admitting this, using the non-negativity of $\wt{\gamma}$ and $z\to z - \log(z) - 1$ on $(0,\infty)$; that $0\leq \chi_n \leq 1$; the definition of $\tprob$ and $\ol{\mbf{A}}$ (in particular that $\tprobnu$ is bounded on $\OO_n$); and the definitions of $Q =  Q_c$ or $Q = Q_i$ in \eqref{E:Qval_complete},\eqref{E:Qval_incomplete} respectively we obtain
\begin{equation*}
    \expvs{Z^{\ol{\qprob}^n}_{\wh{b}_n} \log Z^{\ol{\qprob}^n}_{\wh{b}_n}} \leq \wtexpv{}{}{\int_t^{\wh{b}_n} Q(u,X^{t,x}_u)du} \leq \wtexpv{}{}{\int_t^{T} Q(u,X^{t,x}_u)du}.
\end{equation*}
The upper bound for each $t\leq T, x\in\ol{\OO}_k$ and all $n\geq k+1$ now follows from Assumption \ref{A:complete_mkt} when $Q=Q_c$ and (using Jensen's inequality) from Assumption \ref{A:incomplete_mkt} when $Q=Q_i$. To show \eqref{E:loc_rel_ent_ident}, by conditioning on $\F^{W,H^n}_{\wh{b}_n}$ and using the conditional normality of $\int_t^{\wh{b}_n} (\ol{\mbf{B}}^n_u)'dB_u$ we obtain
\begin{equation*}
    \begin{split}
        &\expvs{Z^{\ol{\qprob}^n}_{\wh{b}_n} \log Z^{\ol{\qprob}^n}_{\wh{b}_n}} = \mathbb{E}\bigg[\ol{Z}^{\mbf{A}}_{\wh{b}_n}\ol{Z}^{\mbf{C},n}_{\wh{b}_n}\bigg(\log\left(\ol{Z}^{\mbf{A}}_{\wh{b}_n}\ol{Z}^{\mbf{C},n}_{\wh{b}_n}\right) + \frac{1}{2}\int_t^{\wh{b}_n}\abs{\ol{\mbf{B}}^n_u}^2du\bigg)\bigg].
    \end{split}
\end{equation*}
The stochastic exponential $\ol{Z}^{\mbf{C},n}$ satisfies
\begin{equation*}
    \begin{split}
        \ol{Z}^{\mbf{C},n}_{\wh{b}_n} &= 1_{\tau^n > \wh{b}_n} e^{-\int_t^{\wh{b}_n} (\ol{\mbf{C}}\chi_n\gamma)(v,X_v)dv} + 1_{\tau^n\leq \wh{b}_n} e^{-\int_t^{\tau^n} (\ol{\mbf{C}}\chi_n\gamma)(v,X_v)dv}\left(1+ \ol{\mbf{C}}(\tau^n,X_{\tau^n})\right),\\
        &= 1_{\tau^n > \wh{b}_n} e^{-\int_t^{\wh{b}_n} ((\wt{\gamma}-\gamma)\chi_n)(v,X_v)dv} + 1_{\tau^n\leq \wh{b}_n} e^{-\int_t^{\tau^n} ((\wt{\gamma}-\gamma)\chi_n)(v,X_v)dv}\left(\frac{\wt{\gamma}}{\gamma}\right)(\tau^n,X_{\tau^n}).
    \end{split}
\end{equation*}
Using \eqref{E:tau_n_def} we first obtain
\begin{equation}\label{E:jump_part_good}
    \begin{split}
        \condexpvs{\ol{Z}^{\mbf{C},n}_{\wh{b}_n}}{\F^W_{\wh{b}_n}} &= e^{-\int_t^{\wh{b}_n} ((\wt{\gamma}-\gamma)\chi_n)(v,X_v)dv - \int_t^{\wh{b}_n} (\chi_n\gamma)(v,X_v)dv}\\
        &\qquad + \int_t^{\wh{b}_n} e^{-\int_t^u ((\wt{\gamma}-\gamma)\chi_n)(v,X_v)dv}\left(\frac{\wt{\gamma}}{\gamma}\right)(u,X_{u})(\chi_n\gamma)(u,X_u) e^{-\int_t^u (\chi_n\gamma)(v,X_v)dv}du,\\
        &= e^{-\int_t^{\wh{b}_n} (\chi_n\wt{\gamma})(v,X_v)dv}  + \int_t^{\wh{b}_n} e^{-\int_t^v (\chi_n\wt{\gamma})(v,X_v)dv}(\chi_n\wt{\gamma})(u,X_{u})du = 1,
    \end{split}
\end{equation}
using integration by parts. This implies
\begin{equation*}
    \begin{split}
        &\expvs{Z^{\ol{\qprob}^n}_{\wh{b}_n} \log Z^{\ol{\qprob}^n}_{\wh{b}_n}} = \mathbb{E}\bigg[\ol{Z}^{\mbf{A}}_{\wh{b}_n}\bigg(\log\left(\ol{Z}^{\mbf{A}}_{\wh{b}_n}\right) + \frac{1}{2}\int_t^{\wh{b}_n}\abs{\ol{\mbf{B}}^n_u}^2du\bigg)\bigg]\\
        &\qquad + \mathbb{E}\bigg[\ol{Z}^{\mbf{A}}_{\wh{b}_n}\condexpvs{\ol{Z}^{\mbf{C},n}_{}\log\left(\ol{Z}^{\mbf{C},n}_{\wh{b}_n}\right)}{\F^{W}_{\wh{b}_n}}\bigg].
    \end{split}
\end{equation*}
Again using \eqref{E:tau_n_def} and similar calculations we obtain 
\begin{equation*}
    \begin{split}    &\condexpvs{\ol{Z}^{\mbf{C},n}_{\wh{b}_n}\log\left(\ol{Z}^{\mbf{C},n}_{\wh{b}_n}\right)}{\F^{W}_{\wh{b}_n}} = - e^{-\int_t^{\wh{b}_n} (\chi_n\wt{\gamma})(v,X_v)dv}\times \int_t^{\wh{b}_n} ((\wt{\gamma}-\gamma)\chi_n)(u,X_u)du\\
    &\qquad + \int_t^{\wh{b}_n} (\chi_n\wt{\gamma})(u,X_u)e^{-\int_t^v (\chi_n\wt{\gamma})(v,X_v)dv}\left(-\int_t^u ((\wt{\gamma}-\gamma)\chi_n)(u,X_u) + \log\left(\frac{\wt{\gamma}}{\gamma}\right)(u,X_u)\right)du,\\
    &\quad = \int_t^{\wh{b}_n} (\chi_n\wt{\gamma})(u,X_u)e^{-\int_t^v (\chi_n\wt{\gamma})(v,X_v)dv}\left(\frac{\gamma}{\wt{\gamma}} - \log\left(\frac{\gamma}{\wt{\gamma}}\right) - 1\right)(u,X_u)d u,
    \end{split}
\end{equation*}
where we again used integration by parts to obtain the last equality.  This gives \eqref{E:loc_rel_ent_ident} and finishes the proof of step (1).

Moving to step $(2)$, recall the Brownian motion $\wh{W}$ is absent in the non-localized market. As such, for each $\qprob \in \M$, the density process $Z^{\qprob}$ takes the form on $[t,T]$
\begin{equation*}
    Z^{\qprob}_{\cdot} = \Ecal\left(\int_t^{\cdot} \mbf{A}_u'dW_u + \mbf{B}_u' d B_u +  \mbf{C}_u d M_u \right),
\end{equation*}
where from \eqref{E:S_def} we deduce $\mbf{A},\mbf{B},\mbf{C}$ must satisfy the market price of risk equations
\begin{equation*}
    0 = (\mu - \gamma\ell)(s,X_s) + \sigma_W(s,X_s) \mbf{A}_s + \sigma_B(s,X_s) \mbf{B}_s - (\gamma\ell)(s,X_s) \mbf{C}_s,  
\end{equation*}
almost surely on the stochastic interval $[t,T\wedge\tau]$. From  \textit{Lemma A.8} we deduce that any optimizer $\wh{\qprob}\in\tM$ to \eqref{E:dual_problem} must have $\filt^{W,B}$ predictable $\mbf{A},\mbf{B},\mbf{C}$, and density process $Z^{\wh{\qprob}}$ stopped at $\tau$. As such, without loss of generality we will assume these facts throughout.

Next, similarly to \eqref{E:wh_an_bn_def} (and as in \textit{equation (A.37)}), define
\begin{equation*}
a_n = \tau^n \wedge T \wedge \zeta^{n - 1} = \tau^n \wedge \wh{a}_n,\qquad b_n = \tau^n \wedge T \wedge \zeta^{n} = \tau^n\wedge\wh{b}_n.
\end{equation*}
We also use \textit{Lemma C.8} which shows almost surely that $a_n = \tau^n \wedge T \wedge \zeta^{n-1}$ so that $a_n < b_n$. Next, recalling $\ol{\qprob}^n$ defined via either \eqref{E:bar_qprob_n_1} or \eqref{E:bar_qprob_n_2}, we modify $Z^{\qprob}$ to create a density process $Z^n$ by setting 
\begin{equation*}
    \begin{split}
        \mbf{A}^n_u &= \mbf{A}_u 1_{u \leq a_n} + \ol{\mbf{A}}_u 1_{a_n < u \leq b_n};\quad \mbf{B}^n_u = \mbf{B}_u 1_{u\leq a_n} + \ol{\mbf{B}}^n_u 1_{a_n < u \leq b_n};\quad \mbf{C}^n_u = \mbf{C}_u 1_{u \leq a_n} + \ol{\mbf{C}} 1_{a_n < u \leq b_n}.
    \end{split}
\end{equation*}
As $\chi_n = 1$ on $\OO_{n-1}$ it is clear that \eqref{E:risk_neutral_drift_condition} holds, and in what follows we show the associated measure $\qprob^n$ is well defined and in $\tM^n$. First, using the same arguments in the proof of \textit{Proposition A.11}, we obtain $\expvs{Z^n_{\wh{b}_n}} = 1$, which shows $\qprob^n$ is well defined an in $\M^n$. Next, we need to show $\expvs{Z^n_{\wh{b}_n} \log (Z^n_{\wh{b}_n})} < \infty$ so that $\qprob^n \in \tM^n$. To do so,  note that
\begin{equation*}
    \expvs{Z^n_{\wh{b}_n} \log (Z^n_{\wh{b}_n})} = \expvs{Z^n_{b_n} \log (Z^n_{b_n})} = \expvs{Z^{\qprob}_{a_n} \log Z^{\qprob}_{a_n}} + \expvs{Z^{\qprob}_{a_n} \condexpvs{\frac{\ol{Z}^n_{b_n}}{\ol{Z}^n_{a_n}} \log \left( \frac{\ol{Z}^n_{b_n}}{\ol{Z}^n_{a_n}} \right)}{\G^n_{a_n}} }.
\end{equation*}
As $s \to Z^{\qprob}_s\log(Z^{\qprob}_s)$ is a sub-martingale and $\qprob\in\tM$ we know
\begin{equation*}
    \expvs{Z^{\qprob}_{a_n} \log Z^{\qprob}_{a_n}} \leq \expvs{Z^{\qprob}_{T \wedge \tau} \log Z^{\qprob}_{T \wedge \tau}} < \infty.
\end{equation*}

For the second term, by Lemma \ref{L:long_cond_exp} with $Q=Q_c$ from \eqref{E:Qval_complete} or $Q=Q_i$ from \eqref{E:Qval_incomplete}, we have
\begin{equation}\label{E:cond_inequality_Q}
    \condexpvs{\frac{\ol{Z}^n_{b_n}}{\ol{Z}^n_{a_n}} \log \left( \frac{\ol{Z}^n_{b_n}}{\ol{Z}^n_{a_n}} \right)}{\G^n_{a_n}} \leq 1_{\tau^n>\wh{a}_n} \wtcondexpv{}{}{\int_{\wh{a}_n}^{\wh{b}_n} Q(u,X_u) du}{\F^W_{\wh{a}_n}}.
\end{equation}
Therefore, as $a_n = \wh{a}_n$ on $\cbra{\tau^n > \wh{a_n}}$ we have
\begin{equation*}
    \begin{split}
        &\expvs{Z^{\qprob}_{a_n} \condexpv{}{}{\frac{\ol{Z}^n_{b_n}}{\ol{Z}^n_{a_n}} \log \left( \frac{\ol{Z}^n_{b_n}}{\ol{Z}^n_{a_n}} \right)}{\G^n_{a_n}} } \leq  \expvs{Z^{\qprob}_{\wh{a}_n} 1_{\tau^n>\wh{a}_n}\wtcondexpv{}{}{\int_{\wh{a}_n}^{\wh{b}_n} Q(u,X_u) du}{\F^W_{\wh{a}_n}}}.
    \end{split}
\end{equation*}
Under Assumption \ref{A:complete_mkt}, the market is complete with unique martingale measure $\tprob$. Thus $Z^{\qprob}_{\wh{a}_n} = Z^{\tprob}_{\wh{a}_n}$ and
\begin{equation*}
    \begin{split}
        &\expvs{Z^{\qprob}_{a_n} \condexpv{}{}{\frac{\ol{Z}^n_{b_n}}{\ol{Z}^n_{a_n}} \log \left( \frac{\ol{Z}^n_{b_n}}{\ol{Z}^n_{a_n}} \right)}{\G^n_{a_n}} } \leq  \wtexpv{}{}{\int_{\wh{a}_n}^{\wh{b}_n} Q_c(u,X_u) du} \leq \wtexpv{}{}{\int_{t}^{T} Q_c(u,X^{t,x}) du}.
    \end{split}
\end{equation*}
The upper bound for each $t\leq T, x\in\ol{\OO}_k$ and all $n\geq k+1$ now follows form Assumption \ref{A:complete_mkt}. Alternatively, when Assumption \ref{A:incomplete_mkt} holds, we  use \eqref{E:inequalities_prob_result} to say for any $K>0$ that
\begin{equation*}
    \begin{split}
        \expvs{Z^{\qprob}_{a_n} \condexpv{}{}{\frac{\ol{Z}^n_{b_n}}{\ol{Z}^n_{a_n}} \log \left( \frac{\ol{Z}^n_{b_n}}{\ol{Z}^n_{a_n}} \right)}{\G^n_{a_n}} } &\leq \frac{1}{K}\left(\expvs{Z^{\qprob}_{\wh{a}_n} \log\left(Z^{\qprob}_{\wh{a}_n}\right)} + \frac{1}{e}\right)\\
        &\qquad + \frac{1}{K} \expvs{1_{\tau^n>\wh{a}_n} \e^{K \wtcondexpv{}{}{\int_{\wh{a}_n}^{\wh{b}_n} Q_i(u,X_u) d u}{\F^W_{\wh{a}_n}}} }.
    \end{split}
\end{equation*}
The first term on the right above is finite as $\qprob\in\tM$. As for the second term,  straight-forward computations (similar to those which established \eqref{E:jump_part_good}) show that on $\tau^n > \wh{a}^n$ we have $Z^{\tprob}_{\wh{a}_n} = \ol{Z}^{\mbf{A}}_{\wh{a}_n}$ for $\ol{Z}^{\mbf{A}}$ from \eqref{E:bZ_ABD_n}. Next, using Jensen's inequality, iterated conditioning and H\"{o}lder inequality with $p > 1$, we obtain
\begin{equation}\label{E:holder_inequality_Qi}
    \begin{split}
        \expvs{1_{\tau^n>\wh{a}_n} \e^{K \wtcondexpv{}{}{\int_{\wh{a}_n}^{\wh{b}_n} Q_i(u,X_u) d u}{\F^W_{\wh{a}_n}}} } & = \expv{}{}{1_{\tau^n>\wh{a}_n} Z^{\tprob}_{\wh{a}_n}\left(\ol{Z}^{\mbf{A}}_{\wh{a}_n}\right)^{-1} \e^{K \wtcondexpv{}{}{\int_{\wh{a}_n}^{\wh{b}_n} Q_i(u,X_u) d u}{\F^W_{\wh{a}_n}}} },\\
        &\leq \wtexpv{}{}{\left(\ol{Z}^{\mbf{A}}_{\wh{a}_n}\right)^{-1} \e^{ K \int_{\wh{a}_n}^{\wh{b}_n} Q_i(u,X_u) d u }},\\
        &\leq \wtexpv{}{}{\left(\ol{Z}^{\mbf{A}}_{\wh{a}_n}\right)^{-p}}^{\frac{1}{p}} \wtexpv{}{}{\e^{ \frac{K p}{p - 1} \int_{\wh{a}_n}^{\wh{b}_n} Q_i(u,X_u) d u } }^{\frac{p - 1}{p}}.
    \end{split}
\end{equation}
For the first term on the right above, we have
\begin{equation*}
    \wtexpv{}{}{\left(\ol{Z}^{\mbf{A}}_{\wh{a}_n}\right)^{-p}} = \expvs{\Ecal\Bigl(\int_t^{\cdot} (1 - p) \ol{\mbf{A}}_u 'dW_u \Bigr)_{\wh{a}_n}\e^{\frac{p(p - 1)}{2} \int_t^{\wh{a}_n} \abs{\tprobnu_u}^2 d u}}.
\end{equation*}
From part (ii) of Assumption \ref{A:incomplete_mkt}, we obtain by strong uniqueness (which holds under our assumptions) that the quantity on the right above is bounded from above by
\begin{equation*}
    \expvs{\e^{\frac{p(p - 1)}{2}\int_t^T \abs{\tprobnu(u,X^{(p),t,x})}^2 d u}},
\end{equation*}
which for $t\leq T, x\in \ol{\OO}_k$ is bounded by a constant $C(p, k)$. As for the second term in \eqref{E:holder_inequality_Qi}, we also have from Assumption \ref{A:incomplete_mkt} that for $K>0$ small enough 
\begin{equation*}
    \wtexpv{}{}{\e^{ \frac{K p}{p - 1} \int_{\wh{a}_n}^{\wh{b}_n} Q_i(u,X_u) d u } } \leq \wtexpv{}{}{\e^{\eps_2 \int_{t}^{T} Q_i(u,X^{t,x}) du } },
\end{equation*}
which for $t\leq T$ and $x\in\ol{\OO}_k$ is bounded by some constant $C(k)$.  This finishes showing $\qprob^n\in \tM^n$. 

We move to show that $\{ Z^n_{\wh{b}_n} \log (Z^n_{\wh{b}_n}) \}$ is uniformly integrable.  Once this is done, we can apply the same arguments in the proof of \textit{Proposition A.10}  to show that 
\begin{equation*}
\left\{ Z^n_{\wh{b}_n} \bigl(1_{\tau^n > \wh{b}_n} (\chi_n \phi)(X_{\wh{b}_n}) + 1_{\tau^n \leq \wh{b}_n} \psi(\tau^n, X_{\tau^n}) \bigr) \right\}
\end{equation*} 
is uniformly integrable. Given this uniform integrability and 
\begin{equation*} 
Z^n_{\wh{b}_n} \bigl( 1_{\tau^n > \wh{b}_n} (\chi_n \phi)(X_{\wh{b}_n}) + 1_{\tau^n \leq \wh{b}_n} \psi(\tau^n,X_{\tau^n}) \bigr) \cvas Z^{\qprob}_{T \wedge \tau} \bigl( 1_{\tau > T} \phi(X_{T}) + 1_{\tau \leq T} \psi(\tau,X_{\tau}) \bigr),
\end{equation*}
we obtain the limit
\begin{equation*}
    \lim_{n \to \infty} \expv{}{}{Z^n_{\wh{b}_n} \bigl(1_{\tau^n > \wh{b}_n} (\chi_n \phi)(X_{\wh{b}_n}) + 1_{\tau^n \leq \wh{b}_n} \psi(\tau^n,X_{\tau^n}) \bigr) } = \expv{}{}{Z^{\qprob}_{T \wedge \tau} \bigl( 1_{\tau > T} \phi(X_{T}) + 1_{\tau \leq T} \psi(\tau,X_{\tau}) \bigr)}.
\end{equation*}
To show the uniform integrability of $\{ Z^n_{\wh{b}_n} \log (Z^n_{\wh{b}_n}) \}$, note that as
\begin{equation*}
    Z^n_{\wh{b}_n} \log (Z^n_{\wh{b}_n}) = Z^{\qprob}_{a_n} \log(Z^{\qprob}_{a_n}) + Z^{\qprob}_{a_n} \condexpv{}{}{\frac{\ol{Z}^n_{b_n}}{\ol{Z}^n_{a_n}} \log \left( \frac{\ol{Z}^n_{b_n}}{\ol{Z}^n_{a_n}} \right)}{\G^n_{a_n}},
\end{equation*}
if we can show each of the families of random variables on the right above are uniformly integrable, then we have the uniform integrability of $\{ Z^n_{\wh{b}_n} \log (Z^n_{\wh{b}_n)} \}$. 

To obtain uniform integrability of $\{ Z^{\qprob}_{a_n} \log(Z^{\qprob}_{a_n})\}$, first we have $Z^{\qprob}_{a_n} \log (Z^{\qprob}_{a_n}) \cvas Z^{\qprob}_{T \wedge \tau} \log(Z^{\qprob}_{T \wedge \tau})$. Next, by non-negativity of $x \log x + \frac{1}{e}$ for $x \in (0, \infty)$, Jensen's inequality and Fatou's lemma imply
\begin{align*}
    \expvs{Z^{\qprob}_{a_n} \log(Z^{\qprob}_{a_n}) + \frac{1}{e}} & \leq \expvs{Z^{\qprob}_{T \wedge \tau} \log(Z^{\qprob}_{T \wedge \tau}) + \frac{1}{e}} < \infty, \\
    \lim_{n \to \infty} \expvs{Z^{\qprob}_{a_n} \log(Z^{\qprob}_{a_n}) + \frac{1}{e}} & = \expvs{Z^{\qprob}_{T \wedge \tau} \log(Z^{\qprob}_{T \wedge \tau}) + \frac{1}{e}}.
\end{align*} 
By Theorem 4.6.3 in \cite{MR3930614} we can conclude that $\{ Z^{\qprob}_{a_n} \log(Z^{\qprob}_{a_n}) + \frac{1}{e} \}$ is uniformly integrable, which further implies that $\{ Z^{\qprob}_{a_n} \log(Z^{\qprob}_{a_n}) \}$ is uniformly integrable. For the second family of random variables, by \eqref{E:inequalities_prob_result} and \eqref{E:cond_inequality_Q}, we have for $K > 0$,
\begin{equation*}
    Z^{\qprob}_{a_n} \condexpv{}{}{\frac{\ol{Z}^n_{b_n}}{\ol{Z}^n_{a_n}} \log \left( \frac{\ol{Z}^n_{b_n}}{\ol{Z}^n_{a_n}} \right)}{\G^n_{a_n}} \leq \frac{1}{K} Z^{\qprob}_{a_n} \log Z^{\qprob}_{a_n} + \frac{1}{K} \wtcondexpv{}{}{\e^{K \int_{\wh{a}_n}^{\wh{b}_n} Q(u, X_u) d u}}{\F^W_{\wh{a}_n}}. 
\end{equation*}
As we have already shown that $\{ Z^{\qprob}_{a_n} \log Z^{\qprob}_{a_n}\}$ is uniformly integrable, we just need to show that $\wtcondexpv{}{}{\e^{K \int_{\wh{a}_n}^{\wh{b}_n} Q(u, X_u) d u}}{\F^W_{\wh{a}_n}}$ is uniformly integrable. For $\wt{p} > 1$, and let $p > 1$ from part (ii) of Assumption \ref{A:incomplete_mkt},
\begin{equation*}
\begin{split}
& \expvs{\left( \wtcondexpv{}{}{\e^{K \int_{\wh{a}_n}^{\wh{b}_n} Q(u, X_u) d u}}{\F^W_{\wh{a}_n}} \right)^{\wt{p}}} 
\leq \expvs{\wtcondexpv{}{}{\e^{ \wt{p} K \int_{\wh{a}_n}^{\wh{b}_n} Q(u, X_u) d u}}{\F^W_{\wh{a}_n}}} \\
& \quad = \wtcondexpv{}{}{(Z^{\tprob}_{\wh{a}_n})^{-1} \e^{K \wt{p} \int_{\wh{a}_n}^{\wh{b}_n} Q(u, X_u) d u}}{\F^W_{\wh{a}_n}} \leq \left( \wtexpv{}{}{(Z^{\tprob}_{\wh{a}_n})^{-p}} \right)^{\frac{1}{p}} \left( \wtexpv{}{}{\e^{ \frac{K \wt{p} p}{p - 1} \int_{\wh{a}_n}^{\wh{b}_n} Q(u, X_u) d u } }\right)^{\frac{p - 1}{p}}.
\end{split}
\end{equation*}
We already showed that $\wtexpv{}{}{(Z^{\tprob}_{\wh{a}_n})^{-p}}$ is bounded, and pick $K > 0$ small enough such that $K \wt{p} p/(p - 1) < \eps_2$ in part (i) of Assumption \ref{A:incomplete_mkt}, we know that $\wtexpv{}{}{\e^{ \frac{K \wt{p} p}{p - 1} \int_{\wh{a}_n}^{\wh{b}_n} Q(u, X_u) d u } }^{\frac{p - 1}{p}}$ is also bounded. Hence, $\wtcondexpv{}{}{\e^{K \int_{\wh{a}_n}^{\wh{b}_n} Q(u, X_u) d u}}{\F^W_{\wh{a}_n}}$ is bounded in $L^{\wt{p}}$, which implies uniform integrability. Putting together, the uniform integrability of $\left\{Z_{a_n} \condexpv{}{}{\frac{\ol{Z}^n_{b_n}}{\ol{Z}^n_{a_n}} \log \left( \frac{\ol{Z}^n_{b_n}}{\ol{Z}^n_{a_n}} \right)}{\G^n_{a_n}} \right\}$ follows.

Finally, by applying the same arguments in the proof of \textit{Proposition A.11}, we can show that $Z_{a_n} \condexpv{}{}{\frac{\ol{Z}^n_{b_n}}{\ol{Z}^n_{a_n}} \log \left( \frac{\ol{Z}^n_{b_n}}{\ol{Z}^n_{a_n}} \right)}{\G^n_{a_n}} \cvprob 0$ as $\condexpv{}{}{\frac{\ol{Z}^n_{b_n}}{\ol{Z}^n_{a_n}} \log \left( \frac{\ol{Z}^n_{b_n}}{\ol{Z}^n_{a_n}} \right)}{\G^n_{a_n}} \cvprob 0$ and $Z_{a_n} \cvas Z_{T \wedge \tau}$. Along with uniform integrability, we obtain 
\begin{equation*}
    \lim_{n \to \infty} \expvs{Z_{a_n} \condexpv{}{}{\frac{\ol{Z}^n_{b_n}}{\ol{Z}^n_{a_n}} \log \left( \frac{\ol{Z}^n_{b_n}}{\ol{Z}^n_{a_n}} \right)}{\G^n_{a_n}} } = 0.
\end{equation*} 
\end{proof}


\section{Lemmas for Proposition \ref{P:localied_PDE_existence}}

Throughout, we continue to write all references to \cite{MR4086602} in italics.

\begin{lemma}\label{L:L_inv_bnd}
Under Assumption \ref{A:incomplete_mkt}, for any $\eps > 0$ when $\abs{\sigma_r} \geq \eps$, there exists constants $\ul{C}(\eps, n)$ and $\ol{C}(\eps, n)$ such that for all $(s,y) \in [0,T]\times \ol{\OO}_n$.
\begin{equation}\label{E:L_inv_bnd_res}
     \ul{C}(\eps, n) \leq \Sigma^{-1}_{i j} \leq  \ol{C}(\eps, n),\quad 1 \leq i,\ j \leq k + 1.
\end{equation}
\end{lemma}

\begin{proof}[Proof of Lemma \ref{L:L_inv_bnd}]
Given $\Sigma_e$ is invertible, and $\Afn(\sigma_r)$ is the Schur complement of $\Sigma_e$ in $\Sigma$, by the standard result in linear algebra, we obtain
\begin{equation*}
    \Sigma^{-1} =
    \begin{pmatrix}
        \Sigma_e^{-1} + \frac{1}{\Afn(\sigma_r)} \Sigma_e^{-1} \Upsilon_e \sigma_r \sigma_r' \Upsilon_e' \Sigma_e^{-1} & - \frac{1}{\Afn(\sigma_r)} \Sigma_e^{-1} \Upsilon_e \sigma_r \\
        - \frac{1}{\Afn(\sigma_r)} \sigma_r' \Upsilon_e'\Sigma_e^{-1} & \frac{1}{\Afn(\sigma_r)}
    \end{pmatrix}.
\end{equation*}
Every term in $\Sigma^{-1}$ except $\frac{1}{\Afn(\sigma_r)}$ is bounded on $[0, T] \times \ol{\OO}_n$. By Assumption \ref{A:region} and \ref{A:incomplete_mkt}, we know that $\bAfn$ is strictly positive definite on $[0, T] \times \ol{\OO}_n$. Hence, there exists $\ol{\lambda}_n>0$ and $\ol{\Lambda}_n > 0$ such that $\ol{\lambda}_n z' z \leq z' \bAfn(s, y) z \leq \ol{\Lambda}_n z' z$ holds for any $z \neq 0 \in \reals^d$, $(s, y) \in [0, T] \times \ol{\OO}_n$. This indicates that when $\abs{\sigma_r} \geq \eps$,
\begin{equation*}
     \frac{1}{\eps^2 \ol{\Lambda}_n} \leq \frac{1}{\Afn(\sigma_r)} \leq \frac{1}{\eps^2 \ol{\lambda}_n}.
\end{equation*}
Therefore, we can find $\ul{C}(\eps, n)$ and $\ol{C}(\eps, n)$ such that equation \eqref{E:L_inv_bnd_res} holds when $\abs{\sigma_r} \geq \eps$.
\end{proof}


\begin{lemma}\label{L:g_checka_bnd}
For $\ca^n$ defined in \eqref{E:checka_def}, there exists a constant $C(n)$ such that $g \ca^n(g, 0) \leq C(n)(1 + g^2)$ for all $(s, x) \in [0, T] \times \OO_n$.
\end{lemma}

\begin{proof}[Proof of Lemma \ref{L:g_checka_bnd}]
We will separate our proof based upon if Assumption \ref{A:complete_mkt} or Assumption \ref{A:incomplete_mkt} holds, and unless explicitly stated otherwise, omit the functional dependence on $(s,y)$. When Assumption \ref{A:complete_mkt} holds, from \eqref{E:checka_def} and \eqref{E:Hcal_is_ok_1} we obtain  $\ca^n(g, 0) = \chi_n \left(\gamma/\alpha + \wh{\Hcal}(g) \right)$ where 
\begin{equation*}
    \wh{\Hcal}(g) = -\frac{\gamma}{\alpha}-\wt{\gamma}g + F.
\end{equation*}
This clearly implies
\begin{equation*}
    \sup_{(s, y) \in [0, T] \times \ol{\OO}_n} \frac{g \ca^n(g, 0)}{1 + g^2} < \infty.
\end{equation*}
Similarly, under Assumption \ref{A:incomplete_mkt} from \eqref{E:Hcal_is_ok_2} we find
\begin{equation*}
    \ca^n(g, 0) = \chi_n \left( \frac{\gamma}{\alpha} + \wh{\Hcal}(g) + R_{\Hcal}(\sigma_r, g, 0) \right).
\end{equation*}
First consider when $\abs{\sigma}_r \equiv 0$. As $R_{\Hcal} = 0$ we similarly obtain 
\begin{equation*}
    \sup_{\substack{(s, y) \in [0, T] \times \ol{\OO}_n, \\ \abs{\sigma_r} \equiv 0}} \frac{g \ca^n(g, 0)}{1 + g^2}  < \infty.
\end{equation*}
Next, consider when $\abs{\sigma_r}(s,y) > 0$ for each fixed $(s,y)$.  Using the regularity of $(s,y)\to \sigma_r(s,y)$ we deduce the existence of an $\eps(n)>0$ such that $\abs{\sigma_r} \geq \eps(n)$ for $(s,y)\in [0,T]\times \ol{\OO_n}$. From Lemma \ref{L:L_inv_bnd} we see the matrix $\Sigma$ in \eqref{E:SigE_UpsE} is uniformly elliptic in $[0,T]\times \ol{\OO}_n$ (with ellipticity constant depending on $n$) and hence we can invoke \textit{Lemma C.2} to obtain
\begin{equation*}
    \sup_{\substack{(s, y) \in [0, T] \times \ol{\OO}_n, \\ \abs{\sigma_r} \geq \eps(n)}} \frac{g \ca^n(g, 0)}{1 + g^2} < \infty.
\end{equation*}
This finishes the result.
\end{proof}


\begin{lemma}\label{L:p_checka_bnd}
For $\ca^n$ defined in \eqref{E:checka_def}, and any interval $[g_1, g_2]$,
\begin{equation*}
    \limsup_{\abs{p} \uparrow \infty} \sup_{\substack{(s, y) \in [0, T] \times \ol{\OO}_n, \\ g \in [g_1, g_2]}} \frac{\abs{\ca^n(g, p)}}{\abs{p}^2} < \infty.
\end{equation*}
\end{lemma}

\begin{proof}[Proof of Lemma \ref{L:p_checka_bnd}]
We will separate our proof based on Assumptions \ref{A:complete_mkt} and \ref{A:incomplete_mkt}, and omit the functional dependence on $(s,y)$. First, under Assumption \eqref{A:complete_mkt}, using \eqref{E:checka_def} and \eqref{E:Hcal_is_ok_1} we see that $\ca^n(g, p)$ is a quadratic function with respect to $p$, with quadratic term $-\frac{\alpha}{2} p' (A - \chi_n \Upsilon_e' \Sigma_e^{-1} \Upsilon_e) p$. Since $\rho = \idmat{d}$ we have $A - \chi_n \Upsilon_e' \Sigma_e^{-1} \Upsilon_e = (1 - \chi_n) A$ and hence
\begin{equation*}
    \limsup_{\abs{p} \uparrow \infty} \sup_{\substack{(s, y) \in [0, T] \times \ol{\OO}_n, \\ g \in [g_1, g_2]}} \frac{\abs{\ca^n(g, p)}}{\abs{p}^2} < \infty.
\end{equation*}
Under Assumption \ref{A:incomplete_mkt} from \eqref{E:Hcal_is_ok_2} we see
\begin{equation*}
    \ca^n(g, p) = b' p - \frac{\alpha}{2} p' A p + \chi_n \left( - \wt{\gamma}g + F + \frac{\alpha}{2}p'\Upsilon_e'\Sigma_e^{-1}\Upsilon_e p - p'\Upsilon_e'\Sigma_e^{-1}(\mu_e-\wt{\gamma}\ell_e) + R_{\Hcal}(\sigma_r,g,p) \right).
\end{equation*}
When $\abs{\sigma_r} \equiv 0$ we  know $R_{\Hcal} = 0$ and hence $\ca^n(g, p)$ is a quadratic function with respect to $p$. Therefore 
\begin{equation*}
    \limsup_{\abs{p} \uparrow \infty} \sup_{\substack{(s, y) \in [0, T] \times \ol{\OO}_n, \\ g \in [g_1, g_2],\ \abs{\sigma_r} \equiv 0}} \frac{\abs{\ca^n(g, p)}}{\abs{p}^2}  < \infty.
\end{equation*}
Alternatively, when $\abs{\sigma}_r(s,y) > 0$ for all $(s,y)$ by the regularity of $(s,y)\to\sigma_r(s,y)$ we know there is an $\eps(n)$ so that $\abs{\sigma_r}\geq \eps(n)$ on $[0,T]\times \ol{\OO}_n$. Therefore, using Lemma \ref{L:L_inv_bnd} we may invoke \textit{Lemma C.3} to obtain 
\begin{equation*}
    C(\eps, n) \dfn  \limsup_{\abs{p} \uparrow \infty} \sup_{\substack{(s, y) \in [0, T] \times \ol{\OO}_n, \\ g \in [g_1, g_2], \abs{\sigma_r} \geq \eps(n)}} \frac{\abs{\ca^n(g, p)}}{\abs{p}^2} < \infty,
\end{equation*}
finishing the result.
\end{proof}



\section{Lemmas for Proposition \ref{P:analytical_unwind}} Throughout, we continue to write all references to \cite{MR4086602} in italics.

\begin{lemma}\label{L:BCD_k_infty_finite_1}
Under Assumption \ref{A:complete_mkt}, for $B^{\infty}_k$, $C^{\infty}_k$, and $D^{\infty}_k$ defined in \eqref{E:ABCD_def}, \eqref{E:ABCD_k_infty}, we have
\begin{equation*}
    B^{\infty}_k < \infty, \quad C^{\infty}_k = 0, \quad D^{\infty}_k < \infty.
\end{equation*}
\end{lemma}

\begin{proof}[Proof of Lemma \ref{L:BCD_k_infty_finite_1}] Recall that for $n\geq k+1$ we have $\chi_n = 1$ on $\ol{\OO}_k$ and hence we write $\check{a}$ for $\check{a}^n$ in \eqref{E:checka_def}. Next, under Assumption \ref{A:complete_mkt} we know from  \eqref{E:Hcal_is_ok_1} that 
\begin{equation}\label{E:check_a_comp_0}
    \ca(g, p) = \bigl( b -\Upsilon_e'\Sigma_e^{-1}(\mu_e-\wt{\gamma}\ell_e)\bigr)'p  - \wt{\gamma}g + F.
\end{equation}
(the quadratic term $-(\alpha/2)p' (A - \Upsilon_e' \Sigma_e^{-1} \Upsilon_e) p$ vanishes as $\rho = \idmat{d}$). Therefore, 
\begin{equation*}
    \begin{split}
    \nabla_p \Ecal &= Ap;\quad (\ol{\delta}(p)-1)[\Ecal] = \Ecal;\quad \delta(p)[\Ecal] = \frac{1}{2\abs{p}^2} p'\bigg(\sum_{i,j=1}^d p_i p_j \nabla_x A^{ij}\bigg),\\
    \nabla_p \ca(g, p) &= b - \Upsilon_e' \Sigma_e^{-1} \left(\mu_e - \wt{\gamma}\ell_e\right);\quad \bigl( \ol{\delta}(p) - 1 \bigr) [\ca](g, p) =  \wt{\gamma}g - F,\\
    \delta(p) [\ca](g, p) &= -\wt{\gamma} + \frac{1}{\abs{p}^2} p'\bigg(\sum_{i = 1}^{d} p^i \nabla_x \bigl( b -\Upsilon_e'\Sigma_e^{-1}(\mu_e-\wt{\gamma}\ell_e)\bigr)^i -g\nabla_x \wt{\gamma} + \nabla_x F\bigg).
    \end{split}
\end{equation*}
As $\delta(p) [\Ecal]$ is on the order of $\abs{p}$ and $\bigl( \ol{\delta}(p) - 1 \bigr) [\ca]$ is on the order of $1$, $B_k^\infty < \infty$ (in fact $B^k_{\infty} = 0$). As $\delta(p)[A^{ij}]$ is on the order of $1/\abs{p}$ and $\delta(p) [\ca]$ is on the order of $1$, $C_k^\infty = 0$. Finally, as $\Lambda_k p' p$ is on the order of $\abs{p}^2$ and $\abs{p} \left( \abs{\nabla_p \Ecal} + \abs{\nabla_p \ca} \right)$ is on the order of $\abs{p}^2$, $D_k^\infty < \infty$. This finishes the result.
\end{proof}


\begin{lemma}\label{L:BD_k_infty_finite_2}
Under Assumption \ref{A:incomplete_mkt}, for $B^{\infty}_k$ and $D^{\infty}_k$ defined in \eqref{E:ABCD_k_infty}, we have
\begin{equation*}
    B^{\infty}_k < \infty,\quad C^{\infty}_k = 0,\quad D^{\infty}_k < \infty.
\end{equation*}
\end{lemma}

\begin{proof}[Proof of Lemma \ref{L:BD_k_infty_finite_2}]
We start by considering when  $\abs{\sigma_r} = 0$.  Here, from \eqref{E:Hcal_is_ok_2} we find
\begin{equation*}
    \check{a} = \left(b-\Upsilon_e'\Sigma_e^{-1}(\mu_e - \wt{\gamma}\ell_e)\right)'p - \frac{\alpha}{2}p'\left(A-\Upsilon_e'\Sigma_e^{-1}\Upsilon_e\right)p  -\wt{\gamma} g + F.
\end{equation*}
This is very similar to \eqref{E:check_a_comp_0} and direct computation shows
\begin{equation*}
    \begin{split}
    \nabla_p \ca(g, p) &= b - \Upsilon_e' \Sigma_e^{-1} \left(\mu_e - \wt{\gamma}\ell_e\right) - \alpha \left(A-\Upsilon_e'\Sigma_e^{-1}\Upsilon_e\right)p,\\
    \bigl( \ol{\delta}(p) - 1 \bigr) [\ca](g, p) &=  -\frac{\alpha}{2}p'\left(A-\Upsilon_e'\Sigma_e^{-1}\Upsilon_e\right)p + \wt{\gamma}g - F,\\
    \delta(p) [\ca](g, p) &= -\wt{\gamma} + \frac{1}{\abs{p}^2} p'\bigg(\sum_{i = 1}^{d} p^i \nabla_x \bigl( b -\Upsilon_e'\Sigma_e^{-1}(\mu_e-\wt{\gamma}\ell_e)\bigr)^i \\
    &\qquad\qquad - \frac{\alpha}{2}\sum_{i,j=1}^d p^ip^j \nabla_x \left(A-\Upsilon_e'\Sigma_e^{-1}\Upsilon_e\right)^{ij} -g\nabla_x \wt{\gamma} + \nabla_x F\bigg).
    \end{split}
\end{equation*}
Just like in Lemma \ref{L:BCD_k_infty_finite_1}, as $\delta(p) [\Ecal]$ is on the order of $\abs{p}$ and $\bigl( \ol{\delta}(p) - 1 \bigr) [\ca]$ is on the order of $\abs{p}^2$, $B_k^{\infty} < \infty$. As $\delta(p)[A^{ij}]$ is on the order of $1/\abs{p}$ and $\delta(p) [\ca]$ is on the order of $\abs{p}$, $C_k^{\infty} = 0$. Finally, as $\Lambda_k p' p$ is on the order of $\abs{p}^2$ and $\abs{p} \left( \abs{\nabla_p \Ecal} + \abs{\nabla_p \ca} \right)$ is on the order of $\abs{p}^2$, $D_k^{\infty} < \infty$. This finishes the result restricting to $\abs{\sigma_r} = 0$. 

Next, when $\abs{\sigma}_r(s,y) > 0$ for all $(s,y)$, just like above, the regularity of $(s,y)\to\sigma_r(s,y)$ yields an $\eps(k)$ so that $\abs{\sigma_r}\geq \eps(k)$ on $[0,T]\times \ol{\OO}_k$. Therefore, using Lemma \ref{L:L_inv_bnd} we may invoke \textit{Lemma C.5} to obtain $B^{\infty}_k <\infty$, $C^{\infty}_k = 0$ and $D^{\infty}_k < \infty$.  This finishes the result.

\end{proof}



\section{Lemma for Proposition \ref{P:unwind_prob}}\label{AS:PropE9Lem}

\begin{lemma}\label{L:long_cond_exp}
\begin{equation*}
    \begin{split}
        &\condexpvs{\frac{\ol{Z}^n_{b_n}}{\ol{Z}^n_{a_n}} \log \left( \frac{\ol{Z}^n_{b_n}}{\ol{Z}^n_{a_n}} \right)}{\G^n_{a_n}} = 1_{\tau^n > \wh{a}_n} \wt{\mathbb{E}}\bigg[\int_{\wh{a}_n}^{\wh{b}_n} e^{-\int_{\wh{a}_n}^u (\chi_n\wt{\gamma})(v,X_v)dv}\bigg(\frac{1}{2}\abs{\mbf{\ol{A}}_u}^2\\
        &\qquad\qquad + \frac{1}{2}\abs{\ol{\mbf{B}}^{n}_u}^2 + \bigg(\chi_n\wt{\gamma}\bigg(\frac{\gamma}{\wt{\gamma}} - 1 - \log\bigg(\frac{\gamma}{\wt{\gamma}}\bigg)\bigg)\bigg)(u,X_u)\bigg)du\ \bigg| \ \F^W_{\wh{a}_n}\bigg].
    \end{split}
\end{equation*}
\end{lemma}

\begin{proof}
Throughout this proof, we write $Z_{a,b} = Z_b/Z_a$ for any strictly positive process $Z$. First, on the set $\{ \tau^n \leq \wh{a}_n\}$, $a_n = b_n$ and $\ol{Z}^n_{a_n} = \ol{Z}^n_{b_n}$, which implies $\ol{Z}^n_{a_n,b_n} = \ol{Z}_{a_n,a_n} = 1$ and hence 
\begin{equation*}
    \condexpvs{ \ol{Z}^n_{a_n,b_n}\log \left( \ol{Z}^n_{a_n,b_n}\right)}{\G^n_{a_n}} = 0,\quad \left(\textrm{ on } \{ \tau^n \leq \wh{a}_n\}\right).
\end{equation*}
Therefore, from now on, we restrict to $\{ \tau^n > \wh{a}_n\}$ where  $a_n = \tau^n \wedge \wh{a}_n = \wh{a}_n$ and hence $\ol{Z}^n_{a_n,b_n} = \ol{Z}^n_{\wh{a}_n,b_n}$. Using \eqref{E:bZ_ABD_n} we have
\begin{equation*}
    \ol{Z}^n_{\wh{a}_n,b_n} = \ol{Z}^{\mbf{A}}_{\wh{a}_n,b_n} \times \ol{Z}^{\mbf{B},n}_{\wh{a}_n,b_n}\times \ol{Z}^{\mbf{C},n}_{\wh{a}_n,b_n}.
\end{equation*}
From \cite[Chapter 1]{MR1121940}, we may replace $\G^n_{a_n}$ with $\G^n_{\wh{a}_n}$ as we are restricting to $\cbra{\tau^n > \wh{a}_n}$.  Lastly, we write $Y^{a}$ for any process $Y$ stopped at any random time $a$.  Given all of this, let $A_{\wh{a}_n} \in \F^{W,B,\wh{W}}_{\wh{a}_n}$ and let $\varphi$ any function of the path $H^n$.  First, we have
\begin{equation*}
    \begin{split}
        &\expvs{\ol{Z}^n_{\wh{a}_n,b_n}\log\left(\ol{Z}^n_{\wh{a}_n,b_n}\right)1_{\tau^n > \wh{a}_n} 1_{A_{\wh{a}_n}} \varphi((H^n)^{\wh{a}_n})}\\
        &\qquad\qquad = \expvs{\ol{Z}^{\mbf{A}}_{\wh{a}_n,b_n} \ol{Z}^{\mbf{B},n}_{\wh{a}_n,b_n} \ol{Z}^{\mbf{C},n}_{\wh{a}_n,b_n}\log\left(\ol{Z}^{\mbf{A}}_{\wh{a}_n,b_n}  \ol{Z}^{\mbf{B},n}_{\wh{a}_n,b_n} \ol{Z}^{\mbf{C},n}_{\wh{a}_n,b_n}\right) 1_{A_{\wh{a}_n}} \varphi((0)^{\wh{a}_n})\left(1_{\tau^n > \wh{b}_n} + 1_{\wh{b}_n\geq \tau^n > \wh{a}_n}\right)}.
    \end{split}
\end{equation*}
Let us focus on the term with $\tau^n > \wh{b}^n$ (so that $b_n = \wh{b}_n$), labeling it $\mbf{Q}_1$.  Here, (omitting function arguments so e.g.$f(v,X_v)$ is written $f_v$ and defining $\Hcal_n = \F^{W,B,\wh{W}}_{\wh{a}_n} \vee \F^{W}_{\wh{b}_n}\vee \sigma(\tau^n)$)
\begin{equation*}
    \begin{split}
        \mbf{Q}_1 &= \E\Bigg[1_{A_{\wh{a}_n}} \varphi((0)^{\wh{a}_n}) 1_{\tau^n > \wh{b}_n} \ol{Z}^{\mbf{A}}_{\wh{a}_n,\wh{b}_n} e^{-\int_{\wh{a}_n}^{\wh{b}_n} (\chi_n(\wt{\gamma}-\gamma))_vdv}\bigg(\bigg(\log\big(\ol{Z}^{\mbf{A}}_{\wh{a}_n,\wh{b}_n}\big) - \int_{\wh{a}_n}^{\wh{b}_n} (\chi_n(\wt{\gamma}-\gamma))_u du\bigg)\\
        &\qquad \times \condexpvs{\ol{Z}^{\mbf{B},n}_{\wh{a}_n,b_n}}{\Hcal_n} + \condexpvs{\ol{Z}^{\mbf{B},n}_{\wh{a}_n,b_n}\log\big(\ol{Z}^{\mbf{B},n}_{\wh{a}_n,b_n}\big)}{\Hcal_n}\bigg)\Bigg],\\
        &= \E\Bigg[1_{A_{\wh{a}_n}} \varphi((0)^{\wh{a}_n}) 1_{\tau^n > \wh{b}_n} \ol{Z}^{\mbf{A}}_{\wh{a}_n,\wh{b}_n} e^{-\int_{\wh{a}_n}^{\wh{b}_n} (\chi_n(\wt{\gamma}-\gamma))_v dv}\bigg(\int_{\wh{a}_n}^{\wh{b}_n} \mbf{\ol{A}}_u'd\wt{W}_u + \frac{1}{2}\int_{\wh{a}_n}^{\wh{b}_n} \abs{\mbf{\ol{A}}_u}^2 du\\
        &\qquad - \int_{\wh{a}_n}^{\wh{b}_n} (\chi_n(\wt{\gamma}-\gamma))_u du  + \frac{1}{2}\int_{\wh{a}_n}^{\wh{b}_n} \abs{\mbf{\ol{B}}^n_u}^2 du\bigg)\Bigg],\\
        &= \E\Bigg[1_{A_{\wh{a}_n}} \varphi((0)^{\wh{a}_n}) e^{-\int_t^{\wh{a}_n} (\chi_n\gamma)_v dv} \E\bigg[\ol{Z}^{\mbf{A}}_{\wh{a}_n,\wh{b}_n} e^{-\int_{\wh{a}_n}^{\wh{b}_n} (\chi_n\wt{\gamma})_v dv}\bigg(\int_{\wh{a}_n}^{\wh{b}_n} \mbf{\ol{A}}_u'd\wt{W}_u + \frac{1}{2}\int_{\wh{a}_n}^{\wh{b}_n} \abs{\mbf{\ol{A}}_u}^2 du\\
        &\qquad - \int_{\wh{a}_n}^{\wh{b}_n} (\chi_n(\wt{\gamma}-\gamma))_u du  + \frac{1}{2}\int_{\wh{a}_n}^{\wh{b}_n} \abs{\mbf{\ol{B}}^n_u}^2 du\bigg)\ \bigg| \ \F^W_{\wh{a}_n} \bigg]\Bigg],\\
    \end{split}
\end{equation*}
Above, to obtain the third equality we first conditioned on $\F^W_{\wh{b}_n}$ and used the intensity function of $\tau^n$ under $\prob$, and then we conditioned on $\F^W_{\wh{a}_n}$.  Next, we focus on the term with $\wh{a}_n \leq \tau^n < \wh{b}_n$, labeling it $\mbf{Q}_2$.
\begin{equation*}
    \begin{split}
        \mbf{Q}_2 &= \E\Bigg[1_{A_{\wh{a}_n}} \varphi((0)^{\wh{a}_n}) 1_{\wh{a}_n\leq \tau^n < \wh{b}_n} \ol{Z}^{\mbf{A}}_{\wh{a}_n,\tau^n} e^{-\int_{\wh{a}_n}^{\tau^n} (\chi_n(\wt{\gamma}-\gamma))_vdv}\left(\frac{\wt{\gamma}}{\gamma}\right)_{\tau^n}\bigg(\bigg(\log\big(\ol{Z}^{\mbf{A}}_{\wh{a}_n,\tau^n}\big)\\
        &\qquad - \int_{\wh{a}_n}^{\tau^n} (\chi_n(\wt{\gamma}-\gamma))_u du + \left(\frac{\wt{\gamma}}{\gamma}\right)_{\tau^n}\bigg)\condexpvs{\ol{Z}^{\mbf{B},n}_{\wh{a}_n,b_n}}{\Hcal_n} + \condexpvs{\ol{Z}^{\mbf{B},n}_{\wh{a}_n,b_n}\log\big(\ol{Z}^{\mbf{B},n}_{\wh{a}_n,b_n}\big)}{\Hcal_n}\bigg)\Bigg],\\
        &= \E\Bigg[1_{A_{\wh{a}_n}} \varphi((0)^{\wh{a}_n}) 1_{\wh{a}_n\leq \tau^n < \wh{b}_n} \ol{Z}^{\mbf{A}}_{\wh{a}_n,\tau^n} e^{-\int_{\wh{a}_n}^{\tau^n} (\chi_n(\wt{\gamma}-\gamma))_vdv}\left(\frac{\wt{\gamma}}{\gamma}\right)_{\tau^n}\bigg(\int_{\wh{a}_n}^{\tau^n} \mbf{\ol{A}}_u'd\wt{W}_u\\
        &\qquad + \frac{1}{2}\int_{\wh{a}_n}^{\tau^n} \abs{\mbf{\ol{A}}_u}^2 du - \int_{\wh{a}_n}^{\tau^n} (\chi_n(\wt{\gamma}-\gamma))_u du + \left(\frac{\wt{\gamma}}{\gamma}\right)_{\tau^n} + \frac{1}{2}\int_{\wh{a}_n}^{\tau^n} \abs{\mbf{\ol{B}}^n_u}^2 du\bigg)\Bigg],\\
        &= \E\Bigg[1_{A_{\wh{a}_n}} \varphi((0)^{\wh{a}_n})e^{-\int_t^{\wh{a}_n}(\chi_n\gamma)_v dv}\E\bigg[\bigg( \int_{\wh{a}_n}^{\wh{b}_n} \ol{Z}^{\mbf{A}}_{\wh{a}_n,u} (\chi_n\wt{\gamma})_u e^{-\int_{\wh{a}_n}^{u} (\chi_n\wt{\gamma})_vdv}\bigg(\int_{\wh{a}_n}^{u} \mbf{\ol{A}}_v'd\wt{W}_v\\
        &\qquad + \frac{1}{2}\int_{\wh{a}_n}^{u} \abs{\mbf{\ol{A}}_v}^2 dv  - \int_{\wh{a}_n}^{u} (\chi_n(\wt{\gamma}-\gamma))_v dv + \left(\frac{\wt{\gamma}}{\gamma}\right)_{u} + \frac{1}{2}\int_{\wh{a}_n}^{u} \abs{\mbf{\ol{B}}^n_v}^2 dv\bigg)du\bigg)\ \bigg| \ \F^W_{\wh{a}_n}\bigg]\Bigg].
    \end{split}
\end{equation*}
For a generic $\F^W_{\wh{a}_n}$ measurable random variable $Y_n$ one has
\begin{equation*}
    \E\Bigg[Y_n 1_{\tau^n>\wh{a}_n}1_{A_{\wh{a}_n}}\varphi((H^n)^{\wh{a}_n}\Bigg] = \E\Bigg[1_{A_{\wh{a}_n}}\varphi((0)^{\wh{a}_n})e^{-\int_t^{\wh{a}_n}(\chi_n\gamma)_v dv} Y_n\Bigg].
\end{equation*}
This shows that
\begin{equation*}
    \condexpvs{\ol{Z}^n_{\wh{a}_n,b_n}\log\left(\ol{Z}^n_{\wh{a}_n,b_n}\right)1_{\tau^n > \wh{a}_n}}{\G^n_{\wh{a}_n}} = 1_{\tau^n>\wh{a}_n} \times \mbf{Q}_3,
\end{equation*}
where
\begin{equation*}
    \begin{split}
         \mbf{Q}_3 &=  \E\bigg[\ol{Z}^{\mbf{A}}_{\wh{a}_n,\wh{b}_n} e^{-\int_{\wh{a}_n}^{\wh{b}_n} (\chi_n\wt{\gamma})_v dv}\bigg(\int_{\wh{a}_n}^{\wh{b}_n} \mbf{\ol{A}}_u'd\wt{W}_u + \frac{1}{2}\int_{\wh{a}_n}^{\wh{b}_n} \abs{\mbf{\ol{A}}_u}^2 du \\
        &\quad - \int_{\wh{a}_n}^{\wh{b}_n} (\chi_n(\wt{\gamma}-\gamma))_u du  + \frac{1}{2}\int_{\wh{a}_n}^{\wh{b}_n} \abs{\mbf{\ol{B}}^n_u}^2 du\bigg) + \int_{\wh{a}_n}^{\wh{b}_n} \ol{Z}^{\mbf{A}}_{\wh{a}_n,u} (\chi_n\wt{\gamma})_u e^{-\int_{\wh{a}_n}^{u} (\chi_n\wt{\gamma})_vdv}\\
        &\quad \times \bigg(\int_{\wh{a}_n}^{u} \mbf{\ol{A}}_v'd\wt{W}_v + \frac{1}{2}\int_{\wh{a}_n}^{u} \abs{\mbf{\ol{A}}_v}^2 dv  - \int_{\wh{a}_n}^{u} (\chi_n(\wt{\gamma}-\gamma))_v dv + \left(\frac{\wt{\gamma}}{\gamma}\right)_{u} + \frac{1}{2}\int_{\wh{a}_n}^{u} \abs{\mbf{\ol{B}}^n_v}^2 dv\bigg)du \ \bigg| \ \F^W_{\wh{a}_n}\bigg]. 
    \end{split}
\end{equation*}
For any $\filt^{W}$ adapted process $R$ one can show
\begin{equation*}
    \condexpvs{\int_{\wh{a}_n}^{\wh{b}_n} \ol{Z}^{\mbf{A}}_{\wh{a}_n,u} R_u du}{\F^W_{\wh{a}_n}} = \wtcondexpv{}{}{\int_{\wh{a}_n}^{\wh{b}_n} R_u du}{\F^W_{\wh{a}_n}}.
\end{equation*}
So, with 
\begin{equation*}
    \begin{split}
        R_u &= \int_{\wh{a}_n}^u \mbf{\ol{A}}_v'd\wt{W}_v + \frac{1}{2}\int_{\wh{a}_n}^u \abs{\mbf{\ol{A}}_v}^2 dv - \int_{\wh{a}_n}^u (\chi_n(\wt{\gamma}-\gamma))_v dv + \frac{1}{2}\int_{\wh{a}_n}^u \abs{\mbf{\ol{B}}^n_v}^2 dv,\\
        D_u &= e^{-\int_{\wh{a}_n}^u (\chi_n\wt{\gamma})_v dv},
    \end{split}
\end{equation*}
we obtain
\begin{equation*}
    \begin{split}
        \mbf{Q}_3 &= \wt{\E}\Bigg[ D_{\wh{b}_n} R_{\wh{b}_n} - \int_{\wh{a}_n}^{\wh{b}_n} R_u dD_u + \int_{\wh{a}_n}^{\wh{b}_n} (\chi_n\wt{\gamma})_u e^{-\int_{\wh{a}_n}^u (\chi_n\wt{\gamma})_v dv}\log\left(\frac{\wt{\gamma}}{\gamma}\right)_u du \ \bigg| \ \F^W_{\wh{a}_n} \Bigg],\\
        &= \wt{\E}\Bigg[ \int_{\wh{a}_n}^{\wh{b}_n} D_u dR_u + \int_{\wh{a}_n}^{\wh{b}_n} (\chi_n\wt{\gamma})_u e^{-\int_{\wh{a}_n}^u (\chi_n\wt{\gamma})_v dv}\log\left(\frac{\wt{\gamma}}{\gamma}\right)_u du \ \bigg| \ \F^W_{\wh{a}_n} \Bigg],\\
        &= \wt{\E}\Bigg[ \int_{\wh{a}_n}^{\wh{b}_n} D_u \bigg( \bar{A}_u'd\wt{W}_u + \bigg((\chi_n\wt{\gamma})_u \log\left(\frac{\wt{\gamma}}{\gamma}\right)_u   + \frac{1}{2}\abs{\mbf{\ol{A}}_u}^2  -  (\chi_n(\wt{\gamma}-\gamma))_u  + \frac{1}{2}\abs{\mbf{\ol{B}}^n_u}^2\bigg) du\bigg) \ \bigg| \ \F^W_{\wh{a}_n} \Bigg],\\
        &= \wt{\E}\Bigg[ \int_{\wh{a}_n}^{\wh{b}_n} D_u \bigg(\frac{1}{2}\abs{\mbf{\ol{A}}_u}^2 + \frac{1}{2}\abs{\mbf{\ol{B}}^n_u}^2 + (\chi_n\wt{\gamma})_u\bigg(\frac{\gamma}{\gamma} - 1 - \log\left(\frac{\gamma}{\wt{\gamma}}\right)\bigg)_u \bigg)du \ \bigg| \ \F^W_{\wh{a}_n} \Bigg].
    \end{split}
\end{equation*}
This gives the result.

\end{proof}


\end{document}